\documentclass[%
  reprint,
  amsmath,amssymb,
  aps,
  superscriptaddress,
  longbibliography,
]{revtex4-2}

\usepackage{graphicx}
\usepackage{dcolumn}
\usepackage{bm}

\usepackage{amsthm, mathtools, amsfonts}
\usepackage{braket}
\usepackage{newtxtext,newtxmath}
\mathtoolsset{showonlyrefs}

\usepackage[table,dvipsnames]{xcolor}
\usepackage{tcolorbox}

\usepackage[colorlinks=true,citecolor=teal,linkcolor=teal,urlcolor=teal]{hyperref}

\usepackage{algorithm}
\usepackage{algpseudocode}

\usepackage{tikz-cd}
\usetikzlibrary{backgrounds}

\usepackage{indentfirst}
\usepackage{enumitem}
\usepackage{ulem}
\usepackage{makecell}
\usepackage{booktabs}
\usepackage{pdflscape}
\makeatletter
\let\cc@landscaperotate\LS@rot
\let\LS@rot\@empty
\g@addto@macro\landscape{\global\let\LS@rot\cc@landscaperotate}
\makeatother
\usepackage{multirow}
\usepackage{subcaption}
\usepackage{ragged2e}
\usepackage{graphicx}
\usepackage{csquotes}
\usepackage{nicefrac}

\makeatletter
\long\def\@makecaption#1#2{%
  \vskip\abovecaptionskip
  {\justifying\noindent #1.\ #2\par}%
  \vskip\belowcaptionskip
}
\makeatother

\newcommand{\eqa}[1]{\begin{equation*}\begin{aligned}#1\end{aligned}\end{equation*}}

\newcommand{\lr}[1]{\left(#1\right)}

\newcommand{\LR}[1]{\{#1\}}

\newcommand{\im}{\mathrm{im}\text{ }}
\renewcommand{\ker}{\mathrm{ker}\text{ }}
\newcommand{\rank}{\mathrm{rank}\text{ }}

\newcommand{\diag}{\mathrm{diag}}
\newcommand{\cnot}[2]{\mathit{CNOT}_{#1\to #2}}
\newcommand{\bcnot}{\overline{\mathit{CNOT}}}

\newcommand{\bsgate}{\overline{\mathit{S}}}
\newcommand{\hgate}{\mathit{H}}
\newcommand{\bhgate}{\overline{\mathit{H}}}

\newcommand{\AAA}{\mathcal{A}}
\newcommand{\BBB}{\mathcal{B}}
\newcommand{\CCC}{\mathcal{C}}

\newcommand{\LLL}{\mathcal{L}}
\newcommand{\MMM}{\mathcal{M}}
\newcommand{\NNN}{\mathcal{N}}
\newcommand{\SSS}{\mathcal{S}}
\newcommand{\PPP}{\mathcal{P}}
\newcommand{\QQQ}{\mathcal{Q}}
\newcommand{\UUU}{\mathcal{U}}

\newcommand{\BB}{\mathbb{B}}

\newcommand{\FF}{\mathbb{F}}

\newcommand{\ox}{\otimes_R}
\newcommand{\x}{\times}

\newcommand{\tand}{\text{ and }}

\newcommand{\supp}{\mathrm{supp }}

\newcommand{\wt}{\mathrm{wt}}
\newcommand{\LP}{\mathrm{LP}}
\newcommand{\bZ}{\overline{Z}}
\newcommand{\bX}{\overline{X}}
\newcommand{\bY}{\overline{Y}}
\newcommand{\bket}[1]{\overline{\ket{#1}}}

\newcommand{\cc}{\mathrm{CC}}
\newcommand{\lpab}{\LP(H_a,H_b)}
\newcommand{\ccab}{\cc(H_a,H_b)}
\newcommand{\lpAB}{\LP(H_a',H_b')}

\newcommand{\code}[1]{[\![#1]\!]}
\newcommand{\lef}{\mathrm{left}}
\newcommand{\rig}{\mathrm{right}}

\theoremstyle{definition}
\newtheorem{definition}{Definition}
\newtheorem{remark}[definition]{Remark}

\newtheorem{theorem}[definition]{Theorem}
\newtheorem{lemma}[definition]{Lemma}

\newtheorem{corollary}[definition]{Corollary}

\newtheorem{example}[definition]{Example}

\renewcommand{\emph}[1]{\textit{#1}}

\begin{document}

\preprint{APS/123-QED}

\title{QGPU: Parallel logic in quantum LDPC codes}
\author{Boren Gu}
\email{b.gu@fu-berlin.de}
\affiliation{Dahlem Center for Complex Quantum Systems,
  Freie Universität Berlin, Berlin, 14195, Germany}
\affiliation{Quantum Software Lab, The University of Edinburgh, Edinburgh, EH8 9AB, United Kingdom}
\affiliation{Higgs Centre for Theoretical Physics, The University of Edinburgh, Edinburgh, EH9 3FD, United Kingdom}

\author{Andy Zeyi Liu}
\affiliation{Yale Quantum Institute, Yale University, New Haven, CT 06511, United States}
\affiliation{Department of Applied Physics, Yale University, New Haven, CT 06511, United States}

\author{Armanda O. Quintavalle}
\affiliation{Dahlem Center for Complex Quantum Systems,
  Freie Universität Berlin, Berlin, 14195, Germany}

\author{Qian Xu}
\affiliation{Institute for Quantum Information and Matter, California Institute of Technology, Pasadena, CA 91125, United States}
\affiliation{Walter Burke Institute for Theoretical Physics, California Institute of Technology, Pasadena, CA 91125, United States}

\author{Jens Eisert}
\email{jense@zedat.fu-berlin.de}
\affiliation{Dahlem Center for Complex Quantum Systems,
  Freie Universität Berlin, Berlin, 14195, Germany}

\author{Joschka Roffe}
\email{joschka@roffe.eu}
\affiliation{Quantum Software Lab, The University of Edinburgh, Edinburgh, EH8 9AB, United Kingdom}

\begin{abstract}
  Quantum low-density parity-check codes promise resource-efficient fault-tolerant quantum memories, but their use for computation is limited by the difficulty of addressing and parallelising logical operations. Here we introduce clustered-cyclic codes, a family of quantum low-density parity-check codes whose logical operators form a directly addressable clustered basis, with competitive finite-size instances such as \(\code{136,8,15}\) and \(\code{198,18,10}\). We exploit this structure through parallel product surgery, using an auxiliary copy of the data patch and a product connection code. At fixed auxiliary overhead, it performs up to \(k/2\) independent commuting logical Pauli-product measurements in one round, attaining the natural maximum for disjoint pairs. We prove distance preservation for hypergraph product codes and give a weak expansion criterion for preserving the protected data distance in product surgery for clustered-cyclic codes, verified for the finite constructions studied here. Using a single logical qubit as an auxiliary, we can synthesise any Clifford operation on the remaining logical qubits via product surgery. With half of the logical qubits prepared as auxiliaries, parallel logical Hadamard and phase gates can be implemented in two surgery rounds and disjoint directed \(\bcnot\) gates in three, translating measurement parallelism into gate parallelism. We test the scheme under a circuit-level depolarising noise model, without idling noise, inspired by neutral-atom hardware, with experimentally achieved average single- and two-qubit gate fidelities of \(99.875\%\) and \(99.85\%\), respectively. Under these physical assumptions, the \(\code{136,8,15}\) code achieves block logical error rates per syndrome-extraction round of approximately \(10^{-9}\) for memory and \(4\times10^{-9}\) for two parallel logical \(\bcnot\) gates.
\end{abstract}

\maketitle

\section*{Introduction}
Quantum error correction has emerged as one of the fastest-growing areas in quantum computing driven by overwhelming evidence that scalable and reliable quantum computation requires protecting fragile quantum information through error correction and fault tolerance~\cite{Kitaev2003,GottesmanNewReview, Roffe_2019, Roads, RevModPhys.87.307,Gottesman,NielsenChuang,MindTheGaps}. Without these tools, large-scale quantum devices would remain fundamentally unstable. In particular, the systematic study of \emph{quantum low-density parity-check} (qLDPC) codes proved to be extraordinarily fruitful \cite{TillichZemor,Panteleev,LDPCReview,LeverrierLDPC}. It became clear that genuinely `good' quantum codes could exist—codes whose distance and encoding rate both scale linearly with the number of physical qubits~\cite{Panteleev_almost_linear,Panteleev_good_and_local_testable,LeverrierLDPC}. Together with rapid experimental progress in quantum hardware~\cite{acharya_quantum_2024, ErhardPoulsenNautrupMethEtAl2021,BluvsteinEtAl2024,PuttermanNohHannEtAl2025}, these developments have reshaped expectations for scalable quantum error correction.
Nevertheless, strong code parameters alone are not sufficient and major bottlenecks remain. High thresholds and favourable scaling as such do not automatically translate into practical quantum computing. We require not only robust quantum memories, but codes that enable efficient logical quantum computation across a universal gate set~\cite{Eastin_Knill_restrictions_transversal,HorsmanSurgery,Game_of_surface_code,PhysRevA.71.022316,magic_cultivation,lukin_magic,dimensionjump}.

The distinction between memory efficiency and computational efficiency is especially important for qLDPC codes~\cite{webster2026pinnacle,cain2026shorsalgorithmpossible10000}. A high-rate code may store many logical qubits with low physical overhead, but this advantage can be lost if logical operations require large auxiliary systems or long sequential schedules. Thus, the relevant architectural question is not only whether a qLDPC code has good parameters, but whether its logical degrees of freedom can be addressed and manipulated as an organised computational resource. This is the role played by spatial patches in the surface code~\cite{HorsmanSurgery,Game_of_surface_code}. Our aim is to recover an analogous notion of addressability inside a single high-rate qLDPC block.

A common way to organise fault-tolerant logical computation is through \emph{Pauli-product measurements} (PPMs)~\cite{Gottesman1997}. 
One performs layers of mutually commuting logical Pauli measurements, updates a Pauli frame using 
the classical outcomes, and introduces non-Clifford resources through designated primitives such as 
state injection, cultivation or code switching~\cite{Game_of_surface_code,Qian_fast_and_para,CKBB,PhysRevA.71.022316,magic_cultivation,dimensionjump,color_code_magic_switch}. In this measurement-based view, computational throughput is controlled 
by a simple architectural question: how many logical PPMs can be implemented \emph{in parallel per 
round} with controlled overhead?

\begin{figure*}[t]
  \centering
  \includegraphics[width=\textwidth]{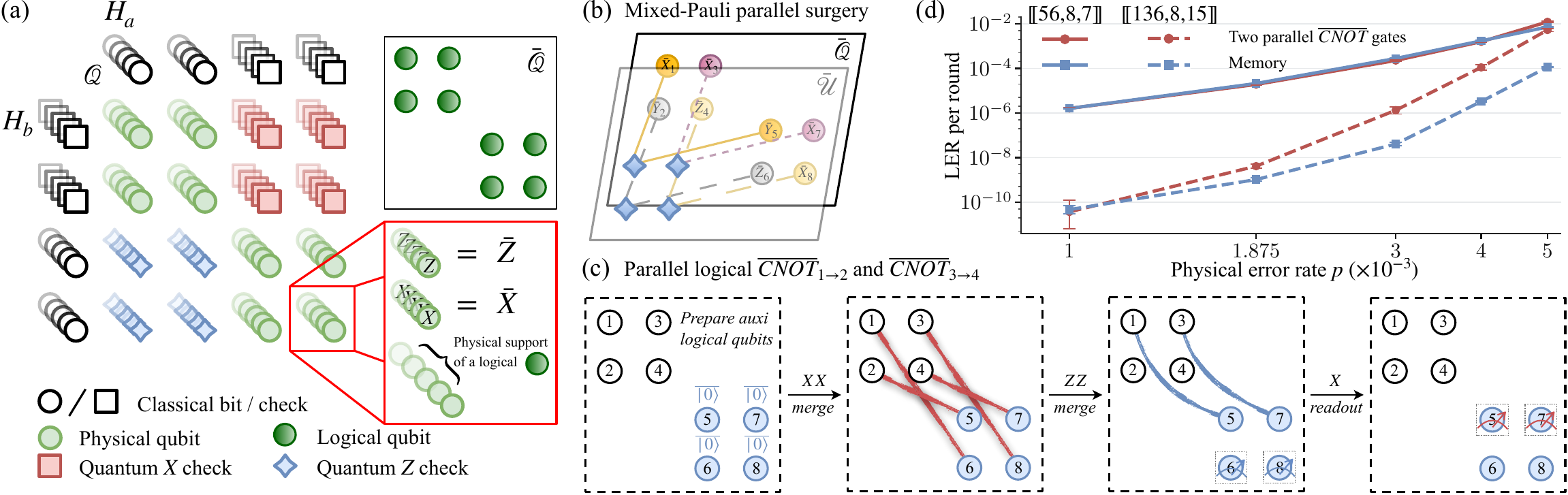}
  \caption{\textbf{Clustered logical addresses, maximally parallel surgery and logic.}
    \textbf{(a)} A CC code is constructed as the tensor product of two classical seed codes \(H_a\) and \(H_b\) over \(R=\FF_2[x]/(x^l+1)\). In the binary representation, each ring element lifts to a cluster of \(l\) physical qubits. Every CC code admits a clustered logical operator basis in which each logical operator is supported on one whole cluster, and anti-commuting \(X\)- and \(Z\)-basis elements overlap on that cluster. Each logical qubit can therefore be addressed through an associated physical cluster.
    \textbf{(b)} A mixed-type product surgery round measures four joint logical products involving \(\bX\), \(\bY\), and \(\bZ\) components, saturating the natural upper bound. The auxiliary patch \(\mathcal U\), chosen as another copy of the data code, mediates these measurements.
    \textbf{(c)} Parallel logical \(\bcnot_{1\to2}\) and \(\bcnot_{3\to4}\) gates in an eight-logical CC code, while taking the second sector as auxiliaries. Protocol details are given in Appendix~\ref{Appendix: two parallel CNOTs}.
    \textbf{(d)} Circuit-level \textit{logical error rates} (LERs) per syndrome extraction round for memory (blue squares) and the two parallel \(\bcnot\) gates (red circles) shown in (c), using the \(\code{56,8,7}\) (solid lines) and \(\code{136,8,15}\) (dashed lines) codes, with idling noise omitted. Uniform depolarising noise $p=0.1875\%$ is equivalent to average two-qubit and single-qubit gate fidelities of \(99.85\%\)~\cite{EveredXuLi2026} and \(99.875\%\), respectively. At $p=0.1\%$, the $\code{136,8,15}$ code gives block LERs per round of $4.7\times10^{-11}$ for memory and two parallel $\bcnot$ gates only observe $1$ $Z$-type logical error with no $X$-type error in around 1.4 billion shots. 
    Details of simulations are illustrated in Appendix~\ref{app:simulation-details}.}
  \label{Figure: key_results}
  \vspace{-1em}
\end{figure*}

For the surface code~\cite{Kitaev2003,TopologicalQuantumMemory,HorsmanSurgery,Game_of_surface_code}, the answer is geometrically transparent. Each logical qubit occupies its own patch, so logical qubits are spatially addressable. Lattice surgery implements logical PPMs through geometric merge and split operations between patches~\cite{HorsmanSurgery}; if \(k\) patches are available and each PPM consumes two patches, then at most \(\lfloor k/2\rfloor\) disjoint merges can be performed in one surgery round, assuming sufficient routing space. The patch architecture therefore makes parallel logical control almost automatic.

For qLDPC codes, the situation is more subtle. Encoding density is gained by storing many logical qubits in a single block, but the logical operators need not be attached to distinct groups of physical qubits and are often highly overlapping~\cite{Armanda_partitioning,bb_logical_basis,ying_li_surgery}. This strongly obscures the notion of an individually addressable logical qubit and makes parallel compilation harder. Two bottlenecks are particularly important. First, many promising qLDPC families lack a simple logical-operator basis that makes selective logical measurements uniform and easy to specify~\cite{Panteleev_almost_linear,Improved_surgery,ying_li_surgery,clement_css_surgery,Cowtan2024,webster2026pinnacle}. Second, even when logical measurements are possible, they do not automatically parallelise as in surface-code patch architectures. This limits logical throughput and increases effective overhead.

In this work, we develop a code--logic co-design for fault-tolerant computation with qLDPC codes, guided by a single architectural goal: \emph{high-throughput logical computation via parallel PPMs at small and fixed overhead}. We introduce a subfamily of \emph{lifted product} (LP) codes~\cite{Panteleev_almost_linear}, the \emph{clustered-cyclic} (CC) codes, whose logical operator basis is engineered to behave like a surface-code patch model within a single qLDPC block (Fig.~\ref{Figure: key_results}~(a)).

Building on this clustered basis, we introduce \emph{parallel product surgery}, a surgery primitive for quantum product codes based on an auxiliary copy of the data patch and an engineered product connection structure~\cite{Xanadu_Homological_measurements}. The protocol implements many logical PPMs within a single surgery round at small and fixed auxiliary overhead (Fig.~\ref{Figure: key_results}~(b)). For a CC code encoding \(k\) logical qubits, the product surgery supports individually addressed measurements of same- and mixed-type logical products. It can measure \(k/2\) disjoint PPMs in one surgery round, attaining the counting bound for disjoint pair measurements. Achieving this parallelism inside a single qLDPC block therefore shows that high encoding density need not come at the expense of logical measurement throughput, provided the code exposes a sufficiently structured logical basis.

This embeds parallel logical control directly into the code structure. Building on the existing corpus of work on logical computation in qLDPC codes ~\cite{GuoqLDPC,swaroop2025universaladaptersquantumldpc,AutQEC,malcolm2025computingefficientlyqldpccodes,Extractor,Gauging_logical,PlanarQLDPC,qian_batched,GolowichChangZhu2025,ConstantOverheadQLDPC}, we believe this work represents an important step towards the development of universal instruction sets that will ultimately underpin utility-scale qLDPC-based quantum computers.

\section*{Results}
\label{Section: key_results}

Several recent approaches address parallel logical computation, including schemes based on concatenated codes~\cite{YamasakiKoashi2024,YoshidaTamiyaYamasaki2025,Goto2024}, colour codes defined on hyperbolic manifolds~\cite{ZhuSikanderPortnoyCrossBrown2024}, and hypergraph product (HGP) codes~\cite{Qian_fast_and_para,Qian_high_rate,GolowichChangZhu2025,ChangHeYoderZhuJochymOConnor2026}. Our work complements these studies by focusing on CC codes and an explicit auxiliary construction for parallel product surgery. The clustered logical basis allows many logical operations to be addressed and scheduled within a single encoded block.

\subsection*{Clustered logical addresses}

Clustered-cyclic codes are a constrained family of LP codes over \(R=\FF_2[x]/(x^l+1)\), with odd integer lift parameter \(l\geq3\), specified by cyclically supported seed matrices \(H_a\) and \(H_b\) satisfying the joint condition in Definition~\ref{Definition: cc code}. We write \(\ccab\) for a CC code with square seed matrices of sizes \(n_a\) and \(n_b\), respectively. The resulting binary CSS code has parameters
\begin{equation*}
  \code{n=2ln_an_b,k=2n_an_b,d\leq l},
\end{equation*}
and check weight \(w=2(w_a+w_b)\), where \(w_a\) and \(w_b\) count the non-zero entries in each row of \(H_a\) and \(H_b\), respectively. Bounded seed row and column weights ensure the LDPC property.

The defining computational feature of CC codes is their clustered logical operator basis. The physical qubits decompose into clusters of size \(l\), and each basis logical operator is supported on an entire cluster. Logical operators of the same type have disjoint cluster support, while anti-commuting \(X\)- and \(Z\)-basis operators overlap on a whole cluster (Fig.~\ref{Figure: key_results}(a)). Theorem~\ref{Theorem: logical operator basis} establishes that every CC code admits such a basis. This makes logical scheduling a combinatorial problem over clusters, with each cluster providing an explicit logical address.
\begin{table}[tbp]
  \centering
  \setlength{\tabcolsep}{1pt}
  \begin{tabular}{cccc}
    \toprule
    \multicolumn{2}{c}{CC codes}
    & \multirow{2}{*}{BB codes} & \multirow{2}{*}{QRC} \\
    $k=8$ & $k=18$ & & \\
    \midrule
    $\code{24,8,3}\ (3)$ & -- & $\code{28,6,4}$ & -- \\
    $\code{40,8,5}\ (5)$ & -- & $\code{42,6,6}$ & -- \\
    $\code{56,8,7}\ (7)$ & $\code{54,18,3}\ (3)$ & $\code{56,6,8}$ & $\code{54,8,6}$ \\
    $\code{72,8,9}\ (9)$ & -- & -- & -- \\
    $\code{88,8,10}\ (11)$ & $\code{90,18,5}\ (5)$ & $\code{90,8,10}$ & $\code{90,8,10}$ \\
    $\code{104,8,11}\ (13)$ & $\code{126,18,7}\ (7)$ & $\code{120,8,12}$ & $\code{126,8,14}$ \\
    $\code{120,8,13}\ (15)$ & $\code{162,18,9}\ (9)$ & -- & -- \\
    $\code{136,8,15}\ (17)$ & $\code{198,18,10}\ (11)$ & $\code{144,12,12}$ & -- \\
    \bottomrule
  \end{tabular}
  \caption{\textbf{Representative clustered-cyclic codes.} The fixed CC catalogue is compared with \emph{bicycle-bivariate} (BB) codes~\cite{Yu-an_generalised_toric_code,BB_IBM} and \emph{quantum radial codes} (QRC)~\cite{QRC}. The CC parameters are $\code{n,k,d}\,(l)$, with $n_a=n_b=2$ for $k=8$ and $n_a=n_b=3$ for $k=18$. Every CC code listed here has check weight $8$. Appendix~\ref{Appendix: code examples} specifies the seed matrices.}
  \label{Table: example code para}
\end{table}
\begin{figure}
    \centering
    \includegraphics[width=\linewidth]{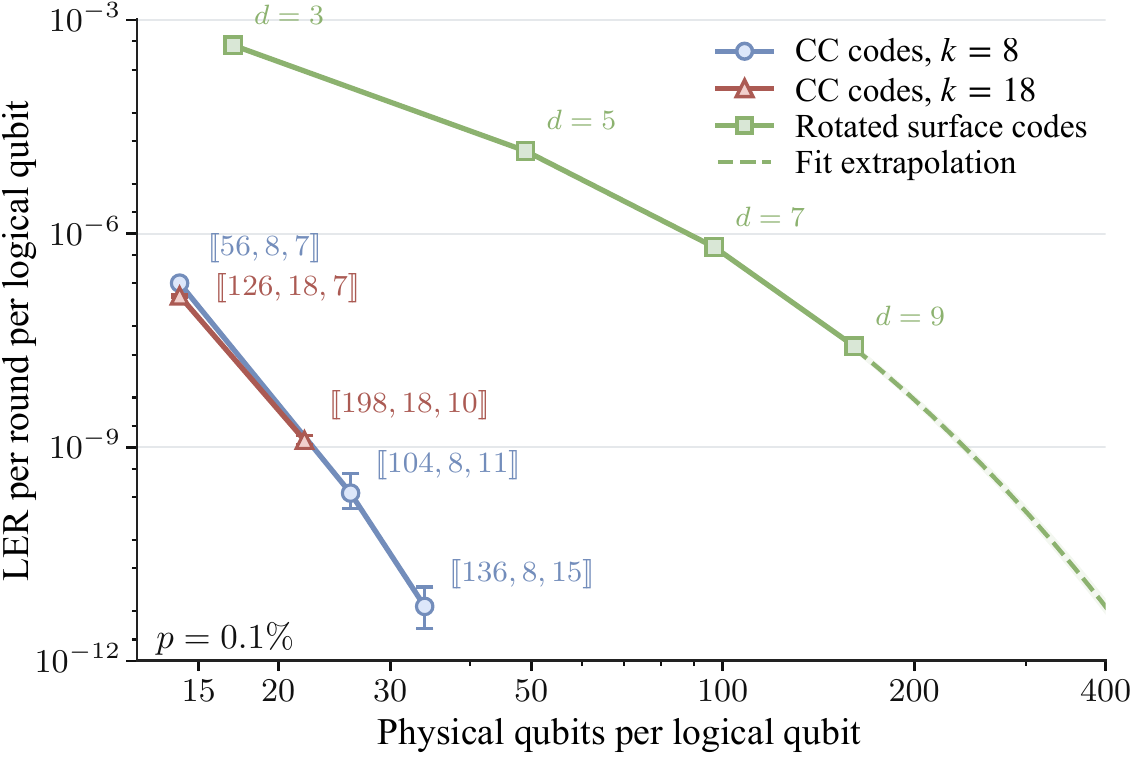}
    \caption{\textbf{Memory footprint advantage over rotated surface-codes.} Circuit-level LER averaged over $X$ and $Z$ bases per syndrome-extraction round per logical qubit at physical error rate $p=0.1\%$, omitting idling noise. CC codes simulated here are subset of representatives in Table~\ref{Table: example code para}.}
    \label{fig:footprint}
    \vspace{-1em}
\end{figure}
Table~\ref{Table: example code para} lists representative CC codes found by numerical search, including \(\code{136,8,15}\) and \(\code{198,18,10}\), and compares them with other known code families.

\subsection*{Parallel Pauli measurements in one surgery round}

Parallel product surgery is a surgery primitive tailored to quantum product codes. Starting from the CSS surgery framework~\cite{Xanadu_Homological_measurements}, we choose the auxiliary patch to be a copy of the data patch. For same-type measurements, the connection maps are chosen to be the \textit{parity-check matrices} (PCMs) of a product  code, as in Definition~\ref{Definition: product connection code}. For CC codes, such connection map is specialised to an LP code whose seed entries are restricted to $\{0, 1\}\subset R$, the rows of \(H_Z'\) specify the measured \(Z\)-type logical products, and \(M=\rank H_Z'\leq k/2\) counts the independent PPMs. In particular, choosing \(H_a'=I_{n_a}\) and \(H_b'=I_{n_b}\) gives \(H_Z'=[I_{n_an_b}\ I_{n_an_b}]\), measuring \(k/2\) disjoint pairs, each containing one logical qubit from each sector. More generally, this maximum is attained whenever at least one of the connection seed matrices is full rank (Theorem~\ref{Theorem: max merge}).

To support mixed-type products, we add an \(X\) attachment while retaining the product connection for the \(Z\) component. Under the compatibility conditions in Methods, the standard three-layer construction~\cite{Xanadu_Homological_measurements} reduces to a two-layer PCM that uses another copy of the data code as the auxiliary patch, matching the physical overhead of same-type product surgery. For each fixed \(P,Q\in\{X,Y,Z\}\), the resulting product connection measures the \(k/2\) disjoint logical products \(\overline P_i^{\lef}\overline Q_i^{\rig}\) in one surgery round, as illustrated in Method~\ref{Subsec: surgery procedure}. These nine choices generalise the same-type pair merges (Corollary~\ref{Corollary: mixed pair merges}).

Allowing both attachment components to vary beyond the product connection form gives further addressability while retaining the same auxiliary patch (Remark~\ref{remark: general attachment}). Rows of \((g_1^*\mid f_1^*)\) with entries in $\{0, 1\}\subset R$ specify the \(X\) and \(Z\) components of the selected logical products, and
\[
M=\rank(g_1^*\mid f_1^*)\leq n_an_b=k/2.
\]
The rows can describe individually addressed \(\bX\), \(\bY\) or \(\bZ\) measurements, selected pairs with different Pauli types, or any prescribed commuting family within this bound.

Fault tolerance is established by two points: bounded check weights and preservation of merged code distance. Same-type product surgery remains LDPC when the data and product connection codes are LDPC, and we numerically verify that mixed-type merge checks have weight at most \(13\) for the CC codes presented. Same-type product surgery on HGP codes preserves both \(d_X\) and \(d_Z\) (Appendix~\ref{Appendix: HGP distance preserving}). For CC codes, weak expansion (Eq.~\eqref{eq: weak expansion}) ensures that measuring \(M\) independent logical PPMs leaves \(k-M\) protected data logical qubits with distance at least \(d\) (Theorem~\ref{Theorem: mixed protected distance}). Appendix~\ref{Appendix: sparse mixed surgery} verifies this condition using the actual physical attachment weights for all \(579\) selected constructions, covering all \(3k\) addressed single measurements and nine uniform pair rounds for each of thirteen presented CC codes.

\subsection*{A maximally parallel Clifford block}

Using a single logical qubit as a reusable auxiliary, we can synthesise any Clifford operation on the remaining \(k-1\) logical qubits. The construction uses addressed single- and two-logical PPMs, Pauli-state preparation and Pauli-frame tracking: each \(\bhgate\) or \(\bsgate\) takes two surgery rounds, and a \(\bcnot\) between any two data logical qubits takes three (Theorem~\ref{Theorem: one auxiliary clifford}).

When half of the logical qubits are prepared as auxiliaries, measurement parallelism enables logical gate parallelism on the other half of the logical qubits.  Any selected combination of logical Hadamard and phase gates can be implemented in two surgery rounds, and disjoint directed \(\bcnot\) gates in three. The single-qubit gates may act on all \(k/2\) data logical qubits, while up to \(\lfloor k/4\rfloor\) \(\bcnot\) gates can run simultaneously with independently chosen directions. A disjoint layer combining Hadamard or phase gates with \(\bcnot\) gates also takes three rounds (Theorem~\ref{Theorem: full clifford group}).

\subsection*{Reconfigurable hardware implementation}

The code--logic co-design developed here is well matched to platforms supporting \emph{dynamic, sparse reconfiguration} of the interaction graph. During memory operation, each patch extracts syndromes on its native graph; surgery temporarily activates a structured set of couplings between the data and auxiliary patches. The product connection defines this interface algebraically through its attachment and dressing maps. The auxiliary patch is a copy of the data code and can be reused between surgery rounds using the same preparation, measurement, and reset operations.

This interaction pattern is well suited to optical-tweezer neutral-atom arrays~\cite{LukinFaultTolerant,Qian_reconfigurable}, shuttling-based trapped-ion architectures~\cite{QCCD,IonTrapBreakEven}, and superconducting processors with tunable or bus-mediated couplers~\cite{BB_IBM,IBMTuneable}. For CC codes with suitable cyclic seed structure and layout, both memory operation and product surgery can be implemented on reconfigurable neutral-atom arrays with constant shuttling depth per syndrome-extraction round~\cite{ZhouGuEisertZhao2026}.

\section*{Discussion}
\label{Section: conclusion}

This work develops a code--logic co-design perspective for fault-tolerant computation with qLDPC codes, motivated by the practical observation that, even when good code parameters are available, using a code for computation can still be constrained by (i) how transparently the logical operators can be represented and acted on, and (ii) how much parallelism can be achieved in logical measurement layers at a fixed, small overhead. Within this framing, CC codes are introduced as an LP code subfamily whose logical operator basis is engineered to have a ‘clustered’ organisation, with the aim of making selective and parallel logical measurements and scheduling more systematic. The clustered logical basis of CC codes makes logical control a combinatorial problem over clusters, while parallel product surgery turns this addressability into simultaneous logical PPMs and hence parallel logic using an auxiliary copy of the data code.

Same- and mixed-type products can be measured at fixed auxiliary overhead, with up to \(k/2\) disjoint pair measurements in one surgery round. One reusable auxiliary logical qubit allows Clifford computation on the remaining \(k-1\) logical qubits. Allocating half of the logical qubits to auxiliaries supports parallel layers of independently chosen \(\bhgate\) or \(\bsgate\) gates in two surgery rounds, and disjoint directed \(\bcnot\) layers in three, after auxiliary preparation. The logical allocation therefore controls the balance between data capacity and parallel gate execution within a single code block.

On fault-tolerance, we show that the same-type parallel product surgery procedure is fault-tolerant for HGP codes in full generality as long as both data code and product connection
codes are LDPC. For CC codes, weak expansion gives a sufficient condition for preserving the protected data distance. The finite constructions for addressed single measurements and the nine uniform pair rounds satisfy this criterion, with merged-check weights at most \(13\). The circuit-level simulations complement these protections: at experimentally achieved average two-qubit and single-qubit gate fidelities of \(99.85\%\) and \(99.875\%\), respectively, with reconfigurable atoms~\cite{EveredXuLi2026}, the \(\code{136,8,15}\) code gives estimated block LERs per round of approximately \(10^{-9}\) for memory and \(4\times10^{-9}\) for two parallel logical \(\bcnot\) gates. At \(p=0.1\%\), the memory block LER is \(4.7\times10^{-11}\) per round, while the two-gate simulation recorded one failure in the \(Z\)-basis diagnostic and none in the \(X\)-basis diagnostic, across approximately \(1.4\times10^9\) shots.

The reusable auxiliary patch and algebraically specified couplings make this approach well matched to reconfigurable hardware~\cite{ZhouGuEisertZhao2026}, as both CC memory and product surgery can be implemented on reconfigurable neutral-atom arrays with constant shuttling depth per syndrome-extraction round with suitable cyclic seed codes.

As a final comment, our work also raises some natural pathways for further investigation. Establishing weak expansion and sparse dressings for broader code families and more general logical measurement patterns remains open. Beyond Clifford operations, the clustered basis supports parallel preparation of raw magic states at selected logical addresses. This provides a starting point for parallel magic-state distillation and motivates exploring cultivation protocols tailored to CC codes. More broadly, these results highlight the logical-basis structure of high-rate qLDPC codes as a design resource. Different code families offer distinct strengths in finite-size parameters, logical addressability, automorphism groups and compatibility with measurement gadgets. Developing computational primitives that exploit these structures is essential for an efficient universal fault-tolerant quantum computer.

\section*{Methods}

\subsection{Clustered-cyclic quantum LDPC code}
\label{Section: CC qLDPC code}

The ring, binary representation and lifted product notation used below are collected in Appendix~\ref{Section: preliminaries}.
A central ingredient for fault-tolerant logical operations is the availability of a well-organised logical operator basis. In HGP codes, the tensor product structure induces a canonical decomposition of logical operators via the Künneth formula~\cite{Reshape, Armanda_partitioning}. This canonical logical operator basis supports structured parallel implementations of logical operations~\cite{Qian_fast_and_para,Qian_high_rate,GolowichChangZhu2025}.

In this section, we introduce a subclass of LP codes, which we call \emph{clustered-cyclic} (CC) codes. By imposing a joint condition on the classical seed codes and decomposing the product complex into field components, we obtain a logical operator basis with explicit representatives and a well-controlled symmetric structure. These constraints are chosen with the aim of facilitating logical measurements, including both same-type and mixed-type products, which can then be scheduled in a highly parallel manner in our surgery-based setting introduced in Section~\ref{Section: parallelisable surgery}. As we show in Section~\ref{Section: clifford}, these measurement primitives enable parallel implementations of logical Clifford operations on $k/2$ data logical qubits, using the other $k/2$ as auxiliary logical qubits. We begin with a formal definition of CC codes.

\begin{definition}[Clustered-cyclic code]
  \label{Definition: cc code}
  Let $l\geq 3$ be an odd integer, let $R=\FF_2[x]/(x^l+1)$, and write $\chi=1+x+\cdots+x^{l-1}$. A \textit{clustered-cyclic code} is a lifted product code $\lpab$ whose seed matrices $H_a\in R^{n_a\times n_a}$ and $H_b\in R^{n_b\times n_b}$ are square. Their entries are either zero or binomials $x^u+x^v$, with distinct $u,v\in\{0,\ldots,l-1\}$. The non-zero entries have cyclic support patterns, and all rows and columns of $H_a$ or $H_b$ contain the same number of non-zero entries. These seed matrices must satisfy
    \begin{equation}
      \gcd(\det H_a,\det H_b,\chi)=1.
    \end{equation}
    The $\gcd$ is the greatest common divisor computed in $\FF_2[x]$ using any polynomial representatives of the determinants. Equivalently, for each monic irreducible factor $f$ of $\chi$, at least one of $H_a\bmod f$ and $H_b\bmod f$ is invertible over $\FF_2[x]/(f)$.

  We denote this code by $\ccab$. Its CSS PCMs are given by Eq.~\eqref{eq: LP check matrices}, with $m_a=n_a$ and $m_b=n_b$:
  \begin{equation}
    \begin{aligned}
      H_X&=\left(H_a\ox I_{n_b}\ \middle|\ I_{n_a}\ox H_b\right),\\
      H_Z&=\left(I_{n_a}\ox H_b^*\ \middle|\ H_a^*\ox I_{n_b}\right).
    \end{aligned}
  \end{equation}

\end{definition}

The seed sizes $n_a$ and $n_b$ need not be equal. For example, over $R=\FF_2[x]/(x^5+1)$, the matrices
\begin{equation*}
  H_a=\begin{pmatrix}
    x^3+x^4&x^2+x^3&0\\
    0&x+x^2&x^3+x^4\\
    x+x^3&0&x^3+x^4
  \end{pmatrix},\
  H_b=\begin{pmatrix}
    1+x^2&0\\
    0&x^2+x^4
  \end{pmatrix}
\end{equation*}
form a valid seed pair under Definition~\ref{Definition: cc code}.

Both matrices $H_X$ and $H_Z$ have size $(n_an_b) \x (2n_an_b)$. After taking the binary representation, these two matrices are mapped to matrices $\BB(H_X)$ and $\BB(H_Z)$ of size $(ln_an_b)\x (2ln_an_b)$ over $\FF_2$. Thus the code has $n=2ln_an_b$ physical qubits.
Following Theorem~\ref{Theorem: logical operator basis} and Corollary~\ref{Corollary: code para}, the code $\ccab$ has parameters
\eqa{
  \code{n=2ln_an_b,\ k=2n_an_b,\ d\leq l},
}
where $n$ is the number of physical qubits, $k$ is the number of logical qubits, and $d$ is the code distance. When the code distance saturates its upper bound, we have $n = k d$. For the displayed examples in Table~\ref{Table: example code para}, this scaling is achieved for small $l$, whereas the distances go below $l$ for $l\geq 11$. Once the row weights of the seed matrices, i.e., the number of non-zero entries of each row, are fixed, CC codes have constant check weight regardless of the lift parameter $l$ or seed matrix sizes $n_a$ and $n_b$. In particular, if the row weights of $H_a$ and $H_b$ are $w_a$ and $w_b$, respectively, the check weight of the resulting CC code $\ccab$ is

\begin{equation}
  w = 2(w_a+w_b).
  \label{eq: check weight}
\end{equation}

Table~\ref{Table: example code para} lists numerically discovered representative finite examples of CC codes. The codes are ordered by their number of physical qubits and are compared with \emph{bicycle-bivariate} (BB) codes~\cite{Yu-an_generalised_toric_code,BB_IBM} and \emph{quantum radial codes}  (QRC)~\cite{QRC} of similar sizes. The memory footprint advantage of CC codes over rotated surface codes is displayed in Figure~\ref{fig:footprint}.

\subsubsection{Clustered logical operator basis}
Surface codes come with a particularly transparent set of logical operators: one can identify logical $X$ and $Z$ operators as string-like Pauli operators on the underlying 2D lattice, and this geometric picture directly underpins fault-tolerant logical primitives such as lattice surgery. By contrast, for most known families of high-rate qLDPC codes, logical operators are typically not given by an equally simple logical operator basis, which makes it difficult to systematically design and schedule logical operations. The purpose of this subsection is to show that CC codes admit a \emph{clustered} logical operator basis, in which each logical operator has an explicit representative supported on a `cluster' of $l$ physical qubits, as demonstrated in Figure~\ref{Figure: key_results} (a). This added structure plays the same organisational role as the geometric logical basis in surface codes: it provides a concrete set of measurement primitives that can be composed into fault-tolerant logical operations.

\begin{definition}[Clustered logical operator basis]
  \label{Definition: clustered basis}
  A choice of binary logical Pauli representatives $(\bX_i,\bZ_i)_{i=1}^k$ for an LP code over $R=\FF_2[x]/(x^l+1)$ is \emph{clustered} if
  \begin{enumerate}
    \item Each representative is supported on one entire cluster in the partition of the physical qubits into consecutive blocks of $l$ qubits.
    \item The representatives $\bX_i$ and $\bZ_i$ have the same support and anticommute, while $\bX_i$ commutes with $\bZ_j$ for $i\neq j$.
  \end{enumerate}
\end{definition}

\begin{theorem}[Clustered logical operator basis of CC codes]
  \label{Theorem: logical operator basis}
  Clustered-cyclic codes $\ccab$ defined over $R=\FF_2[x]/(x^l+1)$ have a clustered logical operator basis as defined in Definition~\ref{Definition: clustered basis}.

\end{theorem}

\begin{proof}
  Put $\chi=1+x+\cdots+x^{l-1}$. By Corollary~\ref{Corollary: ker and im of Ha}, the logical classes of both Pauli types have unique representatives in $\chi R^{2{n_an_b}}$. Thus both types are represented over $R$ by the rows of
  \begin{equation}
    \label{eq: logical basis matrix}
    L=\left(
    \begin{array}{c|c}
      \diag_{n_an_b}(\chi) & 0_{n_an_b} \\
      \midrule
      0_{n_an_b}& \diag_{n_an_b}(\chi)\\
    \end{array}\right),
  \end{equation}
  where $0_{n_an_b}$ denotes the zero matrix. Selecting one binary row from each repeated block of $\BB(L)$ gives the independent binary representatives
  \begin{equation}
    L_X^{\BB}=L_Z^{\BB}=I_{2{n_an_b}}\otimes\mathbf{1}_l^\intercal,
  \end{equation}
  where $\mathbf{1}_l$ is the column vector of $l$ ones. Each row is supported on one entire cluster. Since $l$ is odd, $L_X^{\BB}(L_Z^{\BB})^\intercal=I_{2{n_an_b}}$. Hence, corresponding $X$- and $Z$-type representatives anticommute, and every other pair commutes. Corollary~\ref{Corollary: ker and im of Ha} establishes completeness, so these representatives form a clustered logical operator basis.
\end{proof}

\begin{remark}[Rank-one binary lift of $\chi$]
  \label{Remark: chi rank}
  Over $R=\FF_2[x]/(x^l+1)$, the binary representation of $\chi$ is the all-ones matrix: $\BB(\chi)_{i,j}=1$ for all $i,j\in\LR{1,2,\dots,l}$. This matrix has rank $1$. Consequently, each row of $L$ in Eq.~\eqref{eq: logical basis matrix} lifts to $l$ identical binary rows representing a single logical operator of Hamming weight $l$.
\end{remark}

We now present a specific instance of a $\code{12,4,3}$ code as a concrete example of the CC code construction:

\begin{example}
  \label{Example: 12,4,3}
  Consider the code $\ccab$ defined over the ring $R=\FF_2[x]/(x^3+1)$ with seed matrices defined as follows
  \eqa{
    H_a =
    \begin{pmatrix}
      1+x & 1+x^2\\
      1+x^2 & 1+x
    \end{pmatrix}\tand H_b=
    \begin{pmatrix}
      1+x
    \end{pmatrix}.
  }

  By Eq.~\eqref{eq: logical basis matrix}, the independent binary representatives are
  \begin{equation*}
    L_X^{\BB}=L_Z^{\BB}=
    \begin{pmatrix}
      111&000&000&000\\
      000&111&000&000\\
      000&000&111&000\\
      000&000&000&111
    \end{pmatrix}.
  \end{equation*}
  Each row represents one logical operator of the corresponding Pauli type.
\end{example}

The clustered basis gives $n=12$, $k=4$ and $d\leq3$ for this example. An exhaustive search confirms $d=3$.
\begin{corollary}[Code parameters of CC codes]
  \label{Corollary: code para}
  The codes $\ccab$ over $R=\FF_2[x]/(x^l+1)$ have code parameters
  \eqa{\code{n=2ln_an_b,\ k=2n_an_b,\ d\leq l}.}
\end{corollary}
\begin{proof}
  The code has $n=2l{n_an_b}$ physical qubits. By Corollary~\ref{Corollary: ker and im of Ha}, both binary PCMs have rank ${n_an_b}(l-1)$. The CSS dimension formula therefore gives
  \begin{equation}
    \begin{split}
      k&=n-\rank\BB(H_X)-\rank\BB(H_Z)\\
      &=2l{n_an_b}-2{n_an_b}(l-1)=2{n_an_b}.
    \end{split}
  \end{equation}
  Theorem~\ref{Theorem: logical operator basis} supplies nontrivial logical representatives of each Pauli type with Hamming weight $l$. Thus $d_X,d_Z\leq l$, and consequently $d\leq l$.
\end{proof}

\subsection{Parallel surgery for product codes}
\label{Section: parallelisable surgery}

Our protocol uses an auxiliary copy of the data product complex to perform parallel logical PPMs. We call a logical Pauli product \emph{same-type} if it involves only \(X\)-type or only \(Z\)-type logical operators, such as \(\bX_i\bX_j\) or \(\bZ_i\bZ_j\), and \emph{mixed-type} if it involves both, such as \(\bX_i\bZ_j\) or \(\bX_i\bY_j\). We first describe the same-type construction, then derive a mixed-type extension with the same auxiliary patch by folding the three-layer PCM. For CC codes, product connections attain $k/2$ parallel PPMs for each of the nine uniform Pauli-pair choices, while general attachments also enable individually addressed measurements. Each surgery round adds $n$ auxiliary data qubits and $n$ check auxiliaries, independent of the number of PPMs. The attachment maps specify the measured logical space, while the attachment and dressing weights determine the merged check weights.

\subsubsection{Parallel product surgery for same-type measurements}
Let the data and auxiliary CSS codes be represented by the length-$2$ chain complexes
\begin{equation*}
  \QQQ:
  \begin{tikzcd}[row sep=scriptsize, column sep=small]
    Q_2 \arrow[r, "H_Z^*"] & Q_1 \arrow[r, "H_X"] & Q_0
  \end{tikzcd}
  \tand \UUU:
  \begin{tikzcd}[row sep=scriptsize, column sep=small]
    U_1 \arrow[r, "\partial_1"] & U_0 \arrow[r, "\partial_0"] & U_{-1}
  \end{tikzcd}.
\end{equation*}

\begin{definition}[Code surgery and merged code~\cite{Xanadu_Homological_measurements}]
  \label{Definition: code surgery}
  A $Z$-type surgery is specified by the commuting diagram
  \begin{equation}
    \label{eq: CSS surgery}
    \begin{tikzcd}
      U_1 \arrow[r, "\partial_1"]\arrow[rd,"f_1"] & U_0 \arrow[r, "\partial_0"]\arrow[rd,"f_0"] & U_{-1}\\
      Q_2 \arrow[swap,r, "H_Z^*"] & Q_1 \arrow[swap,r, "H_X"] & Q_0
    \end{tikzcd},
  \end{equation}
  with $H_Xf_1=f_0\partial_1$. Its mapping cone is the merged CSS code on $Q_1\oplus U_0$, with PCMs
  \begin{equation}
    \widetilde{H_X}=\begin{pmatrix}H_X&f_0\\0&\partial_0\end{pmatrix}
    \tand
    \widetilde{H_Z}=\begin{pmatrix}H_Z&0\\f_1^*&\partial_1^*\end{pmatrix}.
  \end{equation}
  Products of stabilisers in the merged code whose auxiliary support cancels act on the data patch alone, and their non-trivial logical classes determine the $Z$-type PPMs counted in Eq.~\eqref{eq: number of bits being extracted}. Logical operators supported entirely on the auxiliary patch are regarded as gauge operators.
\end{definition}

Let $V=\BB(f_1)[\ker\BB(\partial_1)],\tand S_Z=\im\BB(H_Z^*)$. The chain-map condition ensures $V\subseteq\ker\BB(H_X)$. The number of independent logical PPMs extracted during surgery in parallel is therefore
  \begin{equation}
    \label{eq: number of bits being extracted}
    M=\dim_{\FF_2}\frac{V+S_Z}{S_Z}.
  \end{equation}
  This reduces to $M=\dim_{\FF_2}V$ when $V\cap S_Z=\{0\}$.

\begin{definition}[Product connection code]
  \label{Definition: product connection code}
  Given $\QQQ=\lpab$ over $R$, with $H_a\in R^{m_a\times n_a}$ and $H_b\in R^{m_b\times n_b}$, a \emph{product connection code} is $\PPP=\lpAB$ over the same ring, where $H_a'\in R^{m_a\times n_a}$ and $H_b'\in R^{m_b\times n_b}$.
\end{definition}

\begin{lemma}[Parallel product surgery]
  \label{Lemma: surgery for product codes}
  For a product code $\QQQ$ and a product connection code $\PPP$, taking another copy of data code as the auxiliary patch gives the surgery diagram
  \begin{equation}
    \label{eq: surgery diagram for product codes}
    \begin{tikzcd}
      Q_2^{\mathrm{auxi}} \arrow[r, "H_Z^*"]\arrow[rd,"{H_Z'}^*"] & Q_1^{\mathrm{auxi}} \arrow[r, "H_X"]\arrow[rd,"H_X'"] & Q_0^{\mathrm{auxi}}\\
      Q_2 \arrow[swap,r, "H_Z^*"] & Q_1 \arrow[swap,r, "H_X"] & Q_0
    \end{tikzcd},
  \end{equation}
  where $\QQQ^{\mathrm{auxi}} = \QQQ$, and the CSS PCMs of this merged code are
  \begin{equation}
    \label{eq: product merged checks}
    \widetilde{H_X}=\begin{pmatrix}H_X&H_X'\\0&H_X\end{pmatrix}
    \tand
    \widetilde{H_Z}=\begin{pmatrix}H_Z&0\\H_Z'&H_Z\end{pmatrix}.
  \end{equation}
\end{lemma}
\begin{proof}
  The product structure gives
  \begin{equation}
    \label{eq: product mixed identity}
    H_X{H_Z'}^*=H_a\ox H_b'+H_a'\ox H_b=H_X'H_Z^*,
  \end{equation}
  which is precisely the required chain-map commutativity.
\end{proof}

\subsubsection{Parallel surgery for same-type measurements of CC codes}
\label{Subsec: maximal para CC surgery}
For a CC code $\QQQ=\ccab$, we choose the seed matrices of the product connection code $\PPP$ in Lemma~\ref{Lemma: surgery for product codes} to have entries in $\{0,1\}\subset R$. Each non-zero row $v\in\FF_2^k$ of $H_Z'$ specifies the $Z$-type logical product $\bZ(v)=\prod_j\bZ_j^{v_j}$ in the clustered basis. This product acts on $r=\mathrm{wt}(v)$ logical qubits and has physical representative $\chi v$ of weight $lr$.

Writing $h$ for the corresponding row of $H_Z$, the merged check row $\begin{pmatrix}v&h\end{pmatrix}$ lifts to $l$ physical $Z$ checks, each of weight at most $w+r$, where $w$ is the maximum weight of an original check. Their product is represented by
\[
  \chi\begin{pmatrix}v&h\end{pmatrix}=\begin{pmatrix}\chi v&0\end{pmatrix},
\]
since every entry of $h$ is zero or a binomial and $\chi(x^\alpha+x^\beta)=0$ in $R$. The auxiliary support cancels, so multiplying these $l$ check outcomes gives the outcome of measuring $\bZ(v)$.

\begin{corollary}[Rank formula for number of parallel logical measurements]
  The number of independent PPMs in one surgery round is $M=\rank H_Z'$.
\end{corollary}
\begin{proof}
  By Corollary~\ref{Corollary: ker and im of Ha}, $\ker\BB(H_Z^*)=\im\BB(\chi I_{n_an_b})$. Thus the data operators extracted by the surgery span
  \begin{equation}
    \BB({H_Z'}^*)[\ker\BB(H_Z^*)]=\im\BB(\chi{H_Z'}^*).
  \end{equation}
  Every vector in this image is constant within each cluster. Since $l$ is odd, every non-zero such vector has odd parity in at least one cluster. By contrast, every data stabiliser has even parity in every cluster, because the entries of $H_Z$ are binomials or zero, thus, their product cancels out. Since $H_Z'$ is binary and $\BB(\chi)$ has rank one, Eq.~\eqref{eq: number of bits being extracted} gives $M=\rank\BB(\chi{H_Z'}^*)=\rank H_Z'$.
\end{proof}

\begin{definition}[Maximally parallel]
  A surgery protocol is \emph{maximally parallel for disjoint pair measurements} if it performs $\lfloor k/2\rfloor$ such measurements in one round.
\end{definition}

For CC codes, which always encode $k=2n_an_b$ logical qubits, this upper bound is achieved provided that the product connection code $\PPP$ used in the surgery has full rank PCMs. In particular, when $\rank H_Z'=n_an_b=k/2$, all targeted logical pairs can be measured in a single surgery round, thereby realising maximal parallelism.

\begin{theorem}[Maximally parallel $Z$-type merges for CC codes]
  \label{Theorem: max merge}
  A CC code $\ccab$ encoding $k=2n_an_b$ logical qubits supports at most $M_{\max}=k/2=n_an_b$ independent $Z$-type logical PPMs in one product surgery round. This bound is attained when $\rank H_Z'=n_an_b$, which for square binary connection seed matrices holds if and only if at least one of $H_a'$ and $H_b'$ is full rank.
\end{theorem}
\begin{proof}
  The number of independent merges is $M=\rank H_Z'\leq n_an_b$. As $H_a'$ and $H_b'$ are square matrices with entries in $\{0,1\}$, $H_Z'$ has full row rank if and only if at least one of $H_a'$ and $H_b'$ is full rank. Indeed, the product construction gives
  \[
    \ker (H_Z')^\intercal
    =\ker H_a'\ox\ker H_b'.
  \]
  For example, choosing $H_a'=I_{n_a}$ and $H_b'=I_{n_b}$ gives $H_Z'=\begin{pmatrix}I_{n_an_b}&I_{n_an_b}\end{pmatrix}$. The logical action is to merge each $i$th $Z$-type logical operator in the left sector with the $i$th in the right sector, measuring the $n_an_b=k/2$ disjoint products $\overline{Z}_i^{\mathrm{left}}\overline{Z}_i^{\mathrm{right}}$, $i=1,\ldots,n_an_b$.
\end{proof}

\subsubsection{Parallel product surgery for mixed-type measurements}
  \label{Subsec: mixed product surgery}

  We start from the standard mixed-type homological measurement construction~\cite{Xanadu_Homological_measurements}, which uses separate auxiliary systems for the $Z$ and $X$ components. After combining the required merge generators, its three-layer PCM can be written in our notation as
  \[
    \widetilde H_{\rm three}=
    \left(\begin{array}{ccc|ccc}
      H_X&f_0&0&0&0&0\\
      0&0&0&H_Z&0&g_0\\
      g_1^*&0&{\partial_1'}^*&f_1^*&\partial_1^*&0\\
      0&\Gamma&0&0&0&\Gamma'
    \end{array}\right),
  \]
  with columns ordered as $(X_{\mathrm{data}},X_U,X_{U'}\mid Z_{\mathrm{data}},Z_U,Z_{U'})$. Here $f_1,g_1$ specify the $Z,X$ attachment components, $f_0,g_0$ dress the original checks, and $\partial_1,\partial_1'$ give the two auxiliary boundaries. The maps are chosen so that all check rows commute, requiring
  \begin{equation}
    \label{eq: mixed folding conditions}
    H_Xf_1=f_0\partial_1,\ H_Zg_1=g_0\partial_1'\tand g_1^*f_1+f_1^*g_1=0,
  \end{equation}
  as well as the auxiliary checks to satisfy $\Gamma\partial_1+\Gamma'\partial_1'=0$. Logical PPMs are extracted from the data patch as products of stabilisers of the merged code specified by the row
  \[
    \left(\begin{array}{ccc|ccc}
      g_1^*&0&{\partial_1'}^*&f_1^*&\partial_1^*&0\\
    \end{array}\right).
  \]

  For parallel product surgery, take both auxiliary systems to be copies of the data product complex. Choose the $Z$ attachment from the product connection code $\PPP$ in Lemma~\ref{Lemma: surgery for product codes}, giving
  \[
    \begin{aligned}
      \partial_1'=\partial_1&=H_Z^*,&\partial_0&=H_X,\\
      f_1&={H_Z'}^*,&f_0&=H_X'.
    \end{aligned}
  \]
  The first chain-map identity is then automatic by Eq.~\eqref{eq: product mixed identity}. The remaining conditions on the $X$ attachment and its dressing are
  \begin{equation}
    \label{eq: mixed folding conditions for product codes}
    H_Zg_1=g_0H_Z^*\tand g_1^*{H_Z'}^*+H_Z'g_1=0.
  \end{equation}
  As $H_XH_Z^*=0$, a compatible choice of auxiliary checks is
  \[
    (\Gamma\mid\Gamma')=
    \begin{pmatrix}0&H_X\\ I&I\end{pmatrix}.
  \]
  After binary lifting, the $(I\mid I)$ row gives the auxiliary bridge checks $X_{U_i}Z_{U_i'}$. Preparing $\ket{+}_U^{\otimes n}\ket{0}_{U'}^{\otimes n}$ fixes these checks to $+1$, independent of the data state. To separate $U'$ from the remaining qubits, apply
  \[
    W=\left(\prod_{i=1}^n\cnot{U_i'}{U_i}\right)
    \hgate_{U'}^{\otimes n}.
  \]
  Each bridge $X_{U_i}Z_{U_i'}$ becomes the single-qubit stabiliser $X_{U_i'}$. The transformed second and fourth row groups still contain $X$ factors on $U'$, specified by $g_0$ and $H_X$, respectively. Multiplying each such check by the single-qubit stabilisers on its $U'$ support cancels these factors, since $X_{U_i'}^2=I$. This changes only the choice of stabiliser generators, leaving the merged code unchanged. The resulting reduction is
  \setlength{\arraycolsep}{2pt}
  \begin{equation}
    \label{eq: mixed PCM reduction}
    \left(\begin{array}{ccc|ccc}
      H_X&H_X'&0&0&0&0\\
      0&0&0&H_Z&0&g_0\\
      g_1^*&0&H_Z&H_Z'&H_Z&0\\
      0&0&0&0&0&H_X\\
      0&I&0&0&0&I
    \end{array}\right)\xrightarrow{\scriptstyle W}
    \left(\begin{array}{ccc|ccc}
      H_X&H_X'&0&0&0&0\\
      0&g_0&0&H_Z&0&0\\
      g_1^*&0&0&H_Z'&H_Z&0\\
      0&H_X&0&0&0&0\\ \hline
      0&0&I&0&0&0
    \end{array}\right).
  \end{equation}
  The last row group now fixes $U'$ in $\ket{+}^{\otimes n}$ independently of the remaining qubits. Omitting this prepared layer gives the two-layer merged PCM
  \begin{equation}
    \label{eq: mixed product merged checks}
    \widetilde H=
    \left(\begin{array}{cc|cc}
      H_X&H_X'&0&0\\
      0&g_0&H_Z&0\\
      g_1^*&0&H_Z'&H_Z\\
      0&H_X&0&0
    \end{array}\right),
  \end{equation}
  in column order $(X_{\mathrm{data}},X_{\mathrm{auxi}}\mid Z_{\mathrm{data}},Z_{\mathrm{auxi}})$, with $U$ retained as the auxiliary patch. In practice, the protocol uses this two-layer PCM directly, while the extra $U'$ patch and folding Clifford are needed only for the derivation. Taking $g_0 =g_1^* =0$, such surgery protocols reduce to the same-type construction with CSS PCMs as in Eq.~\eqref{eq: product merged checks}.

  \begin{lemma}[Mixed-type parallel product surgery]
    \label{Subsec: limitations}
    Writing $P(x,z)$ for a Pauli with symplectic supports $x,z$, the mixed-type connection above measures exactly the logical space
    \[
      \left\{\left[P\bigl(\BB(g_1)s,\BB({H_Z'}^*)s\bigr)\right]:
      s\in\ker\BB(H_Z^*)\right\}.
    \]
    Brackets denote classes modulo data stabilisers and phases. Its binary dimension is the number $M$ of independent PPMs.
  \end{lemma}
  \begin{proof}
    See Appendix~\ref{Appendix: mixed Pauli surgery}.
  \end{proof}

\subsubsection{Parallel surgery for mixed-type measurements of CC codes}
  \label{Subsec: mixed Pauli surgery}

  For a CC code $\QQQ=\ccab$, choose the $X$ attachment $g_1$ with entries in $\{0,1\}\subset R$ and a dressing $g_0$ over $R$ such that Eq.~\eqref{eq: mixed folding conditions for product codes} holds. By Corollary~\ref{Corollary: ker and im of Ha},
  \[
    \ker\BB(\partial_1)=\ker\BB(H_Z^*)=\im\BB(\chi I_{n_an_b}).
  \]
  Each row of $(g_1^*\mid H_Z')$ therefore specifies the $X$ and $Z$ components of a logical product in the clustered basis. The product of its $l$ lifted merge checks annihilates the auxiliary support and gives this logical PPM. The measured logical space is precisely the row span of this matrix with binary entries, so
  \[
    M=\rank(g_1^*\mid H_Z')\leq n_an_b=k/2,
  \]
  where the rank is again taken before lifting and the maximum is achieved whenever $g_1$ or $H_Z'$ has full rank. Compatible $g_0$ dressings exist by Appendix~\ref{Appendix: mixed Pauli surgery}, while Appendix~\ref{Appendix: sparse mixed surgery} lists the achieved finite-code check weights.

  Write $(p_X,p_Z)$ and $(q_X,q_Z)$ for the symplectic representations of $P$ and $Q$, with $X=(1,0)$, $Y=(1,1)$ and $Z=(0,1)$.

  \begin{corollary}[Maximally parallel Pauli-pair measurements]
    \label{Corollary: mixed pair merges}
    For each fixed Pauli-pair type
    \[
      PQ\in \{XX, XY, XZ, YX, YY, YZ, ZX, ZY, ZZ\},
    \]
    one product surgery round measures precisely the $k/2$ disjoint products $\overline P_i^{\lef}\overline Q_i^{\rig}$, $i=1,\ldots,n_an_b$ from a CC code $\ccab$ encoding $k=2n_an_b$ logical qubits.
  \end{corollary}
  \begin{proof}
    With $I=I_{n_an_b}$, take $H_a'=q_ZI_{n_a},\ H_b'=p_ZI_{n_b}$ and $g_1^*=(p_XI\mid q_XI)$. Then $H_Z'=(p_ZI\mid q_ZI)$, and $(g_1^*\mid H_Z')$ has $k/2$ independent commuting rows specifying these products. Compatible $g_0$ dressings are given in Appendix~\ref{Appendix: mixed Pauli surgery}.
  \end{proof}

  \begin{remark}[General mixed-type attachments]
    \label{remark: general attachment}
    Allowing $f_0,f_1$ beyond $f_0=H_X'$ and $f_1^*=H_Z'$ gives additional flexibility such that any commuting family of at most $k/2$ logical PPMs can be measured in one surgery round while retaining the same auxiliary product complex.

    For example, to measure only the addressed logical $\overline{Y}_j$, let $e_j\in\FF_2^k$ select its logical coordinate and let $e_r\in\FF_2^{n_an_b}$ select any one attachment row. Choose
    \[
      g_1^*=f_1^*=e_re_j^\intercal.
    \]
    The measured logical space then has exactly dimension one and is generated by $\overline{Y}_j$. Compatible $f_0,g_0$ exist by Appendix~\ref{Appendix: mixed Pauli surgery}, while Appendix~\ref{Appendix: sparse mixed surgery} lists the achieved finite-code check weights.
  \end{remark}

\subsection{Parallel product surgery procedure}
\label{Subsec: surgery procedure}
Assume that the data patch is in the codespace of $\QQQ$. Choose attachment and dressing maps satisfying Eq.~\eqref{eq: mixed folding conditions}, with logical measurement space equal to the requested span. A surgery round consists of the following three stages, where all matrices are understood in the binary representation:
\begin{enumerate}
\item \textbf{Prepare:} Initialise the $n$ auxiliary data qubits in $\ket{+}^{\otimes n}$. The initial $X$ and $Z$ stabilisers are generated by $\mathrm{diag}(H_X,I)$ and $\begin{pmatrix}H_Z&0\end{pmatrix}$, respectively.
\item \textbf{Merge:} Measure the stabilisers of the merged code in Eq.~\eqref{eq: mixed product merged checks} or Remark~\ref{remark: general attachment}. Its merge generators are
  \[
    G_j=P_j\otimes Z_{\mathrm{auxi}}^{H_Z^*e_j},
  \]
  where $P_j=P(g_1e_j,f_1e_j)$ and $e_j$ indexes the columns of $H_Z^*$. Each $G_j$ is measured as one joint Pauli observable, including its $Y$ components. Run syndrome extraction for the merged code for a total of $d$ cycles, where $d$ is the data-code distance.

  Write $(-1)^{b_j}$ for the merge outcome and set $P(s)=\prod_jP_j^{s_j}$. For every $u\in\ker H_Z^*$, the auxiliary support cancels, giving
  \[
    \prod_jG_j^{u_j}=P(u)\otimes I
    \tand \text{outcome }(-1)^{b^\intercal u}.
  \]
\item \textbf{Split:} Measure the auxiliary data qubits transversally in the $X$ basis, with outcome bits $t$, and resume syndrome extraction of $\QQQ$. For an ideal record, $t\in\im H_Z^*$; choose any $s_0$ satisfying $H_Z^*s_0=t$ and absorb the measurement byproduct $P(s_0)$ into the data Pauli frame (Appendix~\ref{Appendix: mixed Pauli surgery}). With noise, decode the auxiliary readout together with the syndrome history across the merge and split to infer the logical outcomes and the data Pauli-frame update~\cite{Gauging_logical,CowtanHeWilliamsonYoder2025}.
\end{enumerate}

Same-type product surgery is recovered by taking $g_1=g_0=0$, $f_1={H_Z'}^*$ and $f_0=H_X'$, with new $Z$-type stabilisers $\begin{pmatrix}H_Z'&H_Z\end{pmatrix}$. Exchanging $X$ and $Z$ throughout gives the corresponding $X$-type procedure, with $\ket{0}^{\otimes n}$ preparation and transversal $Z$ readout.

\subsection{Fault tolerance of parallel product surgery}
\label{Section: FT}

A surgery procedure is fault-tolerant if (i) all homomorphisms appearing in Eq.~\eqref{eq: CSS surgery} or Eq.~\eqref{eq: mixed product merged checks} are LDPC, and (ii) the procedure preserves the phenomenological distance at $\Omega(d)$. The first condition ensures bounded-degree circuit connectivity, thereby limiting error propagation under circuit-level noise. The second guarantees that error protection during surgery remains proportional to the distance $d$ of the underlying code~\cite{Gottesman1997,Qian_high_rate}.

For same-type product surgery, Eq.~\eqref{eq: product merged checks} ensures that the merged code remains LDPC whenever both $\QQQ$ and $\PPP$ are LDPC. Specifically, for CC codes, the maximum weight of the merged code is upper bounded by $w+r$, where $r$ is the number of logical qubits included in one PPM. For mixed-type surgery of CC codes, all finite constructions presented in Appendix~\ref{Appendix: sparse mixed surgery} have maximum check weight at most $13$; their individual check weights are given in Table~\ref{Table: mixed surgery weights}.

For HGP codes, Appendix~\ref{Appendix: HGP distance preserving} proves that same-type product surgery preserves both $d_X$ and $d_Z$. For CC codes, the proof has two parts. We first identify the $k-M$ surviving data logical qubits and show that the remaining auxiliary logical degrees of freedom act trivially on this data subsystem, so they can be assigned to gauge. We then use \textit{weak expansion} to show that every undetected error of weight below $d$ is equivalent, modulo merged stabilisers, to an operator on these gauges.

It suffices to cancel each auxiliary $Z$ stabiliser using a merge product whose data weight is no greater than the boundary weight. In the binary representation, define
\[
\begin{aligned}
  w_{\mathrm{att}}&=\max_j\wt(P_j),&
  \rho(s)&=\min_{u\in\ker H_Z^*}\wt(s+u).
\end{aligned}
\]
Here $w_{\mathrm{att}}$ is the maximum data weight of a merge generator, counting $Y$ once, and $\rho(s)$ is the minimum number of merge generators producing the stabiliser $H_Z^*s$. A sufficient weak expansion condition is, for every $s$,
\begin{equation}
\label{eq: weak expansion}
\wt(H_Z^*s)<d
\quad\Longrightarrow\quad
w_{\mathrm{att}}\rho(s)\leq\wt(H_Z^*s).
\end{equation}
Only stabilisers below the target distance are constrained.

\begin{theorem}[Protected distance of parallel product surgery]
\label{Theorem: mixed protected distance}
Let $\QQQ$ be a CC code of distance at least $d$. Consider the general attachment form of Eq.~\eqref{eq: mixed product merged checks} and Remark~\ref{remark: general attachment} measuring $M$ independent logical PPMs. Take the gauge group generated by the merged stabilisers and their auxiliary-only Pauli normalisers. If Eq.~\eqref{eq: weak expansion} holds, the merged code has $k-M$ protected data logical qubits with distance at least $d$.
\end{theorem}
\begin{proof}
Appendix~\ref{Appendix: weak expansion} establishes the subsystem identification and the distance bound.
\end{proof}

\subsection{Clifford operations from parallel product surgery}
\label{Section: clifford}

The general attachment rows with binary entries in Remark~\ref{remark: general attachment} allow any chosen commuting family of at most $k/2$ independent logical PPMs in one surgery round. We use addressed single- and two-logical measurements, Pauli-state preparation and Pauli-frame tracking to construct Clifford gates~\cite{Gottesman1997,Game_of_surface_code,NielsenChuang}, with a subset of logical qubits treated as auxiliary. The stated round counts below assume prepared auxiliary logical states, while preparation and resetting are counted separately.

\begin{theorem}[Clifford operations with one auxiliary logical qubit]
\label{Theorem: one auxiliary clifford}
A CC code encoding $k=2n_an_b$ logical qubits supports the full Clifford group on $k-1$ data logical qubits using one reusable auxiliary logical qubit. Each addressed $\bhgate$ or $\bsgate$ gate takes two surgery rounds, and a $\bcnot$ between any two data logical qubits takes three. Every corrected measurement branch implements the intended gate, without measuring an unwanted data observable.
\end{theorem}
\begin{proof}
For a single-qubit gate, use the standard Pauli-rotation protocol~\cite[Fig.~11(b)]{Game_of_surface_code}, with an equivalent auxiliary Pauli-basis convention: prepare $\bket{0}_{\mathrm{auxi}}$ and measure
\[
  \overline P_{\mathrm{data}}\bX_{\mathrm{auxi}}\tand
  \overline Y_{\mathrm{auxi}},\text{ where } P\in\{Y,Z\}.
\]
The correction $\overline P_{\mathrm{data}}^{b+s}$ gives $\exp(-i\pi\overline P_{\mathrm{data}}/4)$ for outcomes $b,s$. Choosing $P=Z$ gives $\bsgate$; choosing $P=Y$ and adding $\bX_{\mathrm{data}}$ to the Pauli frame gives $\bhgate$, up to global phase.

For $\bcnot_{c\to t}$, use the $X/Z$-exchanged form of the standard parity-measurement construction~\cite[Fig.~3]{LitinskiVonOppen2018}: prepare $\bket{0}_{\mathrm{auxi}}$ and measure
\[
  \bX_{\mathrm{auxi}}\bX_t,\ 
  \bZ_c\bZ_{\mathrm{auxi}},\tand
  \bX_{\mathrm{auxi}}.
\]
For outcome bits $a,b,x$, apply $\bX_t^b\bZ_c^{a+x}$. These circuit identities give a corrected branch proportional to the desired unitary, so the full outcome record reveals no information about the input data state. Resetting and reusing the same auxiliary logical qubit gives these generators sequentially on the remaining $k-1$ logical qubits, and hence their full Clifford group.
\end{proof}

Assigning more logical qubits to auxiliaries allows these gates to run in parallel. For the following allocation, the left sector carries data logical qubits and the right sector carries auxiliary logical qubits, with $\overline P_{\mathrm{data},i}=\overline P_i^{\lef}$ and $\overline P_{\mathrm{auxi},i}=\overline P_i^{\rig}$.

\begin{theorem}[Maximally parallel Clifford operations]
\label{Theorem: full clifford group}
Taking $k/2$ logical qubits as auxiliaries allows the following parallel operations on the remaining $k/2$ data logical qubits:
\begin{enumerate}
  \item Any selected combination of independently chosen $\bhgate$ and $\bsgate$ gates, including a configuration on all $k/2$ data logical qubits, takes two surgery rounds;
  \item A disjoint layer containing $c$ directed $\bcnot$ gates and $s$ single-qubit generators, with $s+2c\leq k/2$ and $c>0$, takes three surgery rounds. In particular, up to $\lfloor k/4\rfloor$ disjoint $\bcnot$ gates can be performed in parallel, with their directions chosen independently.
\end{enumerate}
\end{theorem}
\begin{proof}
Assign one distinct auxiliary logical qubit to each active gate. The condition $s+2c\leq k/2$ ensures that the required $s+c$ auxiliaries are available. In rounds 1 and 2, perform the first and second measurements of all the preceding protocols, respectively. Each round contains $s+c\leq k/2$ independent commuting PPMs and is available by Remark~\ref{remark: general attachment}. When $c>0$, a third round measures $\bX_{\mathrm{auxi}}$ for the $c$ auxiliaries used for $\bcnot$ gates.

After Pauli-frame corrections, every branch is proportional to the tensor product of the desired gates, also for entangled data inputs. Its probability is therefore independent of the input state, so no unwanted data observable is measured.
\end{proof}

Preparing or resetting the selected auxiliary logical qubits takes one separate round of addressed $Z$ measurements and Pauli-frame updates. The physical auxiliary patch is reused between rounds. Every preparation and gate round has $w_{\mathrm{att}}\leq2$, so Theorem~\ref{Theorem: mixed protected distance} applies whenever the weak expansion condition holds with factor two.

\subsection{Parallel magic-state preparation on clustered logical qubits}

The clustered logical basis also provides a direct route to preparing magic states at selected logical addresses. To initialise a CC block, we first encode each physical input $\alpha_i\ket{0}+\beta_i\ket{1}$ into a cat state on cluster $i$, using $l-1$ additional qubits in $\ket{0}$ and $l-1$ $\mathit{CNOT}$ gates. Measuring the $X$-type memory stabilisers and applying syndrome-dependent $Z$ corrections then implements the map
\[
  \bigotimes_{i=1}^{k}\left(\alpha_i\ket{0}^{\otimes l}+\beta_i\ket{1}^{\otimes l}\right)
  \longmapsto
  \bigotimes_{i=1}^{k}\left(\alpha_i\bket{0}_i+\beta_i\bket{1}_i\right),
\]
where the output tensor product is over logical qubits. Every memory stabiliser has even support in each cluster, so all $Z$-type stabilisers are already satisfied by the cat states. Since $l$ is odd, no nontrivial product of $X$-type stabilisers flips only complete clusters. The syndrome outcomes are therefore independent of the input states. Moreover, the even-parity subspace of each cluster is nondegenerate under the binary inner product. The $Z$ corrections can thus be chosen with even support in every cluster, so that they commute with all clustered logical operators.

Taking $\alpha_i=1/\sqrt{2}$ and $\beta_i=e^{i\pi/4}/\sqrt{2}$ prepares $\bket{T}_i$ at logical address $i$, with the inputs at other addresses chosen independently. All $k$ states can be encoded in parallel using the memory data qubits and the usual check auxiliaries.

These raw logical magic states can then be distilled using Clifford operations and Pauli measurements~\cite{PhysRevA.71.022316}. Across copies of the same CC code, transversal $\mathit{CNOT}$ gates couple corresponding logical addresses. Alternatively, the inputs and auxiliary logical qubits can occupy selected addresses within one block, with the distillation circuit implemented through the Clifford constructions above. Compatible measurements can be grouped into parallel product-surgery rounds, allowing the logical allocation to determine the available parallelism.

\clearpage
\section*{Acknowledgements}
We acknowledge helpful discussions with Robert I. Booth, Alexander Cowtan, Wang Fang, Shouzhen (Bailey) Gu, Zhiyang (Sunny) He, Shilin Huang, Benjamin Ide, Yingjia Lin, Pavel Panteleev, Clément Poirson, Kaavya Sahay, Hasan Sayginel, Susan J. Sierra, Adithya Sireesh, and Ryan Tiew; and Kyano Levi, Lukas Sieberer and Yangyang Tang for pointing out typos or errors in the early manuscript. We especially thank Stergios Koutsioumpas and Tamas Noszko for discussions and insightful comments.

\section*{Funding}
BG was supported by the Higgs Scholarship funded by the School of Physics and Astronomy, University of Edinburgh.
BG and JR were supported by the project `Algorithmic fault-tolerance: quantum error correction beyond quantum memories' funded by the UKRI Quantum Computing and Simulation (QCS) Hub (grant code: EP/T001062/1). BG and JR were also funded by the Innovate UK project \enquote{QEC Readout Testbed} [reference number 10151107].
JR is currently funded by an EPSRC Quantum Career Acceleration Fellowship (grant code: UKRI1224). QX is funded in part by the Walter Burke Institute for Theoretical Physics at Caltech.
The Berlin team has been supported by the BMFTR (QSolid, MUNIQC-Atoms, PasQuops),
the DFG (CRC 183 and SPP 2514), the Quantum Flagship (Millenion, PasQuanS2), the Munich Quantum Valley, Berlin Quantum and the European Research Council (DebuQC).

\section*{Competing interests}
The authors declare no competing interests.

\section*{Data Availability}
The PCMs and simulation circuits for the CC codes studied in this paper are available on Zenodo: \href{https://doi.org/10.5281/zenodo.24647854}{https://doi.org/10.5281/zenodo.24647854}. The archive contains the memory PCMs and all 579 merged PCMs for the 13 CC codes, together with the $X$- and $Z$-basis Stim circuits used in the memory and logical-gate simulations. Logical operators, construction maps, and documentation of the noise model and measurement conventions are also provided.

\section*{Note on concurrent work}
Logical computation with high-rate qLDPC codes has progressed quickly since October 2025, and we briefly note some related works here.

Refs.~\cite{HirasakiLee2026,Lee2026LogicalSpectroscopy} build on the clustered logical basis of CC codes, extending it to lifted CC codes and to general lifted-product codes. Refs.~\cite{ZhengZhengJiangXu2026,ZhouGuEisertZhao2026,BerentCohenQuintavalle2026,MenonEtAl2026MagicFactories}, co-authored by some of us, generalise the clustered basis to arbitrary Abelian groups, adapt product surgery to constant-depth shuttling on neutral atoms, introduce lifted surgery for parallel, addressable measurements on group-algebra codes, and use clustered logical bases in single-shot magic-state factories. Individually addressed logical measurements by qLDPC surgery have also been demonstrated on trapped-ion hardware~\cite{LuthraOwens2026}.

Refs.~\cite{Qian_high_rate,BaspinBerentCohen2025,CowtanHeWilliamsonYoder2025,YuanCowtanHeLinWilliamson2026,ChangHeYoderZhuJochymOConnor2026,qian_batched,GolowichChangZhu2025,PecorariGuerciPerrinPupillo2025} develop fast, parallel or high-rate surgery for general qLDPC codes; we complement these with an addressable logical basis and taking another copy of the data code as the auxiliary patch. Ref.~\cite{TansuwannontChanTakagi2026} obtains Clifford gates from transversal and fold-transversal structure; we complement this with surgery-based Clifford schedules that apply across the CC family.

\bibliography{compiled_references}

\begin{thebibliography}{88}%
\makeatletter
\providecommand \@ifxundefined [1]{%
 \@ifx{#1\undefined}
}%
\providecommand \@ifnum [1]{%
 \ifnum #1\expandafter \@firstoftwo
 \else \expandafter \@secondoftwo
 \fi
}%
\providecommand \@ifx [1]{%
 \ifx #1\expandafter \@firstoftwo
 \else \expandafter \@secondoftwo
 \fi
}%
\providecommand \natexlab [1]{#1}%
\providecommand \enquote  [1]{``#1''}%
\providecommand \bibnamefont  [1]{#1}%
\providecommand \bibfnamefont [1]{#1}%
\providecommand \citenamefont [1]{#1}%
\providecommand \href@noop [0]{\@secondoftwo}%
\providecommand \href [0]{\begingroup \@sanitize@url \@href}%
\providecommand \@href[1]{\@@startlink{#1}\@@href}%
\providecommand \@@href[1]{\endgroup#1\@@endlink}%
\providecommand \@sanitize@url [0]{\catcode `\\12\catcode `\$12\catcode
  `\&12\catcode `\#12\catcode `\^12\catcode `\_12\catcode `\%12\relax}%
\providecommand \@@startlink[1]{}%
\providecommand \@@endlink[0]{}%
\providecommand \url  [0]{\begingroup\@sanitize@url \@url }%
\providecommand \@url [1]{\endgroup\@href {#1}{\urlprefix }}%
\providecommand \urlprefix  [0]{URL }%
\providecommand \Eprint [0]{\href }%
\providecommand \doibase [0]{https://doi.org/}%
\providecommand \selectlanguage [0]{\@gobble}%
\providecommand \bibinfo  [0]{\@secondoftwo}%
\providecommand \bibfield  [0]{\@secondoftwo}%
\providecommand \translation [1]{[#1]}%
\providecommand \BibitemOpen [0]{}%
\providecommand \bibitemStop [0]{}%
\providecommand \bibitemNoStop [0]{.\EOS\space}%
\providecommand \EOS [0]{\spacefactor3000\relax}%
\providecommand \BibitemShut  [1]{\csname bibitem#1\endcsname}%
\let\auto@bib@innerbib\@empty
\bibitem [{\citenamefont {Kitaev}(2003)}]{Kitaev2003}%
  \BibitemOpen
  \bibfield  {author} {\bibinfo {author} {\bibfnamefont {A.~Y.}\ \bibnamefont
  {Kitaev}},\ }\bibfield  {title} {\bibinfo {title} {{Fault-tolerant quantum
  computation by anyons}},\ }\href
  {https://doi.org/10.1016/S0003-4916(02)00018-0} {\bibfield  {journal}
  {\bibinfo  {journal} {Ann. Phys.}\ }\textbf {\bibinfo {volume} {303}},\
  \bibinfo {pages} {2} (\bibinfo {year} {2003})}\BibitemShut {NoStop}%
\bibitem [{\citenamefont {Gottesman}(2022)}]{GottesmanNewReview}%
  \BibitemOpen
  \bibfield  {author} {\bibinfo {author} {\bibfnamefont {D.}~\bibnamefont
  {Gottesman}},\ }\bibfield  {title} {\bibinfo {title} {{Opportunities and
  challenges in fault-tolerant quantum computation}},\ }\href@noop {} {\
  (\bibinfo {year} {2022})},\ \Eprint {https://arxiv.org/abs/2210.15844}
  {arXiv:2210.15844} \BibitemShut {NoStop}%
\bibitem [{\citenamefont {Roffe}(2019)}]{Roffe_2019}%
  \BibitemOpen
  \bibfield  {author} {\bibinfo {author} {\bibfnamefont {J.}~\bibnamefont
  {Roffe}},\ }\bibfield  {title} {\bibinfo {title} {{Quantum error correction:
  an introductory guide}},\ }\href
  {https://doi.org/10.1080/00107514.2019.1667078} {\bibfield  {journal}
  {\bibinfo  {journal} {Contemp. Phys.}\ }\textbf {\bibinfo {volume} {60}},\
  \bibinfo {pages} {226} (\bibinfo {year} {2019})}\BibitemShut {NoStop}%
\bibitem [{\citenamefont {Campbell}\ \emph {et~al.}(2017)\citenamefont
  {Campbell}, \citenamefont {Terhal},\ and\ \citenamefont {Vuillot}}]{Roads}%
  \BibitemOpen
  \bibfield  {author} {\bibinfo {author} {\bibfnamefont {E.~T.}\ \bibnamefont
  {Campbell}}, \bibinfo {author} {\bibfnamefont {B.~M.}\ \bibnamefont
  {Terhal}},\ and\ \bibinfo {author} {\bibfnamefont {C.}~\bibnamefont
  {Vuillot}},\ }\bibfield  {title} {\bibinfo {title} {{Roads towards
  fault-tolerant universal quantum computation}},\ }\href
  {https://doi.org/10.1038/nature23460} {\bibfield  {journal} {\bibinfo
  {journal} {Nature}\ }\textbf {\bibinfo {volume} {549}},\ \bibinfo {pages}
  {172} (\bibinfo {year} {2017})}\BibitemShut {NoStop}%
\bibitem [{\citenamefont {Terhal}(2015)}]{RevModPhys.87.307}%
  \BibitemOpen
  \bibfield  {author} {\bibinfo {author} {\bibfnamefont {B.~M.}\ \bibnamefont
  {Terhal}},\ }\bibfield  {title} {\bibinfo {title} {{Quantum error correction
  for quantum memories}},\ }\href {https://doi.org/10.1103/RevModPhys.87.307}
  {\bibfield  {journal} {\bibinfo  {journal} {Rev. Mod. Phys.}\ }\textbf
  {\bibinfo {volume} {87}},\ \bibinfo {pages} {307} (\bibinfo {year}
  {2015})}\BibitemShut {NoStop}%
\bibitem [{\citenamefont {Gottesman}(2014)}]{Gottesman}%
  \BibitemOpen
  \bibfield  {author} {\bibinfo {author} {\bibfnamefont {D.}~\bibnamefont
  {Gottesman}},\ }\bibfield  {title} {\bibinfo {title} {{Fault-tolerant quantum
  computation with constant overhead}},\ }\href
  {https://doi.org/10.26421/QIC14.15-16-5} {\bibfield  {journal} {\bibinfo
  {journal} {Quantum Inf. Comput.}\ }\textbf {\bibinfo {volume} {14}},\
  \bibinfo {pages} {1338} (\bibinfo {year} {2014})}\BibitemShut {NoStop}%
\bibitem [{\citenamefont {Nielsen}\ and\ \citenamefont
  {Chuang}(2000)}]{NielsenChuang}%
  \BibitemOpen
  \bibfield  {author} {\bibinfo {author} {\bibfnamefont {M.~A.}\ \bibnamefont
  {Nielsen}}\ and\ \bibinfo {author} {\bibfnamefont {I.~L.}\ \bibnamefont
  {Chuang}},\ }\href@noop {} {\emph {\bibinfo {title} {{Quantum computation and
  quantum information}}}}\ (\bibinfo  {publisher} {Cambridge University
  Press},\ \bibinfo {year} {2000})\BibitemShut {NoStop}%
\bibitem [{\citenamefont {Eisert}\ and\ \citenamefont
  {Preskill}(2025)}]{MindTheGaps}%
  \BibitemOpen
  \bibfield  {author} {\bibinfo {author} {\bibfnamefont {J.}~\bibnamefont
  {Eisert}}\ and\ \bibinfo {author} {\bibfnamefont {J.}~\bibnamefont
  {Preskill}},\ }\bibfield  {title} {\bibinfo {title} {{Mind the gaps: the
  fraught road to quantum advantage}},\ }\href@noop {} {\  (\bibinfo {year}
  {2025})},\ \Eprint {https://arxiv.org/abs/2510.19928} {arXiv:2510.19928}
  \BibitemShut {NoStop}%
\bibitem [{\citenamefont {Tillich}\ and\ \citenamefont
  {Z{\'e}mor}(2014)}]{TillichZemor}%
  \BibitemOpen
  \bibfield  {author} {\bibinfo {author} {\bibfnamefont {J.-P.}\ \bibnamefont
  {Tillich}}\ and\ \bibinfo {author} {\bibfnamefont {G.}~\bibnamefont
  {Z{\'e}mor}},\ }\bibfield  {title} {\bibinfo {title} {{Quantum LDPC codes
  with positive rate and minimum distance proportional to $n^{1/2}$}},\ }\href
  {https://doi.org/10.1109/TIT.2013.2292061} {\bibfield  {journal} {\bibinfo
  {journal} {IEEE Trans. Inf. Theory}\ }\textbf {\bibinfo {volume} {60}},\
  \bibinfo {pages} {1193} (\bibinfo {year} {2014})}\BibitemShut {NoStop}%
\bibitem [{\citenamefont {Panteleev}\ and\ \citenamefont
  {Kalachev}(2021)}]{Panteleev}%
  \BibitemOpen
  \bibfield  {author} {\bibinfo {author} {\bibfnamefont {P.}~\bibnamefont
  {Panteleev}}\ and\ \bibinfo {author} {\bibfnamefont {G.}~\bibnamefont
  {Kalachev}},\ }\bibfield  {title} {\bibinfo {title} {{Degenerate quantum LDPC
  codes with good finite length performance}},\ }\href
  {https://doi.org/10.22331/q-2021-11-22-585} {\bibfield  {journal} {\bibinfo
  {journal} {Quantum}\ }\textbf {\bibinfo {volume} {5}},\ \bibinfo {pages}
  {585} (\bibinfo {year} {2021})}\BibitemShut {NoStop}%
\bibitem [{\citenamefont {Breuckmann}\ and\ \citenamefont
  {Eberhardt}(2021{\natexlab{a}})}]{LDPCReview}%
  \BibitemOpen
  \bibfield  {author} {\bibinfo {author} {\bibfnamefont {N.~P.}\ \bibnamefont
  {Breuckmann}}\ and\ \bibinfo {author} {\bibfnamefont {J.~N.}\ \bibnamefont
  {Eberhardt}},\ }\bibfield  {title} {\bibinfo {title} {{Quantum low-density
  parity-check codes}},\ }\href {https://doi.org/10.1103/PRXQuantum.2.040101}
  {\bibfield  {journal} {\bibinfo  {journal} {PRX Quantum}\ }\textbf {\bibinfo
  {volume} {2}},\ \bibinfo {pages} {040101} (\bibinfo {year}
  {2021}{\natexlab{a}})}\BibitemShut {NoStop}%
\bibitem [{\citenamefont {Leverrier}\ and\ \citenamefont
  {Z{\'e}mor}(2022)}]{LeverrierLDPC}%
  \BibitemOpen
  \bibfield  {author} {\bibinfo {author} {\bibfnamefont {A.}~\bibnamefont
  {Leverrier}}\ and\ \bibinfo {author} {\bibfnamefont {G.}~\bibnamefont
  {Z{\'e}mor}},\ }\bibfield  {title} {\bibinfo {title} {{Quantum Tanner
  codes}},\ }in\ \href {https://doi.org/10.1109/FOCS54457.2022.00117} {\emph
  {\bibinfo {booktitle} {2022 IEEE 63rd Annual Symposium on Foundations of
  Computer Science (FOCS)}}}\ (\bibinfo {year} {2022})\ pp.\ \bibinfo {pages}
  {872--883}\BibitemShut {NoStop}%
\bibitem [{\citenamefont {Panteleev}\ and\ \citenamefont
  {Kalachev}(2022{\natexlab{a}})}]{Panteleev_almost_linear}%
  \BibitemOpen
  \bibfield  {author} {\bibinfo {author} {\bibfnamefont {P.}~\bibnamefont
  {Panteleev}}\ and\ \bibinfo {author} {\bibfnamefont {G.}~\bibnamefont
  {Kalachev}},\ }\bibfield  {title} {\bibinfo {title} {{Quantum LDPC codes with
  almost linear minimum distance}},\ }\href
  {https://doi.org/10.1109/TIT.2021.3119384} {\bibfield  {journal} {\bibinfo
  {journal} {IEEE Trans. Inf. Theory}\ }\textbf {\bibinfo {volume} {68}},\
  \bibinfo {pages} {213} (\bibinfo {year} {2022}{\natexlab{a}})}\BibitemShut
  {NoStop}%
\bibitem [{\citenamefont {Panteleev}\ and\ \citenamefont
  {Kalachev}(2022{\natexlab{b}})}]{Panteleev_good_and_local_testable}%
  \BibitemOpen
  \bibfield  {author} {\bibinfo {author} {\bibfnamefont {P.}~\bibnamefont
  {Panteleev}}\ and\ \bibinfo {author} {\bibfnamefont {G.}~\bibnamefont
  {Kalachev}},\ }\bibfield  {title} {\bibinfo {title} {{Asymptotically good
  quantum and locally testable classical LDPC codes}},\ }in\ \href
  {https://doi.org/10.1145/3519935.3520017} {\emph {\bibinfo {booktitle}
  {Proceedings of the 54th Annual ACM SIGACT Symposium on Theory of
  Computing}}}\ (\bibinfo {year} {2022})\ pp.\ \bibinfo {pages}
  {375--388}\BibitemShut {NoStop}%
\bibitem [{\citenamefont {Acharya}\ \emph {et~al.}(2025)\citenamefont
  {Acharya}, \citenamefont {Abanin}, \citenamefont {Aghababaie-Beni} \emph
  {et~al.}}]{acharya_quantum_2024}%
  \BibitemOpen
  \bibfield  {author} {\bibinfo {author} {\bibfnamefont {R.}~\bibnamefont
  {Acharya}}, \bibinfo {author} {\bibfnamefont {D.~A.}\ \bibnamefont {Abanin}},
  \bibinfo {author} {\bibfnamefont {L.}~\bibnamefont {Aghababaie-Beni}}, \emph
  {et~al.},\ }\bibfield  {title} {\bibinfo {title} {{Quantum error correction
  below the surface code threshold}},\ }\href
  {https://doi.org/10.1038/s41586-024-08449-y} {\bibfield  {journal} {\bibinfo
  {journal} {Nature}\ }\textbf {\bibinfo {volume} {638}},\ \bibinfo {pages}
  {920} (\bibinfo {year} {2025})}\BibitemShut {NoStop}%
\bibitem [{\citenamefont {Erhard}\ \emph {et~al.}(2021)\citenamefont {Erhard},
  \citenamefont {Poulsen~Nautrup}, \citenamefont {Meth} \emph
  {et~al.}}]{ErhardPoulsenNautrupMethEtAl2021}%
  \BibitemOpen
  \bibfield  {author} {\bibinfo {author} {\bibfnamefont {A.}~\bibnamefont
  {Erhard}}, \bibinfo {author} {\bibfnamefont {H.}~\bibnamefont
  {Poulsen~Nautrup}}, \bibinfo {author} {\bibfnamefont {M.}~\bibnamefont
  {Meth}}, \emph {et~al.},\ }\bibfield  {title} {\bibinfo {title} {{Entangling
  logical qubits with lattice surgery}},\ }\href
  {https://doi.org/10.1038/s41586-020-03079-6} {\bibfield  {journal} {\bibinfo
  {journal} {Nature}\ }\textbf {\bibinfo {volume} {589}},\ \bibinfo {pages}
  {220} (\bibinfo {year} {2021})}\BibitemShut {NoStop}%
\bibitem [{\citenamefont {Bluvstein}\ \emph {et~al.}(2024)\citenamefont
  {Bluvstein}, \citenamefont {Evered}, \citenamefont {Geim} \emph
  {et~al.}}]{BluvsteinEtAl2024}%
  \BibitemOpen
  \bibfield  {author} {\bibinfo {author} {\bibfnamefont {D.}~\bibnamefont
  {Bluvstein}}, \bibinfo {author} {\bibfnamefont {S.~J.}\ \bibnamefont
  {Evered}}, \bibinfo {author} {\bibfnamefont {A.~A.}\ \bibnamefont {Geim}},
  \emph {et~al.},\ }\bibfield  {title} {\bibinfo {title} {{Logical quantum
  processor based on reconfigurable atom arrays}},\ }\href
  {https://doi.org/10.1038/s41586-023-06927-3} {\bibfield  {journal} {\bibinfo
  {journal} {Nature}\ }\textbf {\bibinfo {volume} {626}},\ \bibinfo {pages}
  {58} (\bibinfo {year} {2024})}\BibitemShut {NoStop}%
\bibitem [{\citenamefont {Putterman}\ \emph {et~al.}(2025)\citenamefont
  {Putterman}, \citenamefont {Noh}, \citenamefont {Hann} \emph
  {et~al.}}]{PuttermanNohHannEtAl2025}%
  \BibitemOpen
  \bibfield  {author} {\bibinfo {author} {\bibfnamefont {H.}~\bibnamefont
  {Putterman}}, \bibinfo {author} {\bibfnamefont {K.}~\bibnamefont {Noh}},
  \bibinfo {author} {\bibfnamefont {C.~T.}\ \bibnamefont {Hann}}, \emph
  {et~al.},\ }\bibfield  {title} {\bibinfo {title} {{Hardware-efficient quantum
  error correction via concatenated bosonic qubits}},\ }\href
  {https://doi.org/10.1038/s41586-025-08642-7} {\bibfield  {journal} {\bibinfo
  {journal} {Nature}\ }\textbf {\bibinfo {volume} {638}},\ \bibinfo {pages}
  {927} (\bibinfo {year} {2025})}\BibitemShut {NoStop}%
\bibitem [{\citenamefont {Eastin}\ and\ \citenamefont
  {Knill}(2009)}]{Eastin_Knill_restrictions_transversal}%
  \BibitemOpen
  \bibfield  {author} {\bibinfo {author} {\bibfnamefont {B.}~\bibnamefont
  {Eastin}}\ and\ \bibinfo {author} {\bibfnamefont {E.}~\bibnamefont {Knill}},\
  }\bibfield  {title} {\bibinfo {title} {{Restrictions on transversal encoded
  quantum gate sets}},\ }\href {https://doi.org/10.1103/PhysRevLett.102.110502}
  {\bibfield  {journal} {\bibinfo  {journal} {Phys. Rev. Lett.}\ }\textbf
  {\bibinfo {volume} {102}},\ \bibinfo {pages} {110502} (\bibinfo {year}
  {2009})}\BibitemShut {NoStop}%
\bibitem [{\citenamefont {Horsman}\ \emph {et~al.}(2012)\citenamefont
  {Horsman}, \citenamefont {Fowler}, \citenamefont {Devitt} \emph
  {et~al.}}]{HorsmanSurgery}%
  \BibitemOpen
  \bibfield  {author} {\bibinfo {author} {\bibfnamefont {D.}~\bibnamefont
  {Horsman}}, \bibinfo {author} {\bibfnamefont {A.~G.}\ \bibnamefont {Fowler}},
  \bibinfo {author} {\bibfnamefont {S.}~\bibnamefont {Devitt}}, \emph
  {et~al.},\ }\bibfield  {title} {\bibinfo {title} {{Surface code quantum
  computing by lattice surgery}},\ }\href
  {https://doi.org/10.1088/1367-2630/14/12/123011} {\bibfield  {journal}
  {\bibinfo  {journal} {New J. Phys.}\ }\textbf {\bibinfo {volume} {14}},\
  \bibinfo {pages} {123011} (\bibinfo {year} {2012})}\BibitemShut {NoStop}%
\bibitem [{\citenamefont {Litinski}(2019)}]{Game_of_surface_code}%
  \BibitemOpen
  \bibfield  {author} {\bibinfo {author} {\bibfnamefont {D.}~\bibnamefont
  {Litinski}},\ }\bibfield  {title} {\bibinfo {title} {{A game of surface
  codes: large-scale quantum computing with lattice surgery}},\ }\href
  {https://doi.org/10.22331/q-2019-03-05-128} {\bibfield  {journal} {\bibinfo
  {journal} {Quantum}\ }\textbf {\bibinfo {volume} {3}},\ \bibinfo {pages}
  {128} (\bibinfo {year} {2019})}\BibitemShut {NoStop}%
\bibitem [{\citenamefont {Bravyi}\ and\ \citenamefont
  {Kitaev}(2005)}]{PhysRevA.71.022316}%
  \BibitemOpen
  \bibfield  {author} {\bibinfo {author} {\bibfnamefont {S.}~\bibnamefont
  {Bravyi}}\ and\ \bibinfo {author} {\bibfnamefont {A.}~\bibnamefont
  {Kitaev}},\ }\bibfield  {title} {\bibinfo {title} {{Universal quantum
  computation with ideal Clifford gates and noisy ancillas}},\ }\href
  {https://doi.org/10.1103/PhysRevA.71.022316} {\bibfield  {journal} {\bibinfo
  {journal} {Phys. Rev. A}\ }\textbf {\bibinfo {volume} {71}},\ \bibinfo
  {pages} {022316} (\bibinfo {year} {2005})}\BibitemShut {NoStop}%
\bibitem [{\citenamefont {Gidney}\ \emph {et~al.}(2024)\citenamefont {Gidney},
  \citenamefont {Shutty},\ and\ \citenamefont {Jones}}]{magic_cultivation}%
  \BibitemOpen
  \bibfield  {author} {\bibinfo {author} {\bibfnamefont {C.}~\bibnamefont
  {Gidney}}, \bibinfo {author} {\bibfnamefont {N.}~\bibnamefont {Shutty}},\
  and\ \bibinfo {author} {\bibfnamefont {C.}~\bibnamefont {Jones}},\ }\bibfield
   {title} {\bibinfo {title} {{Magic state cultivation: growing T states as
  cheap as CNOT gates}},\ }\href@noop {} {\  (\bibinfo {year} {2024})},\
  \Eprint {https://arxiv.org/abs/2409.17595} {arXiv:2409.17595} \BibitemShut
  {NoStop}%
\bibitem [{\citenamefont {Sales~Rodriguez}\ \emph {et~al.}(2025)\citenamefont
  {Sales~Rodriguez}, \citenamefont {Robinson}, \citenamefont {Jepsen} \emph
  {et~al.}}]{lukin_magic}%
  \BibitemOpen
  \bibfield  {author} {\bibinfo {author} {\bibfnamefont {P.}~\bibnamefont
  {Sales~Rodriguez}}, \bibinfo {author} {\bibfnamefont {J.~M.}\ \bibnamefont
  {Robinson}}, \bibinfo {author} {\bibfnamefont {P.~N.}\ \bibnamefont
  {Jepsen}}, \emph {et~al.},\ }\bibfield  {title} {\bibinfo {title}
  {{Experimental demonstration of logical magic state distillation}},\ }\href
  {https://doi.org/10.1038/s41586-025-09367-3} {\bibfield  {journal} {\bibinfo
  {journal} {Nature}\ }\textbf {\bibinfo {volume} {645}},\ \bibinfo {pages}
  {620} (\bibinfo {year} {2025})}\BibitemShut {NoStop}%
\bibitem [{\citenamefont {Li}\ \emph {et~al.}(2025)\citenamefont {Li},
  \citenamefont {Preskill},\ and\ \citenamefont {Xu}}]{dimensionjump}%
  \BibitemOpen
  \bibfield  {author} {\bibinfo {author} {\bibfnamefont {C.}~\bibnamefont
  {Li}}, \bibinfo {author} {\bibfnamefont {J.}~\bibnamefont {Preskill}},\ and\
  \bibinfo {author} {\bibfnamefont {Q.}~\bibnamefont {Xu}},\ }\bibfield
  {title} {\bibinfo {title} {{Transversal dimension jump for product qLDPC
  codes}},\ }\href@noop {} {\  (\bibinfo {year} {2025})},\ \Eprint
  {https://arxiv.org/abs/2510.07269} {arXiv:2510.07269} \BibitemShut {NoStop}%
\bibitem [{\citenamefont {Webster}\ \emph {et~al.}(2026)\citenamefont
  {Webster}, \citenamefont {Berent}, \citenamefont {Chandra} \emph
  {et~al.}}]{webster2026pinnacle}%
  \BibitemOpen
  \bibfield  {author} {\bibinfo {author} {\bibfnamefont {P.}~\bibnamefont
  {Webster}}, \bibinfo {author} {\bibfnamefont {L.}~\bibnamefont {Berent}},
  \bibinfo {author} {\bibfnamefont {O.}~\bibnamefont {Chandra}}, \emph
  {et~al.},\ }\bibfield  {title} {\bibinfo {title} {{The pinnacle architecture:
  reducing the cost of breaking RSA-2048 to 100 000 physical qubits using
  quantum LDPC codes}},\ }\href@noop {} {\  (\bibinfo {year} {2026})},\ \Eprint
  {https://arxiv.org/abs/2602.11457} {arXiv:2602.11457} \BibitemShut {NoStop}%
\bibitem [{\citenamefont {Cain}\ \emph {et~al.}(2026)\citenamefont {Cain},
  \citenamefont {Xu}, \citenamefont {King} \emph
  {et~al.}}]{cain2026shorsalgorithmpossible10000}%
  \BibitemOpen
  \bibfield  {author} {\bibinfo {author} {\bibfnamefont {M.}~\bibnamefont
  {Cain}}, \bibinfo {author} {\bibfnamefont {Q.}~\bibnamefont {Xu}}, \bibinfo
  {author} {\bibfnamefont {R.}~\bibnamefont {King}}, \emph {et~al.},\
  }\bibfield  {title} {\bibinfo {title} {{Shor's algorithm is possible with as
  few as 10,000 reconfigurable atomic qubits}},\ }\href@noop {} {\  (\bibinfo
  {year} {2026})},\ \Eprint {https://arxiv.org/abs/2603.28627}
  {arXiv:2603.28627} \BibitemShut {NoStop}%
\bibitem [{\citenamefont {Gottesman}(1997)}]{Gottesman1997}%
  \BibitemOpen
  \bibfield  {author} {\bibinfo {author} {\bibfnamefont {D.}~\bibnamefont
  {Gottesman}},\ }\emph {\bibinfo {title} {{Stabilizer codes and quantum error
  correction}}},\ \href@noop {} {Ph.D. thesis},\ \bibinfo  {school} {California
  Institute of Technology} (\bibinfo {year} {1997})\BibitemShut {NoStop}%
\bibitem [{\citenamefont {Xu}\ \emph {et~al.}(2025{\natexlab{a}})\citenamefont
  {Xu}, \citenamefont {Zhou}, \citenamefont {Zheng} \emph
  {et~al.}}]{Qian_fast_and_para}%
  \BibitemOpen
  \bibfield  {author} {\bibinfo {author} {\bibfnamefont {Q.}~\bibnamefont
  {Xu}}, \bibinfo {author} {\bibfnamefont {H.}~\bibnamefont {Zhou}}, \bibinfo
  {author} {\bibfnamefont {G.}~\bibnamefont {Zheng}}, \emph {et~al.},\
  }\bibfield  {title} {\bibinfo {title} {{Fast and parallelizable logical
  computation with homological product codes}},\ }\href
  {https://doi.org/10.1103/PhysRevX.15.021065} {\bibfield  {journal} {\bibinfo
  {journal} {Phys. Rev. X}\ }\textbf {\bibinfo {volume} {15}},\ \bibinfo
  {pages} {021065} (\bibinfo {year} {2025}{\natexlab{a}})}\BibitemShut
  {NoStop}%
\bibitem [{\citenamefont {Cohen}\ \emph {et~al.}(2022)\citenamefont {Cohen},
  \citenamefont {Kim}, \citenamefont {Bartlett} \emph {et~al.}}]{CKBB}%
  \BibitemOpen
  \bibfield  {author} {\bibinfo {author} {\bibfnamefont {L.~Z.}\ \bibnamefont
  {Cohen}}, \bibinfo {author} {\bibfnamefont {I.~H.}\ \bibnamefont {Kim}},
  \bibinfo {author} {\bibfnamefont {S.~D.}\ \bibnamefont {Bartlett}}, \emph
  {et~al.},\ }\bibfield  {title} {\bibinfo {title} {{Low-overhead
  fault-tolerant quantum computing using long-range connectivity}},\ }\href
  {https://doi.org/10.1126/sciadv.abn1717} {\bibfield  {journal} {\bibinfo
  {journal} {Sci. Adv.}\ }\textbf {\bibinfo {volume} {8}},\ \bibinfo {pages}
  {eabn1717} (\bibinfo {year} {2022})}\BibitemShut {NoStop}%
\bibitem [{\citenamefont {Lee}\ \emph {et~al.}(2025)\citenamefont {Lee},
  \citenamefont {Thomsen}, \citenamefont {Fazio} \emph
  {et~al.}}]{color_code_magic_switch}%
  \BibitemOpen
  \bibfield  {author} {\bibinfo {author} {\bibfnamefont {S.-H.}\ \bibnamefont
  {Lee}}, \bibinfo {author} {\bibfnamefont {F.}~\bibnamefont {Thomsen}},
  \bibinfo {author} {\bibfnamefont {N.}~\bibnamefont {Fazio}}, \emph {et~al.},\
  }\bibfield  {title} {\bibinfo {title} {{Low-overhead magic state distillation
  with color codes}},\ }\href {https://doi.org/10.1103/ch5r-cnfq} {\bibfield
  {journal} {\bibinfo  {journal} {PRX Quantum}\ }\textbf {\bibinfo {volume}
  {6}},\ \bibinfo {pages} {030317} (\bibinfo {year} {2025})}\BibitemShut
  {NoStop}%
\bibitem [{\citenamefont {Evered}\ \emph {et~al.}(2026)\citenamefont {Evered},
  \citenamefont {Xu}, \citenamefont {Li} \emph {et~al.}}]{EveredXuLi2026}%
  \BibitemOpen
  \bibfield  {author} {\bibinfo {author} {\bibfnamefont {S.~J.}\ \bibnamefont
  {Evered}}, \bibinfo {author} {\bibfnamefont {M.}~\bibnamefont {Xu}}, \bibinfo
  {author} {\bibfnamefont {S.~H.}\ \bibnamefont {Li}}, \emph {et~al.},\
  }\bibfield  {title} {\bibinfo {title} {{High-fidelity entangling gates and
  nonlocal circuits with neutral atoms}},\ }\href
  {https://doi.org/10.1126/science.aei1312} {\bibfield  {journal} {\bibinfo
  {journal} {Science}\ }\textbf {\bibinfo {volume} {394}},\ \bibinfo {pages}
  {197} (\bibinfo {year} {2026})}\BibitemShut {NoStop}%
\bibitem [{\citenamefont {Dennis}\ \emph {et~al.}(2002)\citenamefont {Dennis},
  \citenamefont {Kitaev}, \citenamefont {Landahl} \emph
  {et~al.}}]{TopologicalQuantumMemory}%
  \BibitemOpen
  \bibfield  {author} {\bibinfo {author} {\bibfnamefont {E.}~\bibnamefont
  {Dennis}}, \bibinfo {author} {\bibfnamefont {A.}~\bibnamefont {Kitaev}},
  \bibinfo {author} {\bibfnamefont {A.}~\bibnamefont {Landahl}}, \emph
  {et~al.},\ }\bibfield  {title} {\bibinfo {title} {{Topological quantum
  memory}},\ }\href {https://doi.org/10.1063/1.1499754} {\bibfield  {journal}
  {\bibinfo  {journal} {J. Math. Phys.}\ }\textbf {\bibinfo {volume} {43}},\
  \bibinfo {pages} {4452} (\bibinfo {year} {2002})}\BibitemShut {NoStop}%
\bibitem [{\citenamefont {Quintavalle}\ \emph {et~al.}(2023)\citenamefont
  {Quintavalle}, \citenamefont {Webster},\ and\ \citenamefont
  {Vasmer}}]{Armanda_partitioning}%
  \BibitemOpen
  \bibfield  {author} {\bibinfo {author} {\bibfnamefont {A.~O.}\ \bibnamefont
  {Quintavalle}}, \bibinfo {author} {\bibfnamefont {P.}~\bibnamefont
  {Webster}},\ and\ \bibinfo {author} {\bibfnamefont {M.}~\bibnamefont
  {Vasmer}},\ }\bibfield  {title} {\bibinfo {title} {{Partitioning qubits in
  hypergraph product codes to implement logical gates}},\ }\href
  {https://doi.org/10.22331/q-2023-10-24-1153} {\bibfield  {journal} {\bibinfo
  {journal} {Quantum}\ }\textbf {\bibinfo {volume} {7}},\ \bibinfo {pages}
  {1153} (\bibinfo {year} {2023})}\BibitemShut {NoStop}%
\bibitem [{\citenamefont {Eberhardt}\ and\ \citenamefont
  {Steffan}(2025)}]{bb_logical_basis}%
  \BibitemOpen
  \bibfield  {author} {\bibinfo {author} {\bibfnamefont {J.~N.}\ \bibnamefont
  {Eberhardt}}\ and\ \bibinfo {author} {\bibfnamefont {V.}~\bibnamefont
  {Steffan}},\ }\bibfield  {title} {\bibinfo {title} {{Logical operators and
  fold-transversal gates of bivariate bicycle codes}},\ }\href
  {https://doi.org/10.1109/TIT.2024.3521638} {\bibfield  {journal} {\bibinfo
  {journal} {IEEE Trans. Inf. Theory}\ }\textbf {\bibinfo {volume} {71}},\
  \bibinfo {pages} {1140} (\bibinfo {year} {2025})}\BibitemShut {NoStop}%
\bibitem [{\citenamefont {Zhang}\ and\ \citenamefont
  {Li}(2025)}]{ying_li_surgery}%
  \BibitemOpen
  \bibfield  {author} {\bibinfo {author} {\bibfnamefont {G.}~\bibnamefont
  {Zhang}}\ and\ \bibinfo {author} {\bibfnamefont {Y.}~\bibnamefont {Li}},\
  }\bibfield  {title} {\bibinfo {title} {{Time-efficient logical operations on
  quantum low-density parity check codes}},\ }\href
  {https://doi.org/10.1103/PhysRevLett.134.070602} {\bibfield  {journal}
  {\bibinfo  {journal} {Phys. Rev. Lett.}\ }\textbf {\bibinfo {volume} {134}},\
  \bibinfo {pages} {070602} (\bibinfo {year} {2025})}\BibitemShut {NoStop}%
\bibitem [{\citenamefont {Cross}\ \emph {et~al.}(2024)\citenamefont {Cross},
  \citenamefont {He}, \citenamefont {Rall} \emph {et~al.}}]{Improved_surgery}%
  \BibitemOpen
  \bibfield  {author} {\bibinfo {author} {\bibfnamefont {A.~W.}\ \bibnamefont
  {Cross}}, \bibinfo {author} {\bibfnamefont {Z.}~\bibnamefont {He}}, \bibinfo
  {author} {\bibfnamefont {P.~J.}\ \bibnamefont {Rall}}, \emph {et~al.},\
  }\bibfield  {title} {\bibinfo {title} {{Improved QLDPC surgery: logical
  measurements and bridging codes}},\ }\href@noop {} {\  (\bibinfo {year}
  {2024})},\ \Eprint {https://arxiv.org/abs/2407.18393} {arXiv:2407.18393}
  \BibitemShut {NoStop}%
\bibitem [{\citenamefont {Poirson}\ \emph {et~al.}(2025)\citenamefont
  {Poirson}, \citenamefont {Roffe},\ and\ \citenamefont
  {Booth}}]{clement_css_surgery}%
  \BibitemOpen
  \bibfield  {author} {\bibinfo {author} {\bibfnamefont {C.}~\bibnamefont
  {Poirson}}, \bibinfo {author} {\bibfnamefont {J.}~\bibnamefont {Roffe}},\
  and\ \bibinfo {author} {\bibfnamefont {R.~I.}\ \bibnamefont {Booth}},\
  }\bibfield  {title} {\bibinfo {title} {{Engineering CSS surgery: compiling
  any CNOT in any code}},\ }\href@noop {} {\  (\bibinfo {year} {2025})},\
  \Eprint {https://arxiv.org/abs/2505.01370} {arXiv:2505.01370} \BibitemShut
  {NoStop}%
\bibitem [{\citenamefont {Cowtan}(2024)}]{Cowtan2024}%
  \BibitemOpen
  \bibfield  {author} {\bibinfo {author} {\bibfnamefont {A.}~\bibnamefont
  {Cowtan}},\ }\bibfield  {title} {\bibinfo {title} {{SSIP: automated surgery
  with quantum LDPC codes}},\ }\href@noop {} {\  (\bibinfo {year} {2024})},\
  \Eprint {https://arxiv.org/abs/2407.09423} {arXiv:2407.09423} \BibitemShut
  {NoStop}%
\bibitem [{\citenamefont {Ide}\ \emph {et~al.}(2025)\citenamefont {Ide},
  \citenamefont {Gowda}, \citenamefont {Nadkarni} \emph
  {et~al.}}]{Xanadu_Homological_measurements}%
  \BibitemOpen
  \bibfield  {author} {\bibinfo {author} {\bibfnamefont {B.}~\bibnamefont
  {Ide}}, \bibinfo {author} {\bibfnamefont {M.~G.}\ \bibnamefont {Gowda}},
  \bibinfo {author} {\bibfnamefont {P.~J.}\ \bibnamefont {Nadkarni}}, \emph
  {et~al.},\ }\bibfield  {title} {\bibinfo {title} {{Fault-tolerant logical
  measurements via homological measurement}},\ }\href
  {https://doi.org/10.1103/PhysRevX.15.021088} {\bibfield  {journal} {\bibinfo
  {journal} {Phys. Rev. X}\ }\textbf {\bibinfo {volume} {15}},\ \bibinfo
  {pages} {021088} (\bibinfo {year} {2025})}\BibitemShut {NoStop}%
\bibitem [{\citenamefont {Zhang}\ \emph {et~al.}(2026)\citenamefont {Zhang},
  \citenamefont {Zhu},\ and\ \citenamefont {Li}}]{GuoqLDPC}%
  \BibitemOpen
  \bibfield  {author} {\bibinfo {author} {\bibfnamefont {G.}~\bibnamefont
  {Zhang}}, \bibinfo {author} {\bibfnamefont {Y.}~\bibnamefont {Zhu}},\ and\
  \bibinfo {author} {\bibfnamefont {Y.}~\bibnamefont {Li}},\ }\bibfield
  {title} {\bibinfo {title} {{Accelerating fault-tolerant quantum computation
  with good qLDPC codes}},\ }\href {https://doi.org/10.1103/3fj5-gn1h}
  {\bibfield  {journal} {\bibinfo  {journal} {PRX Quantum}\ }\textbf {\bibinfo
  {volume} {7}},\ \bibinfo {pages} {033012} (\bibinfo {year}
  {2026})}\BibitemShut {NoStop}%
\bibitem [{\citenamefont {Swaroop}\ \emph {et~al.}(2026)\citenamefont
  {Swaroop}, \citenamefont {Jochym-O'Connor},\ and\ \citenamefont
  {Yoder}}]{swaroop2025universaladaptersquantumldpc}%
  \BibitemOpen
  \bibfield  {author} {\bibinfo {author} {\bibfnamefont {E.}~\bibnamefont
  {Swaroop}}, \bibinfo {author} {\bibfnamefont {T.}~\bibnamefont
  {Jochym-O'Connor}},\ and\ \bibinfo {author} {\bibfnamefont {T.~J.}\
  \bibnamefont {Yoder}},\ }\bibfield  {title} {\bibinfo {title} {{Universal
  adapters between quantum low-density parity check codes}},\ }\href
  {https://doi.org/10.1103/1g44-jp62} {\bibfield  {journal} {\bibinfo
  {journal} {PRX Quantum}\ }\textbf {\bibinfo {volume} {7}},\ \bibinfo {pages}
  {010324} (\bibinfo {year} {2026})}\BibitemShut {NoStop}%
\bibitem [{\citenamefont {Sayginel}\ \emph {et~al.}(2025)\citenamefont
  {Sayginel}, \citenamefont {Koutsioumpas}, \citenamefont {Webster} \emph
  {et~al.}}]{AutQEC}%
  \BibitemOpen
  \bibfield  {author} {\bibinfo {author} {\bibfnamefont {H.}~\bibnamefont
  {Sayginel}}, \bibinfo {author} {\bibfnamefont {S.}~\bibnamefont
  {Koutsioumpas}}, \bibinfo {author} {\bibfnamefont {M.}~\bibnamefont
  {Webster}}, \emph {et~al.},\ }\bibfield  {title} {\bibinfo {title}
  {{Fault-tolerant logical Clifford gates from code automorphisms}},\ }\href
  {https://doi.org/10.1103/vf7v-cpq9} {\bibfield  {journal} {\bibinfo
  {journal} {PRX Quantum}\ }\textbf {\bibinfo {volume} {6}},\ \bibinfo {pages}
  {030343} (\bibinfo {year} {2025})}\BibitemShut {NoStop}%
\bibitem [{\citenamefont {Malcolm}\ \emph {et~al.}(2026)\citenamefont
  {Malcolm}, \citenamefont {Glaudell}, \citenamefont {Fuentes} \emph
  {et~al.}}]{malcolm2025computingefficientlyqldpccodes}%
  \BibitemOpen
  \bibfield  {author} {\bibinfo {author} {\bibfnamefont {A.~J.}\ \bibnamefont
  {Malcolm}}, \bibinfo {author} {\bibfnamefont {A.~N.}\ \bibnamefont
  {Glaudell}}, \bibinfo {author} {\bibfnamefont {P.}~\bibnamefont {Fuentes}},
  \emph {et~al.},\ }\bibfield  {title} {\bibinfo {title} {{Computing
  efficiently in QLDPC codes}},\ }\href
  {https://doi.org/10.1038/s41467-026-73061-9} {\bibfield  {journal} {\bibinfo
  {journal} {Nat. Commun.}\ }\textbf {\bibinfo {volume} {17}},\ \bibinfo
  {pages} {7286} (\bibinfo {year} {2026})}\BibitemShut {NoStop}%
\bibitem [{\citenamefont {He}\ \emph {et~al.}(2025)\citenamefont {He},
  \citenamefont {Cowtan}, \citenamefont {Williamson} \emph
  {et~al.}}]{Extractor}%
  \BibitemOpen
  \bibfield  {author} {\bibinfo {author} {\bibfnamefont {Z.}~\bibnamefont
  {He}}, \bibinfo {author} {\bibfnamefont {A.}~\bibnamefont {Cowtan}}, \bibinfo
  {author} {\bibfnamefont {D.~J.}\ \bibnamefont {Williamson}}, \emph {et~al.},\
  }\bibfield  {title} {\bibinfo {title} {{Extractors: QLDPC architectures for
  efficient Pauli-based computation}},\ }\href@noop {} {\  (\bibinfo {year}
  {2025})},\ \Eprint {https://arxiv.org/abs/2503.10390} {arXiv:2503.10390}
  \BibitemShut {NoStop}%
\bibitem [{\citenamefont {Williamson}\ and\ \citenamefont
  {Yoder}(2026)}]{Gauging_logical}%
  \BibitemOpen
  \bibfield  {author} {\bibinfo {author} {\bibfnamefont {D.~J.}\ \bibnamefont
  {Williamson}}\ and\ \bibinfo {author} {\bibfnamefont {T.~J.}\ \bibnamefont
  {Yoder}},\ }\bibfield  {title} {\bibinfo {title} {{Low-overhead
  fault-tolerant quantum computation by gauging logical operators}},\ }\href
  {https://doi.org/10.1038/s41567-026-03220-8} {\bibfield  {journal} {\bibinfo
  {journal} {Nat. Phys.}\ }\textbf {\bibinfo {volume} {22}},\ \bibinfo {pages}
  {598} (\bibinfo {year} {2026})}\BibitemShut {NoStop}%
\bibitem [{\citenamefont {Yang}\ \emph {et~al.}(2026)\citenamefont {Yang},
  \citenamefont {Zhang},\ and\ \citenamefont {Li}}]{PlanarQLDPC}%
  \BibitemOpen
  \bibfield  {author} {\bibinfo {author} {\bibfnamefont {Y.}~\bibnamefont
  {Yang}}, \bibinfo {author} {\bibfnamefont {G.}~\bibnamefont {Zhang}},\ and\
  \bibinfo {author} {\bibfnamefont {Y.}~\bibnamefont {Li}},\ }\bibfield
  {title} {\bibinfo {title} {{Planar fault-tolerant logical measurements with
  low qubit overhead}},\ }\href {https://doi.org/10.1038/s41534-026-01310-2}
  {\bibfield  {journal} {\bibinfo  {journal} {npj Quantum Inf.}\ }\textbf
  {\bibinfo {volume} {12}},\ \bibinfo {pages} {172} (\bibinfo {year}
  {2026})}\BibitemShut {NoStop}%
\bibitem [{\citenamefont {Xu}\ \emph {et~al.}(2025{\natexlab{b}})\citenamefont
  {Xu}, \citenamefont {Zhou}, \citenamefont {Bluvstein} \emph
  {et~al.}}]{qian_batched}%
  \BibitemOpen
  \bibfield  {author} {\bibinfo {author} {\bibfnamefont {Q.}~\bibnamefont
  {Xu}}, \bibinfo {author} {\bibfnamefont {H.}~\bibnamefont {Zhou}}, \bibinfo
  {author} {\bibfnamefont {D.}~\bibnamefont {Bluvstein}}, \emph {et~al.},\
  }\bibfield  {title} {\bibinfo {title} {{Batched high-rate logical operations
  for quantum LDPC codes}},\ }\href@noop {} {\  (\bibinfo {year}
  {2025}{\natexlab{b}})},\ \Eprint {https://arxiv.org/abs/2510.06159}
  {arXiv:2510.06159} \BibitemShut {NoStop}%
\bibitem [{\citenamefont {Golowich}\ \emph {et~al.}(2025)\citenamefont
  {Golowich}, \citenamefont {Chang},\ and\ \citenamefont
  {Zhu}}]{GolowichChangZhu2025}%
  \BibitemOpen
  \bibfield  {author} {\bibinfo {author} {\bibfnamefont {L.}~\bibnamefont
  {Golowich}}, \bibinfo {author} {\bibfnamefont {K.}~\bibnamefont {Chang}},\
  and\ \bibinfo {author} {\bibfnamefont {G.}~\bibnamefont {Zhu}},\ }\bibfield
  {title} {\bibinfo {title} {{Constant-overhead addressable gates via
  single-shot code switching}},\ }\href@noop {} {\  (\bibinfo {year} {2025})},\
  \Eprint {https://arxiv.org/abs/2510.06760} {arXiv:2510.06760} \BibitemShut
  {NoStop}%
\bibitem [{\citenamefont {Bonilla~Ataides}\ \emph {et~al.}(2025)\citenamefont
  {Bonilla~Ataides}, \citenamefont {Zhou}, \citenamefont {Xu} \emph
  {et~al.}}]{ConstantOverheadQLDPC}%
  \BibitemOpen
  \bibfield  {author} {\bibinfo {author} {\bibfnamefont {J.~P.}\ \bibnamefont
  {Bonilla~Ataides}}, \bibinfo {author} {\bibfnamefont {H.}~\bibnamefont
  {Zhou}}, \bibinfo {author} {\bibfnamefont {Q.}~\bibnamefont {Xu}}, \emph
  {et~al.},\ }\bibfield  {title} {\bibinfo {title} {{Constant-overhead
  fault-tolerant Bell-pair distillation using high-rate codes}},\ }\href
  {https://doi.org/10.1103/s39k-r2kq} {\bibfield  {journal} {\bibinfo
  {journal} {Phys. Rev. Lett.}\ }\textbf {\bibinfo {volume} {135}},\ \bibinfo
  {pages} {130804} (\bibinfo {year} {2025})}\BibitemShut {NoStop}%
\bibitem [{\citenamefont {Yamasaki}\ and\ \citenamefont
  {Koashi}(2024)}]{YamasakiKoashi2024}%
  \BibitemOpen
  \bibfield  {author} {\bibinfo {author} {\bibfnamefont {H.}~\bibnamefont
  {Yamasaki}}\ and\ \bibinfo {author} {\bibfnamefont {M.}~\bibnamefont
  {Koashi}},\ }\bibfield  {title} {\bibinfo {title} {{Time-efficient
  constant-space-overhead fault-tolerant quantum computation}},\ }\href
  {https://doi.org/10.1038/s41567-023-02325-8} {\bibfield  {journal} {\bibinfo
  {journal} {Nat. Phys.}\ }\textbf {\bibinfo {volume} {20}},\ \bibinfo {pages}
  {247} (\bibinfo {year} {2024})}\BibitemShut {NoStop}%
\bibitem [{\citenamefont {Yoshida}\ \emph {et~al.}(2025)\citenamefont
  {Yoshida}, \citenamefont {Tamiya},\ and\ \citenamefont
  {Yamasaki}}]{YoshidaTamiyaYamasaki2025}%
  \BibitemOpen
  \bibfield  {author} {\bibinfo {author} {\bibfnamefont {S.}~\bibnamefont
  {Yoshida}}, \bibinfo {author} {\bibfnamefont {S.}~\bibnamefont {Tamiya}},\
  and\ \bibinfo {author} {\bibfnamefont {H.}~\bibnamefont {Yamasaki}},\
  }\bibfield  {title} {\bibinfo {title} {{Concatenate codes, save qubits}},\
  }\href {https://doi.org/10.1038/s41534-025-01035-8} {\bibfield  {journal}
  {\bibinfo  {journal} {npj Quantum Inf.}\ }\textbf {\bibinfo {volume} {11}},\
  \bibinfo {pages} {88} (\bibinfo {year} {2025})}\BibitemShut {NoStop}%
\bibitem [{\citenamefont {Goto}(2024)}]{Goto2024}%
  \BibitemOpen
  \bibfield  {author} {\bibinfo {author} {\bibfnamefont {H.}~\bibnamefont
  {Goto}},\ }\bibfield  {title} {\bibinfo {title} {{High-performance
  fault-tolerant quantum computing with many-hypercube codes}},\ }\href
  {https://doi.org/10.1126/sciadv.adp6388} {\bibfield  {journal} {\bibinfo
  {journal} {Sci. Adv.}\ }\textbf {\bibinfo {volume} {10}},\ \bibinfo {pages}
  {eadp6388} (\bibinfo {year} {2024})}\BibitemShut {NoStop}%
\bibitem [{\citenamefont {Zhu}\ \emph {et~al.}(2024)\citenamefont {Zhu},
  \citenamefont {Sikander}, \citenamefont {Portnoy} \emph
  {et~al.}}]{ZhuSikanderPortnoyCrossBrown2024}%
  \BibitemOpen
  \bibfield  {author} {\bibinfo {author} {\bibfnamefont {G.}~\bibnamefont
  {Zhu}}, \bibinfo {author} {\bibfnamefont {S.}~\bibnamefont {Sikander}},
  \bibinfo {author} {\bibfnamefont {E.}~\bibnamefont {Portnoy}}, \emph
  {et~al.},\ }\bibfield  {title} {\bibinfo {title} {{Non-Clifford and
  parallelizable fault-tolerant logical gates on constant and almost-constant
  rate homological quantum LDPC codes via higher symmetries}},\ }\href
  {https://doi.org/10.1103/PRXQuantum.5.040338} {\bibfield  {journal} {\bibinfo
   {journal} {PRX Quantum}\ }\textbf {\bibinfo {volume} {5}},\ \bibinfo {pages}
  {040338} (\bibinfo {year} {2024})}\BibitemShut {NoStop}%
\bibitem [{\citenamefont {Zheng}\ \emph {et~al.}(2025)\citenamefont {Zheng},
  \citenamefont {Jiang},\ and\ \citenamefont {Xu}}]{Qian_high_rate}%
  \BibitemOpen
  \bibfield  {author} {\bibinfo {author} {\bibfnamefont {G.}~\bibnamefont
  {Zheng}}, \bibinfo {author} {\bibfnamefont {L.}~\bibnamefont {Jiang}},\ and\
  \bibinfo {author} {\bibfnamefont {Q.}~\bibnamefont {Xu}},\ }\bibfield
  {title} {\bibinfo {title} {{High-rate surgery: towards constant-overhead
  logical operations}},\ }\href@noop {} {\  (\bibinfo {year} {2025})},\ \Eprint
  {https://arxiv.org/abs/2510.08523} {arXiv:2510.08523} \BibitemShut {NoStop}%
\bibitem [{\citenamefont {Chang}\ \emph {et~al.}(2026)\citenamefont {Chang},
  \citenamefont {He}, \citenamefont {Yoder} \emph
  {et~al.}}]{ChangHeYoderZhuJochymOConnor2026}%
  \BibitemOpen
  \bibfield  {author} {\bibinfo {author} {\bibfnamefont {K.}~\bibnamefont
  {Chang}}, \bibinfo {author} {\bibfnamefont {Z.}~\bibnamefont {He}}, \bibinfo
  {author} {\bibfnamefont {T.~J.}\ \bibnamefont {Yoder}}, \emph {et~al.},\
  }\bibfield  {title} {\bibinfo {title} {{Constant-time surgery on 2D
  hypergraph product codes with near-constant space overhead}},\ }\href@noop {}
  {\  (\bibinfo {year} {2026})},\ \Eprint {https://arxiv.org/abs/2603.02157}
  {arXiv:2603.02157} \BibitemShut {NoStop}%
\bibitem [{\citenamefont {Liang}\ \emph {et~al.}(2025)\citenamefont {Liang},
  \citenamefont {Liu}, \citenamefont {Song} \emph
  {et~al.}}]{Yu-an_generalised_toric_code}%
  \BibitemOpen
  \bibfield  {author} {\bibinfo {author} {\bibfnamefont {Z.}~\bibnamefont
  {Liang}}, \bibinfo {author} {\bibfnamefont {K.}~\bibnamefont {Liu}}, \bibinfo
  {author} {\bibfnamefont {H.}~\bibnamefont {Song}}, \emph {et~al.},\
  }\bibfield  {title} {\bibinfo {title} {{Generalized toric codes on twisted
  tori for quantum error correction}},\ }\href
  {https://doi.org/10.1103/rmy6-9n89} {\bibfield  {journal} {\bibinfo
  {journal} {PRX Quantum}\ }\textbf {\bibinfo {volume} {6}},\ \bibinfo {pages}
  {020357} (\bibinfo {year} {2025})}\BibitemShut {NoStop}%
\bibitem [{\citenamefont {Bravyi}\ \emph {et~al.}(2024)\citenamefont {Bravyi},
  \citenamefont {Cross}, \citenamefont {Gambetta} \emph {et~al.}}]{BB_IBM}%
  \BibitemOpen
  \bibfield  {author} {\bibinfo {author} {\bibfnamefont {S.}~\bibnamefont
  {Bravyi}}, \bibinfo {author} {\bibfnamefont {A.~W.}\ \bibnamefont {Cross}},
  \bibinfo {author} {\bibfnamefont {J.~M.}\ \bibnamefont {Gambetta}}, \emph
  {et~al.},\ }\bibfield  {title} {\bibinfo {title} {{High-threshold and
  low-overhead fault-tolerant quantum memory}},\ }\href
  {https://doi.org/10.1038/s41586-024-07107-7} {\bibfield  {journal} {\bibinfo
  {journal} {Nature}\ }\textbf {\bibinfo {volume} {627}},\ \bibinfo {pages}
  {778} (\bibinfo {year} {2024})}\BibitemShut {NoStop}%
\bibitem [{\citenamefont {Scruby}\ \emph {et~al.}(2026)\citenamefont {Scruby},
  \citenamefont {Hillmann},\ and\ \citenamefont {Roffe}}]{QRC}%
  \BibitemOpen
  \bibfield  {author} {\bibinfo {author} {\bibfnamefont {T.~R.}\ \bibnamefont
  {Scruby}}, \bibinfo {author} {\bibfnamefont {T.}~\bibnamefont {Hillmann}},\
  and\ \bibinfo {author} {\bibfnamefont {J.}~\bibnamefont {Roffe}},\ }\bibfield
   {title} {\bibinfo {title} {{High-threshold, low-overhead and single-shot
  decodable fault-tolerant quantum memory}},\ }\href
  {https://doi.org/10.1103/67xf-zdjb} {\bibfield  {journal} {\bibinfo
  {journal} {PRX Quantum}\ }\textbf {\bibinfo {volume} {7}},\ \bibinfo {pages}
  {020310} (\bibinfo {year} {2026})}\BibitemShut {NoStop}%
\bibitem [{\citenamefont {Bluvstein}\ \emph {et~al.}(2026)\citenamefont
  {Bluvstein}, \citenamefont {Geim}, \citenamefont {Li} \emph
  {et~al.}}]{LukinFaultTolerant}%
  \BibitemOpen
  \bibfield  {author} {\bibinfo {author} {\bibfnamefont {D.}~\bibnamefont
  {Bluvstein}}, \bibinfo {author} {\bibfnamefont {A.~A.}\ \bibnamefont {Geim}},
  \bibinfo {author} {\bibfnamefont {S.~H.}\ \bibnamefont {Li}}, \emph
  {et~al.},\ }\bibfield  {title} {\bibinfo {title} {{A fault-tolerant
  neutral-atom architecture for universal quantum computation}},\ }\href
  {https://doi.org/10.1038/s41586-025-09848-5} {\bibfield  {journal} {\bibinfo
  {journal} {Nature}\ }\textbf {\bibinfo {volume} {649}},\ \bibinfo {pages}
  {39} (\bibinfo {year} {2026})}\BibitemShut {NoStop}%
\bibitem [{\citenamefont {Xu}\ \emph {et~al.}(2024)\citenamefont {Xu},
  \citenamefont {Bonilla~Ataides}, \citenamefont {Pattison} \emph
  {et~al.}}]{Qian_reconfigurable}%
  \BibitemOpen
  \bibfield  {author} {\bibinfo {author} {\bibfnamefont {Q.}~\bibnamefont
  {Xu}}, \bibinfo {author} {\bibfnamefont {J.~P.}\ \bibnamefont
  {Bonilla~Ataides}}, \bibinfo {author} {\bibfnamefont {C.~A.}\ \bibnamefont
  {Pattison}}, \emph {et~al.},\ }\bibfield  {title} {\bibinfo {title}
  {{Constant-overhead fault-tolerant quantum computation with reconfigurable
  atom arrays}},\ }\href {https://doi.org/10.1038/s41567-024-02479-z}
  {\bibfield  {journal} {\bibinfo  {journal} {Nat. Phys.}\ }\textbf {\bibinfo
  {volume} {20}},\ \bibinfo {pages} {1084} (\bibinfo {year}
  {2024})}\BibitemShut {NoStop}%
\bibitem [{\citenamefont {Pino}\ \emph {et~al.}(2021)\citenamefont {Pino},
  \citenamefont {Dreiling}, \citenamefont {Figgatt} \emph {et~al.}}]{QCCD}%
  \BibitemOpen
  \bibfield  {author} {\bibinfo {author} {\bibfnamefont {J.~M.}\ \bibnamefont
  {Pino}}, \bibinfo {author} {\bibfnamefont {J.~M.}\ \bibnamefont {Dreiling}},
  \bibinfo {author} {\bibfnamefont {C.}~\bibnamefont {Figgatt}}, \emph
  {et~al.},\ }\bibfield  {title} {\bibinfo {title} {{Demonstration of the
  trapped-ion quantum-CCD computer architecture}},\ }\href
  {https://doi.org/10.1038/s41586-021-03318-4} {\bibfield  {journal} {\bibinfo
  {journal} {Nature}\ }\textbf {\bibinfo {volume} {592}},\ \bibinfo {pages}
  {209} (\bibinfo {year} {2021})}\BibitemShut {NoStop}%
\bibitem [{\citenamefont {Dasu}\ \emph {et~al.}(2026)\citenamefont {Dasu},
  \citenamefont {DeCross}, \citenamefont {Guo} \emph
  {et~al.}}]{IonTrapBreakEven}%
  \BibitemOpen
  \bibfield  {author} {\bibinfo {author} {\bibfnamefont {S.}~\bibnamefont
  {Dasu}}, \bibinfo {author} {\bibfnamefont {M.}~\bibnamefont {DeCross}},
  \bibinfo {author} {\bibfnamefont {A.~Y.}\ \bibnamefont {Guo}}, \emph
  {et~al.},\ }\bibfield  {title} {\bibinfo {title} {{Computing with many
  encoded logical qubits beyond break-even}},\ }\href@noop {} {\  (\bibinfo
  {year} {2026})},\ \Eprint {https://arxiv.org/abs/2602.22211}
  {arXiv:2602.22211} \BibitemShut {NoStop}%
\bibitem [{\citenamefont {{IBM Research}}(2025)}]{IBMTuneable}%
  \BibitemOpen
  \bibfield  {author} {\bibinfo {author} {\bibnamefont {{IBM Research}}},\
  }\href@noop {} {\bibinfo {title} {{IBM delivers new quantum processors,
  software, and algorithm breakthroughs on path to advantage and fault
  tolerance}}},\ \bibinfo {howpublished}
  {\url{https://newsroom.ibm.com/2025-11-12-ibm-delivers-new-quantum-processors,-software,-and-algorithm-breakthroughs-on-path-to-advantage-and-fault-tolerance}}
  (\bibinfo {year} {2025})\BibitemShut {NoStop}%
\bibitem [{\citenamefont {Zhou}\ \emph {et~al.}(2026)\citenamefont {Zhou},
  \citenamefont {Gu}, \citenamefont {Eisert} \emph
  {et~al.}}]{ZhouGuEisertZhao2026}%
  \BibitemOpen
  \bibfield  {author} {\bibinfo {author} {\bibfnamefont {T.-G.}\ \bibnamefont
  {Zhou}}, \bibinfo {author} {\bibfnamefont {B.}~\bibnamefont {Gu}}, \bibinfo
  {author} {\bibfnamefont {J.}~\bibnamefont {Eisert}}, \emph {et~al.},\
  }\bibfield  {title} {\bibinfo {title} {{Parallel algebraic surgery for qLDPC
  codes with constant shuttling depth}},\ }\href@noop {} {\  (\bibinfo {year}
  {2026})},\ \Eprint {https://arxiv.org/abs/2609.39874} {arXiv:2609.39874}
  \BibitemShut {NoStop}%
\bibitem [{\citenamefont {Quintavalle}\ and\ \citenamefont
  {Campbell}(2022)}]{Reshape}%
  \BibitemOpen
  \bibfield  {author} {\bibinfo {author} {\bibfnamefont {A.~O.}\ \bibnamefont
  {Quintavalle}}\ and\ \bibinfo {author} {\bibfnamefont {E.~T.}\ \bibnamefont
  {Campbell}},\ }\bibfield  {title} {\bibinfo {title} {{ReShape: a decoder for
  hypergraph product codes}},\ }\href
  {https://doi.org/10.1109/TIT.2022.3184108} {\bibfield  {journal} {\bibinfo
  {journal} {IEEE Trans. Inf. Theory}\ }\textbf {\bibinfo {volume} {68}},\
  \bibinfo {pages} {6569} (\bibinfo {year} {2022})}\BibitemShut {NoStop}%
\bibitem [{\citenamefont {Cowtan}\ \emph {et~al.}(2025)\citenamefont {Cowtan},
  \citenamefont {He}, \citenamefont {Williamson} \emph
  {et~al.}}]{CowtanHeWilliamsonYoder2025}%
  \BibitemOpen
  \bibfield  {author} {\bibinfo {author} {\bibfnamefont {A.}~\bibnamefont
  {Cowtan}}, \bibinfo {author} {\bibfnamefont {Z.}~\bibnamefont {He}}, \bibinfo
  {author} {\bibfnamefont {D.~J.}\ \bibnamefont {Williamson}}, \emph {et~al.},\
  }\bibfield  {title} {\bibinfo {title} {{Fast and fault-tolerant logical
  measurements: auxiliary hypergraphs and transversal surgery}},\ }\href@noop
  {} {\  (\bibinfo {year} {2025})},\ \Eprint {https://arxiv.org/abs/2510.14895}
  {arXiv:2510.14895} \BibitemShut {NoStop}%
\bibitem [{\citenamefont {Litinski}\ and\ \citenamefont {von
  Oppen}(2018)}]{LitinskiVonOppen2018}%
  \BibitemOpen
  \bibfield  {author} {\bibinfo {author} {\bibfnamefont {D.}~\bibnamefont
  {Litinski}}\ and\ \bibinfo {author} {\bibfnamefont {F.}~\bibnamefont {von
  Oppen}},\ }\bibfield  {title} {\bibinfo {title} {{Lattice surgery with a
  twist: simplifying Clifford gates of surface codes}},\ }\href
  {https://doi.org/10.22331/q-2018-05-04-62} {\bibfield  {journal} {\bibinfo
  {journal} {Quantum}\ }\textbf {\bibinfo {volume} {2}},\ \bibinfo {pages} {62}
  (\bibinfo {year} {2018})}\BibitemShut {NoStop}%
\bibitem [{\citenamefont {Hirasaki}\ and\ \citenamefont
  {Lee}(2026)}]{HirasakiLee2026}%
  \BibitemOpen
  \bibfield  {author} {\bibinfo {author} {\bibfnamefont {Y.}~\bibnamefont
  {Hirasaki}}\ and\ \bibinfo {author} {\bibfnamefont {J.~Y.}\ \bibnamefont
  {Lee}},\ }\bibfield  {title} {\bibinfo {title} {{Lifting lifted product
  codes}},\ }\href@noop {} {\  (\bibinfo {year} {2026})},\ \Eprint
  {https://arxiv.org/abs/2607.28621} {arXiv:2607.28621} \BibitemShut {NoStop}%
\bibitem [{\citenamefont {Lee}(2026)}]{Lee2026LogicalSpectroscopy}%
  \BibitemOpen
  \bibfield  {author} {\bibinfo {author} {\bibfnamefont {J.~Y.}\ \bibnamefont
  {Lee}},\ }\bibfield  {title} {\bibinfo {title} {{Logical spectroscopy:
  lifted-product codes with addressable bases}},\ }\href@noop {} {\  (\bibinfo
  {year} {2026})},\ \Eprint {https://arxiv.org/abs/2607.05386}
  {arXiv:2607.05386} \BibitemShut {NoStop}%
\bibitem [{\citenamefont {Zheng}\ \emph {et~al.}(2026)\citenamefont {Zheng},
  \citenamefont {Zheng}, \citenamefont {Jiang} \emph
  {et~al.}}]{ZhengZhengJiangXu2026}%
  \BibitemOpen
  \bibfield  {author} {\bibinfo {author} {\bibfnamefont {H.}~\bibnamefont
  {Zheng}}, \bibinfo {author} {\bibfnamefont {G.}~\bibnamefont {Zheng}},
  \bibinfo {author} {\bibfnamefont {L.}~\bibnamefont {Jiang}}, \emph {et~al.},\
  }\bibfield  {title} {\bibinfo {title} {{Logical computation with canonical
  lifted product codes}},\ }\href@noop {} {\  (\bibinfo {year} {2026})},\
  \Eprint {https://arxiv.org/abs/2607.28605} {arXiv:2607.28605} \BibitemShut
  {NoStop}%
\bibitem [{\citenamefont {Berent}\ \emph {et~al.}(2026)\citenamefont {Berent},
  \citenamefont {Cohen},\ and\ \citenamefont
  {Quintavalle}}]{BerentCohenQuintavalle2026}%
  \BibitemOpen
  \bibfield  {author} {\bibinfo {author} {\bibfnamefont {L.}~\bibnamefont
  {Berent}}, \bibinfo {author} {\bibfnamefont {L.~Z.}\ \bibnamefont {Cohen}},\
  and\ \bibinfo {author} {\bibfnamefont {A.~O.}\ \bibnamefont {Quintavalle}},\
  }\bibfield  {title} {\bibinfo {title} {{Lifted surgery: fast processing with
  QLDPC codes}},\ }\href@noop {} {\  (\bibinfo {year} {2026})},\ \Eprint
  {https://arxiv.org/abs/2609.11723} {arXiv:2609.11723} \BibitemShut {NoStop}%
\bibitem [{\citenamefont {Menon}\ \emph {et~al.}(2026)\citenamefont {Menon},
  \citenamefont {Mehta}, \citenamefont {Gu} \emph
  {et~al.}}]{MenonEtAl2026MagicFactories}%
  \BibitemOpen
  \bibfield  {author} {\bibinfo {author} {\bibfnamefont {V.}~\bibnamefont
  {Menon}}, \bibinfo {author} {\bibfnamefont {R.}~\bibnamefont {Mehta}},
  \bibinfo {author} {\bibfnamefont {A.}~\bibnamefont {Gu}}, \emph {et~al.},\
  }\bibfield  {title} {\bibinfo {title} {{Single-shot magic state factories for
  quantum LDPC codes}},\ }\href@noop {} {\  (\bibinfo {year} {2026})},\ \Eprint
  {https://arxiv.org/abs/2610.10731} {arXiv:2610.10731} \BibitemShut {NoStop}%
\bibitem [{\citenamefont {Luthra}\ and\ \citenamefont
  {Owens}(2026)}]{LuthraOwens2026}%
  \BibitemOpen
  \bibfield  {author} {\bibinfo {author} {\bibfnamefont {S.}~\bibnamefont
  {Luthra}}\ and\ \bibinfo {author} {\bibfnamefont {A.}~\bibnamefont {Owens}},\
  }\bibfield  {title} {\bibinfo {title} {{Individually addressed logical Pauli
  measurements by QLDPC surgery on a trapped ion quantum computer}},\
  }\href@noop {} {\  (\bibinfo {year} {2026})},\ \Eprint
  {https://arxiv.org/abs/2610.39067} {arXiv:2610.39067} \BibitemShut {NoStop}%
\bibitem [{\citenamefont {Baspin}\ \emph {et~al.}(2025)\citenamefont {Baspin},
  \citenamefont {Berent},\ and\ \citenamefont {Cohen}}]{BaspinBerentCohen2025}%
  \BibitemOpen
  \bibfield  {author} {\bibinfo {author} {\bibfnamefont {N.}~\bibnamefont
  {Baspin}}, \bibinfo {author} {\bibfnamefont {L.}~\bibnamefont {Berent}},\
  and\ \bibinfo {author} {\bibfnamefont {L.~Z.}\ \bibnamefont {Cohen}},\
  }\bibfield  {title} {\bibinfo {title} {{Fast surgery for quantum LDPC
  codes}},\ }\href@noop {} {\  (\bibinfo {year} {2025})},\ \Eprint
  {https://arxiv.org/abs/2510.04521} {arXiv:2510.04521} \BibitemShut {NoStop}%
\bibitem [{\citenamefont {Yuan}\ \emph {et~al.}(2026)\citenamefont {Yuan},
  \citenamefont {Cowtan}, \citenamefont {He} \emph
  {et~al.}}]{YuanCowtanHeLinWilliamson2026}%
  \BibitemOpen
  \bibfield  {author} {\bibinfo {author} {\bibfnamefont {A.~C.}\ \bibnamefont
  {Yuan}}, \bibinfo {author} {\bibfnamefont {A.}~\bibnamefont {Cowtan}},
  \bibinfo {author} {\bibfnamefont {Z.}~\bibnamefont {He}}, \emph {et~al.},\
  }\bibfield  {title} {\bibinfo {title} {{Parsimonious quantum low-density
  parity-check code surgery}},\ }\href@noop {} {\  (\bibinfo {year} {2026})},\
  \Eprint {https://arxiv.org/abs/2603.05082} {arXiv:2603.05082} \BibitemShut
  {NoStop}%
\bibitem [{\citenamefont {Pecorari}\ \emph {et~al.}(2025)\citenamefont
  {Pecorari}, \citenamefont {Guerci}, \citenamefont {Perrin} \emph
  {et~al.}}]{PecorariGuerciPerrinPupillo2025}%
  \BibitemOpen
  \bibfield  {author} {\bibinfo {author} {\bibfnamefont {L.}~\bibnamefont
  {Pecorari}}, \bibinfo {author} {\bibfnamefont {F.~P.}\ \bibnamefont
  {Guerci}}, \bibinfo {author} {\bibfnamefont {H.}~\bibnamefont {Perrin}},
  \emph {et~al.},\ }\bibfield  {title} {\bibinfo {title} {{Addressable
  gate-based logical computation with quantum LDPC codes}},\ }\href@noop {} {\
  (\bibinfo {year} {2025})},\ \Eprint {https://arxiv.org/abs/2511.06124}
  {arXiv:2511.06124} \BibitemShut {NoStop}%
\bibitem [{\citenamefont {Tansuwannont}\ \emph {et~al.}(2026)\citenamefont
  {Tansuwannont}, \citenamefont {Chan},\ and\ \citenamefont
  {Takagi}}]{TansuwannontChanTakagi2026}%
  \BibitemOpen
  \bibfield  {author} {\bibinfo {author} {\bibfnamefont {T.}~\bibnamefont
  {Tansuwannont}}, \bibinfo {author} {\bibfnamefont {T.}~\bibnamefont {Chan}},\
  and\ \bibinfo {author} {\bibfnamefont {R.}~\bibnamefont {Takagi}},\
  }\bibfield  {title} {\bibinfo {title} {{Construction of the full logical
  Clifford group for high-rate quantum Reed--Muller codes using only
  transversal and fold-transversal gates}},\ }\href@noop {} {\  (\bibinfo
  {year} {2026})},\ \Eprint {https://arxiv.org/abs/2602.09788}
  {arXiv:2602.09788} \BibitemShut {NoStop}%
\bibitem [{\citenamefont {Drozd}\ and\ \citenamefont
  {Kirichenko}(1994)}]{Finite_d_algebras}%
  \BibitemOpen
  \bibfield  {author} {\bibinfo {author} {\bibfnamefont {Y.~A.}\ \bibnamefont
  {Drozd}}\ and\ \bibinfo {author} {\bibfnamefont {V.~V.}\ \bibnamefont
  {Kirichenko}},\ }\href {https://doi.org/10.1007/978-3-642-76244-4} {\emph
  {\bibinfo {title} {{Finite dimensional algebras}}}}\ (\bibinfo  {publisher}
  {Springer},\ \bibinfo {year} {1994})\BibitemShut {NoStop}%
\bibitem [{\citenamefont {Hatcher}(2002)}]{Algebraic_topo}%
  \BibitemOpen
  \bibfield  {author} {\bibinfo {author} {\bibfnamefont {A.}~\bibnamefont
  {Hatcher}},\ }\href@noop {} {\emph {\bibinfo {title} {{Algebraic
  topology}}}}\ (\bibinfo  {publisher} {Cambridge University Press},\ \bibinfo
  {year} {2002})\BibitemShut {NoStop}%
\bibitem [{\citenamefont {Breuckmann}\ and\ \citenamefont
  {Eberhardt}(2021{\natexlab{b}})}]{Balanced_product}%
  \BibitemOpen
  \bibfield  {author} {\bibinfo {author} {\bibfnamefont {N.~P.}\ \bibnamefont
  {Breuckmann}}\ and\ \bibinfo {author} {\bibfnamefont {J.~N.}\ \bibnamefont
  {Eberhardt}},\ }\bibfield  {title} {\bibinfo {title} {{Balanced product
  quantum codes}},\ }\href {https://doi.org/10.1109/TIT.2021.3097347}
  {\bibfield  {journal} {\bibinfo  {journal} {IEEE Trans. Inf. Theory}\
  }\textbf {\bibinfo {volume} {67}},\ \bibinfo {pages} {6653} (\bibinfo {year}
  {2021}{\natexlab{b}})}\BibitemShut {NoStop}%
\bibitem [{\citenamefont {Giannisis~Manes}\ and\ \citenamefont
  {Claes}(2025)}]{distance_preserving_syndrome_HGP}%
  \BibitemOpen
  \bibfield  {author} {\bibinfo {author} {\bibfnamefont {A.}~\bibnamefont
  {Giannisis~Manes}}\ and\ \bibinfo {author} {\bibfnamefont {J.}~\bibnamefont
  {Claes}},\ }\bibfield  {title} {\bibinfo {title} {{Distance-preserving
  stabilizer measurements in hypergraph product codes}},\ }\href
  {https://doi.org/10.22331/q-2025-01-30-1618} {\bibfield  {journal} {\bibinfo
  {journal} {Quantum}\ }\textbf {\bibinfo {volume} {9}},\ \bibinfo {pages}
  {1618} (\bibinfo {year} {2025})}\BibitemShut {NoStop}%
\bibitem [{\citenamefont {Strikis}\ \emph {et~al.}(2026)\citenamefont
  {Strikis}, \citenamefont {Browne},\ and\ \citenamefont
  {Beverland}}]{StrikisBrowneBeverland2026}%
  \BibitemOpen
  \bibfield  {author} {\bibinfo {author} {\bibfnamefont {A.}~\bibnamefont
  {Strikis}}, \bibinfo {author} {\bibfnamefont {D.~E.}\ \bibnamefont
  {Browne}},\ and\ \bibinfo {author} {\bibfnamefont {M.~E.}\ \bibnamefont
  {Beverland}},\ }\bibfield  {title} {\bibinfo {title} {{High-performance
  syndrome extraction circuits for quantum codes}},\ }\href@noop {} {\
  (\bibinfo {year} {2026})},\ \Eprint {https://arxiv.org/abs/2603.05481}
  {arXiv:2603.05481} \BibitemShut {NoStop}%
\bibitem [{\citenamefont {Bhardwaj}\ \emph {et~al.}(2026)\citenamefont
  {Bhardwaj}, \citenamefont {Ma}, \citenamefont {Meister} \emph
  {et~al.}}]{BhardwajEtAl2026}%
  \BibitemOpen
  \bibfield  {author} {\bibinfo {author} {\bibfnamefont {A.}~\bibnamefont
  {Bhardwaj}}, \bibinfo {author} {\bibfnamefont {M.}~\bibnamefont {Ma}},
  \bibinfo {author} {\bibfnamefont {N.}~\bibnamefont {Meister}}, \emph
  {et~al.},\ }\bibfield  {title} {\bibinfo {title} {{High-rate qLDPC
  processors}},\ }\href@noop {} {\  (\bibinfo {year} {2026})},\ \Eprint
  {https://arxiv.org/abs/2607.28795} {arXiv:2607.28795} \BibitemShut {NoStop}%
\bibitem [{\citenamefont {Gidney}(2021)}]{Gidney_2021}%
  \BibitemOpen
  \bibfield  {author} {\bibinfo {author} {\bibfnamefont {C.}~\bibnamefont
  {Gidney}},\ }\bibfield  {title} {\bibinfo {title} {{Stim: a fast stabilizer
  circuit simulator}},\ }\href {https://doi.org/10.22331/q-2021-07-06-497}
  {\bibfield  {journal} {\bibinfo  {journal} {Quantum}\ }\textbf {\bibinfo
  {volume} {5}},\ \bibinfo {pages} {497} (\bibinfo {year} {2021})}\BibitemShut
  {NoStop}%
\bibitem [{\citenamefont {M{\"u}ller}\ \emph {et~al.}(2025)\citenamefont
  {M{\"u}ller}, \citenamefont {Alexander}, \citenamefont {Beverland} \emph
  {et~al.}}]{mullerImprovedBeliefPropagation2025b}%
  \BibitemOpen
  \bibfield  {author} {\bibinfo {author} {\bibfnamefont {T.}~\bibnamefont
  {M{\"u}ller}}, \bibinfo {author} {\bibfnamefont {T.}~\bibnamefont
  {Alexander}}, \bibinfo {author} {\bibfnamefont {M.~E.}\ \bibnamefont
  {Beverland}}, \emph {et~al.},\ }\bibfield  {title} {\bibinfo {title}
  {{Improved belief propagation is sufficient for real-time decoding of quantum
  memory}},\ }\href@noop {} {\  (\bibinfo {year} {2025})},\ \Eprint
  {https://arxiv.org/abs/2506.01779} {arXiv:2506.01779} \BibitemShut {NoStop}%
\bibitem [{\citenamefont {Maan}\ \emph {et~al.}(2026)\citenamefont {Maan},
  \citenamefont {Garcia~Herrero}, \citenamefont {Paler} \emph
  {et~al.}}]{MaanGarciaHerreroPaler2026}%
  \BibitemOpen
  \bibfield  {author} {\bibinfo {author} {\bibfnamefont {A.~S.}\ \bibnamefont
  {Maan}}, \bibinfo {author} {\bibfnamefont {F.~M.}\ \bibnamefont
  {Garcia~Herrero}}, \bibinfo {author} {\bibfnamefont {A.}~\bibnamefont
  {Paler}}, \emph {et~al.},\ }\bibfield  {title} {\bibinfo {title} {{Decoding
  correlated errors in quantum LDPC codes}},\ }\href
  {https://doi.org/10.1038/s41467-026-70556-3} {\bibfield  {journal} {\bibinfo
  {journal} {Nat. Commun.}\ }\textbf {\bibinfo {volume} {17}},\ \bibinfo
  {pages} {3965} (\bibinfo {year} {2026})}\BibitemShut {NoStop}%
\bibitem [{\citenamefont {Koutsioumpas}\ \emph {et~al.}(2025)\citenamefont
  {Koutsioumpas}, \citenamefont {Noszko}, \citenamefont {Sayginel} \emph
  {et~al.}}]{KoutsioumpasNoszkoSayginel2025}%
  \BibitemOpen
  \bibfield  {author} {\bibinfo {author} {\bibfnamefont {S.}~\bibnamefont
  {Koutsioumpas}}, \bibinfo {author} {\bibfnamefont {T.}~\bibnamefont
  {Noszko}}, \bibinfo {author} {\bibfnamefont {H.}~\bibnamefont {Sayginel}},
  \emph {et~al.},\ }\bibfield  {title} {\bibinfo {title} {{Colour Codes Reach
  Surface Code Performance using Vibe Decoding}},\ }\href@noop {} {\  (\bibinfo
  {year} {2025})},\ \Eprint {https://arxiv.org/abs/2508.15743}
  {arXiv:2508.15743} \BibitemShut {NoStop}%
\end{thebibliography}%

\clearpage

\appendix
\begingroup
  \makeatletter
  \let\qgpu@contentsline\contentsline
  \let\qgpu@appendixtoc\appendix@toc
  \renewcommand{\contentsline}[4]{}
  \def\appendix@toc{%
    \qgpu@appendixtoc
    \let\contentsline\qgpu@contentsline
  }
  \tableofcontents
  \makeatother
\endgroup

\section{Chain-complex and lifted product code notation}
\label{Section: preliminaries}

\subsection{Ring, module and binary representation}

We now outline the mathematical foundations of our code construction.
Let $G$ be a finite abelian group of order $l$. The group algebra of $G$ over $\FF_2$, denoted $R=\FF_2[G]$, is a finite commutative ring. Throughout this work, the coefficient rings are such group algebras, and the chain groups are finite-rank free $R$-modules. A polynomial ring in $x$ over $\FF_2$, denoted $\FF_2[x]$, consists of polynomials $q=q_0+q_1x+\cdots+q_mx^m$ with coefficients $q_i\in\FF_2$. The quotient ring $\FF_2[x]/(x^l+1)$ is obtained by imposing the relation $x^l=1$. As a group algebra, $\FF_2[x]/(x^l+1)\cong\FF_2[C_l]$, where $C_l=\braket{g\mid g^l=1}$ is the cyclic group of order $l$.

An $R$-module may be viewed as a generalisation of a vector space in which the scalars are drawn from $R$ rather than a field. Homomorphisms between modules of the form $R^\nu$ and $R^m$ are represented by matrices $M\in R^{m\times \nu}$.

Since $R=\FF_2[G]$ is an $l$-dimensional vector space over $\FF_2$ with basis $\{g\}_{g\in G}$, any $R$-linear map admits a faithful binary representation~\cite[Theorem~1.3.1]{Finite_d_algebras}. We define the following binary representation of homomorphisms
\eqa{
  \BB: R^{m\x \nu}\to \FF_2^{(lm)\x (l\nu)};\quad M_{i,j}\mapsto \BB(M_{i,j})\in \FF_2^{l\x l},
}
where $\BB(M_{i,j})$ is the matrix of multiplication by $M_{i,j}$ in the group basis~\cite{Panteleev_almost_linear}. Accordingly, $R$-modules are lifted as $\BB(R^m)=\FF_2^{lm}$.

The group algebra $R=\FF_2[G]$ is equipped with a canonical involution $*:R\to R$, defined as the $\FF_2$-linear extension of group inversion $g\mapsto g^{-1}$. For $R\cong\FF_2[x]/(x^l+1)$, we identify $x\leftrightarrow g$, this involution acts as
\eqa{
  *: q_0+q_1x+\cdots+q_mx^m \mapsto
  q_0+q_1x^{l-1}+\cdots+q_mx^{l-m}.
}
For a matrix $M=(M_{i,j})_{m\times \nu}$, we define its conjugate transpose as
$M^*=(M_{j,i}^*)_{\nu\times m}$. The involution guarantees consistency between
$R$-module conjugate transposition and the binary representation, satisfying
\begin{equation}
  \BB(M^*)=\BB(M)^\intercal.
\end{equation}
For example, let $R = \FF_2[x]/(x^3+1)$, the binary representation of the following matrix and its conjugate transpose are, respectively,
\begin{eqnarray}
  M =
  \begin{pmatrix}
    1 & 1+x\\
    x^2 & x+x^2
  \end{pmatrix},&\ M^* =
  \begin{pmatrix}
    1 & x\\
    1+x^2 & x+x^2
  \end{pmatrix};\\
  \BB(M) =
  \begin{pmatrix}
    1 & 0 & 0 & 1 & 1 & 0\\
    0 & 1 & 0 & 0 & 1 & 1\\
    0 & 0 & 1 & 1 & 0 & 1\\
    0 & 0 & 1 & 0 & 1 & 1\\
    1 & 0 & 0 & 1 & 0 & 1\\
    0 & 1 & 0 & 1 & 1 & 0
  \end{pmatrix},&\
  \BB(M)^\intercal =
  \begin{pmatrix}
    1 & 0 & 0 & 0 & 1 & 0\\
    0 & 1 & 0 & 0 & 0 & 1\\
    0 & 0 & 1 & 1 & 0 & 0\\
    1 & 0 & 1 & 0 & 1 & 1\\
    1 & 1 & 0 & 1 & 0 & 1\\
    0 & 1 & 1 & 1 & 1 & 0
  \end{pmatrix}.
\end{eqnarray}

\subsection{The chain complex representation of CSS codes}

Chain complexes from homological algebra form the structural backbone of our construction. They already occupy a central role in quantum error correction, and here serve as the primary framework for designing CC codes.

A length-$t$ chain complex $\AAA = (\LR{A_i},\LR{\partial_i})$ is a
sequence of $R$-modules, $A_0,\dots, A_t$ and homomorphisms $\partial_i : A_i\to A_{i-1}$,
\begin{equation*}
  \begin{tikzcd}[row sep=scriptsize, column sep = scriptsize]
    A_t \arrow[r, "\partial_t"] & A_{t-1} \arrow[r, "\partial_{t-1}"] & \cdots \arrow[r, "\partial_{2}"] & A_{1} \arrow[r, "\partial_{1}"] & A_0
  \end{tikzcd},
\end{equation*}
satisfying the boundary
condition $\partial_i\circ\partial_{i+1} = 0$. The $i$\textsuperscript{th} homology group, $H_i(\AAA)$, and corresponding $i$\textsuperscript{th} cohomology group, $H^i(\AAA)$, of this chain complex are defined as
\eqa{
  H_i(\AAA) &= \ker (\partial_i)/\im (\partial_{i+1}),\\
  H^i(\AAA) &= \ker (\partial_{i+1}^*)/\im (\partial_{i}^*).
}

Applying the binary representation entry-wise yields a length-$t$ chain complex $\BB(\AAA)$ over $\FF_2$,
\begin{equation*}
  \begin{tikzcd}[row sep=scriptsize, column sep = scriptsize]
    \BB(A_t) \arrow[r, "\BB(\partial_t)"] & \BB(A_{t-1}) \arrow[r, "\BB(\partial_{t-1})"] & \cdots \arrow[r, "\BB(\partial_{2})"] & \BB(A_{1}) \arrow[r, "\BB(\partial_{1})"] & \BB(A_0).
  \end{tikzcd}
\end{equation*}
Since $\BB(\cdot)$ preserves matrix multiplication, the boundary condition is preserved under binary lifting. Taking the binary representation of the homology and cohomology groups yields exactly the corresponding groups of $\BB(\AAA)$.

The simplest example here is a classical error-correcting code as a length-$1$ chain complex,
\begin{equation*}
  \CCC:
  \begin{tikzcd}[row sep=scriptsize, column sep = scriptsize]
    C_{1} \arrow[r, "H"] & C_0
  \end{tikzcd},
\end{equation*}
where $C_1 = R^\nu$, $C_0=R^m$ and $H\in R^{m\x \nu}$ is the PCM over $R$. By taking binary representation, $\BB(\CCC)$ recovers the standard binary format of the code with PCM $\BB(H) \in \FF_2^{(lm)\x(l\nu)}$. Binary classical codes
are the special cases when $R=\FF_2$.

An $\code{n,k,d}$ quantum stabiliser code $\QQQ$ is a $2^k$-dimensional logical subspace of the $n$-qubit physical Hilbert space \cite{Gottesman1997, Roads, Roffe_2019, GottesmanNewReview}. It is specified as the common $+1$ eigenspace of an Abelian subgroup $\SSS$ of $\PPP_n$ that does not contain $-I$. Logical operators are Pauli operators that act non-trivially on logical qubits while preserving the stabilisers. They are elements of the quotient group $\NNN(\SSS)/\SSS$, where $\NNN$ denotes the normaliser in $\PPP_n$. The distance $d$ is the minimum number of physical qubits on which a non-trivial logical Pauli operator acts, modulo phases and stabilisers. For CSS codes, $d_X$ and $d_Z$ are the minimum weights of non-trivial $X$- and $Z$-type logical operators, respectively, with $d=\min\{d_X,d_Z\}$.

The data and auxiliary codes for same-type measurements considered here are CSS codes. They are an important family of stabiliser codes whose generators can be chosen
to be purely $X$- or $Z$-type, described by PCMs
$H_X\in R^{m_X\times \nu}$ and $H_Z\in R^{m_Z\times \nu}$, respectively.
They admit a natural description as length-$2$ chain complexes,
\begin{equation*}
  \QQQ:
  \begin{tikzcd}[row sep=scriptsize, column sep = scriptsize]
    Q_2 \arrow[r, "H_Z^*"] & Q_{1} \arrow[r, "H_X"] & Q_0
  \end{tikzcd}.
\end{equation*}
The binary representation $\BB(\QQQ)$ yields the standard stabiliser formalism for qubits, with $\BB(Q_1)$ corresponding to the $n=l\nu$ physical qubits, and $\BB(Q_2)$ and $\BB(Q_0)$ corresponding to $Z$- and $X$-type stabiliser checks, respectively. Each row of $\BB(H_X)$ ($\BB(H_Z)$) specifies an $X$-type ($Z$-type) stabiliser.

Stabiliser commutativity is equivalent to the condition

\eqa{
  \BB(H_X)\BB(H_Z)^\intercal = 0
  \iff H_X H_Z^* = 0 ,
}
which coincides with the boundary condition of the chain complex and is commonly referred to as the CSS condition.

For CSS codes, logical $X$ and $Z$ operators naturally appear as the cohomology and homology groups, respectively, of the associated chain complex,
\begin{eqnarray}
  \LLL_X\coloneqq H^1(\QQQ) &= \ker H_Z/\im H_X^*,\\
  \LLL_Z\coloneqq H_1(\QQQ) &=  \ker H_X/\im H_Z^*.
\end{eqnarray}

\subsection{Lifted product code}
The lifted product is a technique for constructing CSS codes from pairs of \emph{classical seed codes}. In this section, we derive it from the starting point of chain complexes.

Let $\AAA = (\LR{A_i},\LR{\partial_i^A})$ and $\BBB = (\LR{B_i},\LR{\partial_i^B})$ be length-$m_A$ and length-$m_B$ chain complexes over $R$, respectively. Their tensor product $\AAA \ox \BBB$ is defined as a length-$(m_A+m_B)$ chain complex over $R$ with modules
$$\LR{(\AAA \ox \BBB)_m = \oplus_{i+j=m} A_i\ox B_j}_{m=0}^{m_A+m_B}$$ and homomorphisms $$\LR{\partial^{AB}_m: (\AAA \ox \BBB)_m \to (\AAA \ox \BBB)_{m-1}}_{m=1}^{m_A+m_B},$$
where $\partial^{AB}_m$ is defined as
\eqa{
  \partial^{AB}_m (a_i\ox b_j) = \partial^A_i (a_i) \ox b_j + a_i\ox \partial_j^B (b_j),
}
for $a_i\in A_i$, $b_j\in B_j$ and $i+j =m$~\cite{Algebraic_topo}. This construction is also referred to as the total complex of the product
complex.
We then define the \emph{lifted product} (LP) codes
in this framework.
\begin{definition}[Lifted product codes]
  \label{Definition: LP code}
  Given classical codes $\AAA : A_1\xrightarrow{H_a} A_0$ and $\BBB: B_1\xrightarrow{H_b} B_0$ over $R$, with $H_a\in R^{m_a\times n_a}$ and $H_b\in R^{m_b\times n_b}$, a lifted product code $\QQQ=\lpab$ is the total complex of $\AAA\ox\BBB$,
  \begin{equation}
    \begin{aligned}
      \label{eq: LP chain complexes over R}
      &\begin{tikzcd}[row sep=scriptsize, column sep = scriptsize]
        Q_0 \\
        Q_1 \arrow[u, "H_X"]\\
        Q_2 \arrow[u, "H_Z^*"]
      \end{tikzcd}
      \begin{tikzcd}[row sep=scriptsize, column sep = scriptsize]
        & A_0\ox B_0 & \\
        A_1\ox B_0 \arrow[ur,"H_a\ox I"] & & A_0\ox B_1 \arrow[ul,swap,"I\ox H_b"]\\
        & A_1\ox B_1 \arrow[ul,"I \ox H_b"]\arrow[ur,swap,"H_a\ox I"]
      \end{tikzcd},\\
      &\text{where }Q_1 = A_1\ox B_0 \oplus A_0\ox B_1.
    \end{aligned}
  \end{equation}
  Its CSS PCMs are
  \begin{equation}
    \label{eq: LP check matrices}
    \begin{aligned}
      H_X&=\left(H_a\ox I_{m_b}\ \middle|\ I_{m_a}\ox H_b\right),\\
      H_Z&=\left(I_{n_a}\ox H_b^*\ \middle|\ H_a^*\ox I_{n_b}\right).
    \end{aligned}
  \end{equation}
  The vertical bar denotes horizontal concatenation. The CSS condition follows from
  \begin{equation*}
    H_XH_Z^*=H_a\ox H_b+H_a\ox H_b=0.
  \end{equation*}

\end{definition}
Owing to the tensor product structure of LP codes, the logical operators
can be characterised using the Künneth formula, a
statement relating the homology of two objects to the homology of their product, in particular when
$R=\FF_2[G]$, for a finite odd order cyclic group $G$.
\begin{lemma}[Künneth formula~\cite{Algebraic_topo,Balanced_product}]
  Let $R=\FF_2[G]$ with $G$ a finite odd order cyclic group. The (co-)homology group of the product complex $\QQQ = \AAA \ox \BBB$ is subject to the isomorphism
  \eqa{
    H^m(\AAA \ox \BBB) &\cong \bigoplus_{p+q =m } H^p(\AAA) \ox H^q(\BBB),\\
    H_m(\AAA \ox \BBB) &\cong \bigoplus_{p+q =m } H_p(\AAA) \ox H_q(\BBB).
  }
\end{lemma}
For the rings covered by the lemma, the logical spaces decompose as
\begin{equation}
  \label{eq: logical operator partitioning}
  \begin{aligned}
    \LLL_X&\cong(R^{n_a}/\im H_a^*)\ox\ker H_b^*\oplus\ker H_a^*\ox(R^{n_b}/\im H_b^*),\\
    \LLL_Z&\cong\ker H_a\ox(R^{m_b}/\im H_b)\oplus(R^{m_a}/\im H_a)\ox\ker H_b.
  \end{aligned}
\end{equation}
The first and second summands correspond to the left and right qubit sectors $A_1\ox B_0$ and $A_0\ox B_1$ in Eq.~\eqref{eq: LP chain complexes over R}. We denote them by $\LLL_{X,Z}^{\lef}$ and $\LLL_{X,Z}^{\rig}$.
\begin{remark}[Hypergraph product codes]
  The HGP codes are LP codes defined over the binary field, i.e., $R=\FF_2$.
\end{remark}

\section{Logical structure of CC codes}
\label{Appendix: ker and im of CC seeds}
Here we establish the logical structure of CC codes using the joint condition on the seed pair in Definition~\ref{Definition: cc code}. Let $R=\FF_2[x]/(x^l+1)$, where $l\geq 3$ is odd, and put $\chi=1+x+\cdots+x^{l-1}$.

\begin{lemma}[Decomposition of the cyclic group algebra]
\label{Lemma: diag x+1 ker and im}
There is a decomposition with Chinese remainder theorem~\cite{Lee2026LogicalSpectroscopy}
\begin{equation}
  R\cong\FF_2\oplus\prod_{f\mid\chi}K_f,
  \qquad K_f=\FF_2[x]/(f),
\end{equation}
where the product runs over the distinct monic irreducible factors of $\chi$, and the first component is evaluation at $x=1$. Under this decomposition, $\chi$ corresponds to $(1,0,\ldots,0)$. In particular,
\begin{equation}
  \chi^2=\chi,\ \chi R=\{0,\chi\} \tand \chi h=h(1)\chi
\end{equation}
for every $h\in R$.
\end{lemma}
\begin{proof}
Since $l$ is odd, the derivative of $x^l+1$ over $\FF_2$ is $x^{l-1}$, which is coprime to $x^l+1$. Thus $x^l+1$ is squarefree. Moreover, $\chi(1)=1$, so $x+1$ and $\chi$ are coprime. The polynomial $\chi$ is $1$ on the $x+1$ component and $0$ on every other component, proving the remaining identities.
\end{proof}

\begin{corollary}[Kernels, images and binary ranks of CC check matrices]
\label{Corollary: ker and im of Ha}
For a CC code satisfying Definition~\ref{Definition: cc code}, the PCMs obey
\begin{equation}
  \begin{aligned}
    \ker H_Z^*&=\chi R^{n_an_b},\\
    \ker H_X^*&=\chi R^{n_an_b},\\
    \ker H_X&=\im H_Z^*\oplus\chi R^{2{n_an_b}},\\
    \ker H_Z&=\im H_X^*\oplus\chi R^{2{n_an_b}}.
  \end{aligned}
\end{equation}
Their binary ranks are
\begin{equation}
  \rank\BB(H_X)
  =\rank\BB(H_Z)={n_an_b}(l-1).
\end{equation}
\end{corollary}

\begin{proof}
Write $A=H_a\otimes I_{n_b}$ and $B=I_{n_a}\otimes H_b$. Then $AB=BA$, $H_X=(A\ B)$, and $H_Z^*=\binom{B}{A}$. On the $x+1$ component, both seed matrices vanish because each entry is a binomial or zero.

On each other component $K_f$, Definition~\ref{Definition: cc code} makes at least one of $A_f$ and $B_f$ invertible. Hence the sequence
\begin{equation}
  0\longrightarrow K_f^{n_an_b}
  \xrightarrow{\binom{B_f}{A_f}} K_f^{2{n_an_b}}
  \xrightarrow{(A_f\ B_f)}K_f^{n_an_b}
  \longrightarrow 0
\end{equation}
is exact: the first map is injective, the second map is surjective, and their composition is zero since $A_fB_f+B_fA_f=0$. The image and kernel in the middle both have dimension ${n_an_b}$, so they coincide. Thus all homology is supported on the $x+1$ component. By Lemma~\ref{Lemma: diag x+1 ker and im}, that component is represented by $\chi R$, giving
\begin{equation}
  \ker H_Z^*=\chi R^{n_an_b} \tand
  \ker H_X=\im H_Z^*\oplus\chi R^{2{n_an_b}}.
\end{equation}
The involution $x\mapsto x^{-1}$ fixes the $x+1$ component and permutes the remaining field components. Applying conjugate transpose therefore gives the same componentwise exactness for the dual complex, proving the assertions for $\ker H_X^*$ and $\ker H_Z$.

Each PCM has rank ${n_an_b}$ over every nontrivial field component and rank zero on the $x+1$ component. The nontrivial components have total binary dimension $l-1$, so each binary PCM has rank ${n_an_b}(l-1)$.
\end{proof}

\section{Same-type product surgery is distance preserving for HGP codes}
\label{Appendix: HGP distance preserving}

Consider HGP codes as the specialisation of Definition~\ref{Definition: LP code} and the surgery construction of Lemma~\ref{Lemma: surgery for product codes} with $R=\FF_2$. Thus $M^*=M^\intercal$ for every matrix $M$, and all tensor products below are over $\FF_2$. In this appendix, we prove that $d(\widetilde{\MMM})\geq d(\QQQ)$.

Following the sector ordering in Eq.~\eqref{eq: LP chain complexes over R}, we regard a vector in $Q_1$ as a pair of arrays of sizes $n_a\times m_b$ and $m_a\times n_b$. We first record the support properties of the data-code logical operators~\cite{TillichZemor,Armanda_partitioning,distance_preserving_syndrome_HGP,Reshape}. Let
\begin{equation*}
  \begin{aligned}
    K_a&=\ker H_a,& C_a&=A_0/\im H_a,\\
    K_b&=\ker H_b,& C_b&=B_0/\im H_b,
  \end{aligned}
\end{equation*}
and denote the quotient maps by $\pi_a:A_0\to C_a$ and $\pi_b:B_0\to C_b$. Let $d_a$ and $d_b$ be the minimum weights of non-zero vectors in $K_a$ and $K_b$, respectively.
Under the K\"unneth decomposition in Eq.~\eqref{eq: logical operator partitioning}, the homology class $[c]\in H_1(\QQQ)$ of a cycle $c=(c_{\lef},c_{\rig})\in\ker H_X$ has components
\begin{equation}
  \label{eq: HGP homology projections}
  (I\otimes\pi_b)c_{\lef}\in K_a\otimes C_b,\quad
  (\pi_a\otimes I)c_{\rig}\in C_a\otimes K_b.
\end{equation}
Also, $\ker H_Z^\intercal=K_a\otimes K_b$. If the first component of $[c]$ is non-zero, $c_{\lef}$ has at least $d_a$ non-zero rows. If the second component is non-zero, $c_{\rig}$ has at least $d_b$ non-zero columns. Thus the distance is the minimum over the logical sectors that are present, i.e., when $K_a\otimes C_b\neq 0$ or $C_a\otimes K_b\neq 0$:
\begin{equation}
  d_Z(\QQQ)=\min\{
  d_a,d_b\}.
  \label{eq: HGP data distance}
\end{equation}

  \begin{lemma}[Localisation of a stabiliser]
    \label{Lemma: HGP stabiliser localisation}
    Let $u\in A_1\otimes B_1$. Suppose the left component of the stabiliser $H_Z^\intercal u$ has non-zero rows only in $I\subseteq\{1,\ldots,n_a\}$, and its right component has non-zero columns only in $J\subseteq\{1,\ldots,n_b\}$. If $|I|<d_a$ and $|J|<d_b$, there exists $k\in\ker H_Z^\intercal$ such that $u_0=u+k$ is supported on $I\times J$. In particular, $H_Z^\intercal u_0=H_Z^\intercal u$.
  \end{lemma}
  \begin{proof}
    Let $\rho_I:A_1\to A_1|_{I^c}$ and $\rho_J:B_1\to B_1|_{J^c}$ retain the coordinates outside $I$ and $J$, respectively. The distance assumptions imply that $\rho_I$ is injective on $K_a$ and $\rho_J$ is injective on $K_b$.

    Since $H_Z^\intercal u=((I\otimes H_b)u,(H_a\otimes I)u)$, every row of $u$ outside $I$ lies in $K_b$ and every column outside $J$ lies in $K_a$. The block $w=(\rho_I\otimes\rho_J)u$ therefore belongs to
    \begin{equation*}
      (A_1|_{I^c}\otimes\rho_JK_b)\cap(\rho_IK_a\otimes B_1|_{J^c}) =\rho_IK_a\otimes\rho_JK_b.
    \end{equation*}

    Choose $k\in K_a\otimes K_b$ whose restriction to this block is $w$. Both $(\rho_I\otimes I)u$ and $(\rho_I\otimes I)k$ belong to $A_1|_{I^c}\otimes K_b$, and they agree after applying $I\otimes\rho_J$. Since this map is injective on that space, they agree on every row outside $I$. Interchanging the two factors proves agreement on every column outside $J$. Thus $u+k$ is supported on $I\times J$, and $H_Z^\intercal k=0$.
  \end{proof}

  \begin{theorem}[Distance preservation for HGP product surgery]
    \label{Theorem: HGP product surgery distance}
    For the HGP data code and product connection code defined above, the merged code satisfies
    \begin{equation*}
      d_Z(\widetilde{\MMM})\geq d_Z(\QQQ) \tand
      d_X(\widetilde{\MMM})\geq d_X(\QQQ).
    \end{equation*}
    In particular, $d(\widetilde{\MMM})\geq d(\QQQ)$.
  \end{theorem}
  \begin{proof}
    We prove here the $Z$-type bound, whereas the $X$-type bound follows the general proof~\cite[Theorem 1]{Xanadu_Homological_measurements}. Let $z=(p,q)\in\ker \widetilde{H_X}$ have weight strictly smaller than $d_Z(\QQQ)$. We show that $z\in\im \widetilde{H_Z}^\intercal$. Eq.~\eqref{eq: product merged checks} gives
    \begin{equation*}
      H_X q=0 \tand H_X p+H_X' q=0.
    \end{equation*}
    Since $\wt(q)\leq\wt(z)<d_Z(\QQQ)$, the auxiliary-patch support $q$ forms a $Z$-type stabiliser. Write $q=H_Z^\intercal u$ and put $c=p+{H_Z'}^\intercal u$, then $H_X c=0$, as $H_X{H_Z'}^\intercal = H_X'H_Z^\intercal$.

    Let $I$ be the union of the non-zero row indices of $p_{\lef}$ and $q_{\lef}$, and $J$ the union of the non-zero column indices of $p_{\rig}$ and $q_{\rig}$. Then
    \begin{equation}
      |I|,|J|\leq\wt(z)<d_Z(\QQQ).
      \label{eq: HGP small support}
    \end{equation}
    We distinguish the possible logical sectors in Eq.~\eqref{eq: HGP homology projections}.

    Suppose first that both logical sectors $K_a\otimes C_b$ and $C_a\otimes K_b$ are non-zero. Then $d_Z(\QQQ)=\min\{d_a,d_b\}$, so Lemma~\ref{Lemma: HGP stabiliser localisation} gives a choice $u_0$ supported on $I\times J$ with $H_Z^\intercal u_0 =H_Z^\intercal u =q$. The cycle $c_0=p+{H_Z'}^\intercal u_0$ has non-zero left-sector rows only in $I$ and non-zero right-sector columns only in $J$. However, by Eq.~\eqref{eq: HGP homology projections}, the homology class $[c_0]$ has two components. For $c_0$ to be a logical operator of the data code, the left and right components must have either $|I|\geq d_a$ or $|J|\geq d_b$, contradicting Eq.~\eqref{eq: HGP small support}. Thus there exists $v$ such that $c_0=H_Z^\intercal v$ is a data stabiliser and
    \begin{equation*}
      z=\widetilde{H_Z}^\intercal\begin{pmatrix}v\\u_0\end{pmatrix}
    \end{equation*}
    is a stabiliser of the merged code.

    When only one logical sector of the data code is non-zero, the other seed distance may be smaller than \(d_Z(\mathcal Q)\), so the two distance conditions required by Lemma~\ref{Lemma: HGP stabiliser localisation} need not both hold. We therefore keep track of the freedom in the choice of $u$. Define
    \begin{equation*}
      T=\{[{H_Z'}^\intercal k]:k\in\ker H_Z^\intercal\}\subseteq H_1(\QQQ),
    \end{equation*}
    which records the data logical classes that become merged stabilisers. Consequently,
    \begin{equation}
      z\in\im \widetilde{H_Z}^\intercal\quad\Longleftrightarrow\quad[c]=0
      \text{ in }H_1(\QQQ)/T.
      \label{eq: HGP surviving quotient}
    \end{equation}
    Indeed, if $c=H_Z^\intercal v+{H_Z'}^\intercal k$ with $k\in\ker H_Z^\intercal$, then $z=\widetilde{H_Z}^\intercal\binom{v}{u+k}$. Conversely, any expression $z=\widetilde{H_Z}^\intercal\binom{v}{u_0}$ has $u+u_0\in\ker H_Z^\intercal$ and gives $c=H_Z^\intercal v+{H_Z'}^\intercal (u+u_0)$.

    Suppose $K_a\otimes C_b\neq0$ and $C_a\otimes K_b=0$, so $d_Z(\QQQ)=d_a$. To identify $T$, take $a\in K_a$ and $y\in K_b$. The generator combination $k=a\otimes y$ has zero auxiliary support, since $H_Z^\intercal k=0$. Its data attachment is
    \begin{equation*}
      {H_Z'}^\intercal(a\otimes y)
      =\begin{pmatrix}
        a\otimes H_b'y\\
        H_a'a\otimes y
      \end{pmatrix}.
    \end{equation*}
    By Eq.~\eqref{eq: HGP homology projections}, the left logical component is $a\otimes\pi_b(H_b'y)$ and the right logical component vanishes because $C_a\otimes K_b=0$.
    Define $\beta:K_b\to C_b$ by $\beta(y)=\pi_b(H_b'y)$. Since the products $a\otimes y$ span $\ker H_Z^\intercal=K_a\otimes K_b$, their data logical images give
    \begin{equation*}
      \begin{aligned}
        T&=K_a\otimes\im\beta,\\
        H_1(\QQQ)/T&\cong K_a\otimes(C_b/\im\beta).
      \end{aligned}
    \end{equation*}
    Suppose for contradiction that $z$ is not a merged stabiliser. Eq.~\eqref{eq: HGP surviving quotient} then gives $[c]\neq0$ in this quotient. We represent this class by projecting each row of $c_{\lef}$ to $C_b/\im\beta$, and show that every row outside $I$ becomes zero.

    For $i\notin I$, we have $p_{\lef,i}=q_{\lef,i}=0$. Since $q_{\lef,i}=H_bu_i$, the row $u_i$ belongs to $K_b$. Hence
    \begin{equation*}
      c_{\lef,i}=H_b'u_i\tand
      \pi_b(c_{\lef,i})=\beta(u_i)\in\im\beta.
    \end{equation*}
    This row is therefore zero in $C_b/\im\beta$, even though its physical support may be non-zero. Thus the projected array has non-zero rows only in $I$.

    Choose any basis of $C_b/\im\beta$ to write this array as a binary matrix. Since $[c]\in K_a\otimes(C_b/\im\beta)$, each column is a codeword in $K_a$. The matrix is non-zero, so it has a non-zero column $x$. This column is supported on $I$, giving
    \begin{equation*}
      d_a\leq\wt(x)\leq|I|\leq\wt(z)<d_a,
    \end{equation*}
    a contradiction.

    If instead $C_a\otimes K_b\neq0$ and $K_a\otimes C_b=0$, then $d_Z(\QQQ)=d_b$. Interchanging the two tensor factors gives the same argument with rows replaced by columns.
    Thus, in either case, Eq.~\eqref{eq: HGP surviving quotient} implies that $z$ is a merged stabiliser.

  \end{proof}

\section{Mixed-type product surgery}
  \label{Appendix: mixed Pauli surgery}

  We use the general attachment form of Eq.~\eqref{eq: mixed product merged checks} as in Remark~\ref{remark: general attachment}, where the product connection is recovered by taking $f_1={H_Z'}^*$ and $f_0=H_X'$. Take binary representations of all maps and suppress $\BB$. Let $q$ be the number of physical merge generators. Choose commuting Hermitian data Paulis $P_j$ with supports $(g_1e_j,f_1e_j)$ and set $P(s)=\prod_{j=1}^qP_j^{s_j}$. The merge generators are
  \[
    G_j=P_j\otimes Z^{H_Z^*e_j},\qquad j=1,\ldots,q.
  \]
  Prepare the auxiliary patch in $\ket{+}^{\otimes n}$ and measure the merge generators with outcome bits $b\in\FF_2^q$. Split by measuring all auxiliary qubits in the $X$ basis, with outcomes $t\in\FF_2^n$. Expanding the merge projectors gives the resulting conditional Kraus operator on the data patch
  \[
    \begin{aligned}
      K_{b,t}
      &={}_{\mathrm{auxi}}\bra{t_X}\prod_{j=1}^q
      \frac{I+(-1)^{b_j}G_j}{2}\ket{+}_{\mathrm{auxi}}^{\otimes n}\\
      &=2^{-q}\sum_{H_Z^*s=t}(-1)^{b^\intercal s}P(s).
    \end{aligned}
  \]
  For any outcome record with non-zero probability, $t\in\im H_Z^*$, and choose $s_0$ satisfying $H_Z^*s_0=t$. Then
  \[
    K_{b,t}=2^{-q}(-1)^{b^\intercal s_0}P(s_0)
    \sum_{u\in\ker H_Z^*}(-1)^{b^\intercal u}P(u).
  \]
  Absorbing $P(s_0)$ into the data Pauli frame removes the readout-dependent byproduct and leaves, up to a scalar, the joint eigenspace projector of $P(u)$ for $u\in\ker H_Z^*$. The dressing identities ensure that these operators normalise the data stabilisers. Images that are already data stabilisers constrain only known outcome signs; the independent logical images give exactly the PPMs in Lemma~\ref{Subsec: limitations}. Thus the auxiliary readout adds no further logical measurement. In the updated frame, the data remains in the original codespace with the required logical projection, and syndrome extraction of $\QQQ$ resumes.

  Specifically for CC codes, we return to matrices over $R$. Corollary~\ref{Corollary: ker and im of Ha} gives
  \[
    \ker\BB(H_Z^*)=\im\BB(\chi I_{n_an_b}).
  \]
  A row of $(g_1^*\mid f_1^*)$ with entries in $\{0,1\}\subset R$ therefore specifies the logical Pauli obtained by multiplying its $l$ lifted merge checks. Every non-zero such full-cluster Pauli has odd parity in at least one $X$ or $Z$ component, whereas data stabilisers have even parity in every cluster. Thus independent attachment rows give independent logical PPMs. Any commuting family of at most $n_an_b$ PPMs can be placed in these rows, padding with zeros as necessary. For a logical product containing $Y$, the product of the physical checks may equal the intended logical operator or its negative. This sign is fixed in advance, given lift parameter $l$:
  \[
    Y^{\otimes l}=i^{l}\bX\bZ = (-1)^{(l-1)/2}\overline{Y}.
  \]
  Compatible dressings exist because $\BB(H_Xf_1)$ and $\BB(H_Zg_1)$ annihilate $\ker\BB(H_Z^*)$: the attachments send its generators to logical Pauli components. Both maps factor through $\BB(H_Z^*)$, giving binary solutions for $f_0,g_0$. Summing the $l$ cyclic conjugates of each solution preserves the equations, since $l$ is odd, and gives solutions commuting with the cyclic shifts.

  For the nine parallel uniform pair measurements in Corollary~\ref{Corollary: mixed pair merges}, a separated dressing is available when there exists a pair of matrices $T_a,T_b$ over $R$ such that
  \[
    T_aH_a=H_a^*,\qquad T_bH_b=H_b^*.
  \]
  The two equations admit solutions precisely when $\ker H_a\subseteq\ker H_a^*$ and $\ker H_b\subseteq\ker H_b^*$, conditions that hold for all CC code seed pairs introduced in Tables~\ref{Table: cc_seed_matrices_8} and~\ref{Table: cc_seed_matrices_18}.

  Using the Pauli indicators of Method~\ref{Subsec: mixed Pauli surgery}, take
  \[
    \begin{aligned}
      f_0&=(q_ZI_{n_an_b}\mid p_ZI_{n_an_b}),\\
      g_0&=\left(p_X(I_{n_a}\ox T_b)\ \middle|\
      q_X(T_a\ox I_{n_b})\right).
    \end{aligned}
  \]
  Direct multiplication by $H_Z^*$ verifies the dressing identities. Sparse $T_a,T_b$ give sparse dressings for these merges.

  \section{Protected data subsystem and distance preservation for product surgery of CC codes}
  \label{Appendix: weak expansion}

  We prove Theorem~\ref{Theorem: mixed protected distance} for the general attachment form of Eq.~\eqref{eq: mixed product merged checks} as in Remark~\ref{remark: general attachment} for CC codes. First we identify the surviving data subsystem and justify assigning the independent auxiliary logical modes to gauge. We then use weak expansion to show that every merged normaliser of weight below $d$ acts trivially on that data subsystem.

  \subsection{Identification of the protected data subsystem}

  Let $\widetilde{\mathcal S}$ be the merged stabiliser group. Since $l$ is odd, the binary support space of each cluster is the orthogonal direct sum of the full-cluster subspace, spanned by $\chi$, and the even-parity subspace. Every matrix over $R$ preserves these subspaces. The lifted stabiliser row spaces are also closed under cyclic shifts, so they split in the same way.

  On the full-cluster subspaces, $a\in R$ acts by its coefficient parity $a(1)$, since $a\chi=a(1)\chi$. In particular, $H_X$ and $H_Z$ vanish there. Restricting the merged PCM to these subspaces gives
  \[
    \widetilde H(1)=
    \left(\begin{array}{cc|cc}
      0&f_0(1)&0&0\\
      0&g_0(1)&0&0\\
      g_1^*(1)&0&f_1^*(1)&0\\
      0&0&0&0
    \end{array}\right).
  \]
  The constraints therefore separate into $M$ commuting data logical measurements and auxiliary-only $X$ constraints, where
  \[
    M=\rank\bigl(g_1^*(1)\mid f_1^*(1)\bigr).
  \]
  For each consistent outcome, the data constraints leave $k-M$ logical qubits.

  The even-parity subspaces carry no logical modes. Put $r=n_an_b(l-1)$. On these subspaces, exactness of the CC complex gives rank $r$ for both $H_X$ and $H_Z$, and $H_Z^*$ is injective on the even-parity coefficient vectors. The merge group has rank $r$ and is independent of the other groups through its auxiliary $Z$ block. The two dressed old-check groups contribute rank $r$ each through their data $X$ and $Z$ blocks. The retained auxiliary $X$ group contributes a further independent rank $r$. Thus the stabiliser rank is $4r$, equal to the number of effective qubits in the even-parity subspaces.

  Consequently, the merged codespace has a subsystem decomposition
  \[
    \mathcal C_{\mathrm{merged}}
    \cong\mathcal H_{\mathrm{data}}\otimes\mathcal H_{\mathrm{auxi}},
    \qquad \dim\mathcal H_{\mathrm{data}}=2^{k-M},
  \]
  where $\mathcal H_{\mathrm{auxi}}$ contains the remaining auxiliary logical modes. Assign these modes to gauge. Equivalently, writing $C(\widetilde{\mathcal S})$ for the centraliser of the stabiliser group and $\mathcal P_{\mathrm{auxi}}$ for Pauli operators supported only on the auxiliary patch,
  \[
    \mathcal G=
    \bigl\langle\widetilde{\mathcal S},
    \mathcal P_{\mathrm{auxi}}\cap
    C(\widetilde{\mathcal S})\bigr\rangle.
  \]
  On the codespace, every element of $\mathcal G$ acts as identity on $\mathcal H_{\mathrm{data}}$.

  \subsection{Proof of the distance bound}

  Work in the symplectic representation. Let $E=(x,u\mid z,v)$ commute with all merged stabilisers and satisfy $\wt(E)<d$, where $x,z$ act on data and $u,v$ on auxiliary qubits. Commutation with the retained auxiliary $X$ checks gives $H_Xv=0$. Since the auxiliary code has distance at least $d$ and $\wt(v)<d$, $Z^v$ is an auxiliary stabiliser, so $v=H_Z^*s$. Choose a representation using the fewest merge generators, giving $\wt(s)=\rho(s)$.

  Multiplying $E$ by the merge generators selected by $s$ gives an equivalent error $E'$ with no auxiliary $Z$ component. Its data component is $P(x+g_1s,z+f_1s)$, and Eq.~\eqref{eq: weak expansion} gives
  \[
    \begin{aligned}
      \wt(E'_{\mathrm{data}})
      &\leq\wt(E_{\mathrm{data}})+w_{\mathrm{att}}\rho(s)\\
      &\leq\wt(E_{\mathrm{data}})+\wt(v)
      \leq\wt(E)<d.
    \end{aligned}
  \]
  The remaining auxiliary component is $X$-type, and hence $E'_{\mathrm{data}}$ commutes with the original data stabilisers. The data distance implies that $E'_{\mathrm{data}}$ is a data stabiliser. Write
  \[
    E'=S_{\mathrm{data}}\otimes A_{\mathrm{auxi}}\tand
    \widetilde S_{\mathrm{data}}=S_{\mathrm{data}}\otimes D_{\mathrm{auxi}}
    \in\widetilde{\mathcal S}.
  \]
  The second operator exists because the original data-check generators are retained with dressings. Multiplication gives
  \[
    E'\widetilde S_{\mathrm{data}}
    =I_{\mathrm{data}}\otimes
    A_{\mathrm{auxi}}D_{\mathrm{auxi}}.
  \]
  This residual operator commutes with all merged stabilisers and is supported only on the auxiliary patch, hence belongs to $\mathcal G$. Since both multiplications used merged stabilisers, $E$ also belongs to $\mathcal G$. We have therefore excluded every element of $C(\widetilde{\mathcal S})\setminus\mathcal G$ below weight $d$, proving the protected-distance bound.

  The same proof applies to the \(X/Z\)-exchanged construction, with the corresponding weak expansion condition on \(H_X^*\).

  \section{Mixed-type surgeries and distance verification of representative CC codes}
  \label{Appendix: sparse mixed surgery}

  For every fixed seed pair in Tables~\ref{Table: cc_seed_matrices_8} and~\ref{Table: cc_seed_matrices_18}, we construct all $3k$ individually addressed logical $X,Y,Z$ measurements and the nine uniform pair rounds. A round $PQ$ measures precisely $\overline P_i^{\lef}\overline Q_i^{\rig}$ for $i=1,\ldots,k/2$. This gives $33$ constructions per $k=8$ code and $63$ per $k=18$ code, or $579$ in total. Each uses one auxiliary copy of the memory code with its original $H_X,H_Z$ blocks.

  \subsection{Polynomial attachments and auxiliary correction}

  For $a\in R$, the identity $a\chi=a(1)\chi$ shows that polynomial attachments act on the clustered logical basis through their coefficient parities. Even-parity additions can therefore reduce physical check weights without changing the logical attachment. Some selected constructions also use an auxiliary $X$ component $\Lambda$ in the merge checks in order to permit low check-weight implementations for $Y$ logical measurements, giving
  \begin{equation}
    \label{eq: corrected mixed merged checks}
    \widetilde H_{\Lambda}=
    \left(\begin{array}{cc|cc}
      H_X&f_0&0&0\\
      0&g_0&H_Z&0\\
      g_1^*&\Lambda&f_1^*&H_Z\\
      0&H_X&0&0
    \end{array}\right).
  \end{equation}
  The dressing equations remain
  \[
    H_Xf_1=f_0H_Z^*,\qquad H_Zg_1=g_0H_Z^*,
  \]
  while mutual commutation of the merge rows requires
  \[
    g_1^*f_1+f_1^*g_1=\Lambda H_Z^*+H_Z\Lambda^*.
  \]
  Setting $\Lambda=0$ recovers the construction in Appendix~\ref{Appendix: mixed Pauli surgery}.

  For $\Lambda\neq0$, the cleaning step adds $\Lambda^*s$ to the auxiliary $X$ support, leaving the data operator and its weight bound unchanged. Multiplying by the dressed representative of the data stabiliser again leaves an auxiliary-only normaliser. Thus the protected-distance bound in Appendix~\ref{Appendix: weak expansion} applies provided these residual operators act trivially on the surviving data logical qubits.

  The original data checks may also be replaced by another generating set with the same row spaces, provided the corresponding dressing identities hold and the auxiliary PCMs remain unchanged.

  \subsection{Sparse constructions and algebraic verification}

  For an addressed measurement, each non-zero Pauli component of the attachment is chosen as a tensor product of two seed-sized polynomial matrices embedded in one sector. For a parallel uniform pair measurement, each attachment component is built from tensor products between a seed-sized polynomial matrix on one factor and an identity matrix. The dressing maps are optimised over their compatible solution spaces.

  For every selected PCM, exact binary checks verify commutation and equality of the measured data span with the requested logical span, modulo the original data stabilisers and including the known auxiliary preparation. The subsystem checks identify $k-1$ protected data logical qubits for a single measurement and $k/2$ for a full pair round, with auxiliary-only normalisers acting trivially on them. These checks supply the additional identification required above for $\Lambda\neq0$.

  Table~\ref{Table: mixed surgery weights} gives the achieved physical check weights. All $579$ selected constructions have maximum check weight at most $13$, counting each physical $Y$ once.

  \begin{table*}[t]
    
      \centering
      \setlength{\tabcolsep}{6pt}
      \caption{\textbf{Check weights for the parallel product surgery of representative CC codes.} Each single-axis column is the maximum over all addressed logical positions, and the pair column is the maximum over the nine uniform $PQ$ rounds. A physical $Y$ counts once. Memory checks have weight $8$ throughout.}
      \label{Table: mixed surgery weights}
      \begin{tabular}{ccccccc}
        \toprule
        Code & $l$ & Single $\bX$ & Single $\bY$ & Single $\bZ$ & Pair rounds & Overall\\
        \midrule
        $\code{24,8,3}$ & $3$ & $10$ & $10$ & $10$ & $10$ & $10$\\
        $\code{40,8,5}$ & $5$ & $10$ & $10$ & $10$ & $10$ & $10$\\
        $\code{56,8,7}$ & $7$ & $10$ & $11$ & $10$ & $12$ & $12$\\
        $\code{72,8,9}$ & $9$ & $10$ & $11$ & $10$ & $12$ & $12$\\
        $\code{88,8,10}$ & $11$ & $11$ & $12$ & $11$ & $10$ & $12$\\
        $\code{104,8,11}$ & $13$ & $11$ & $13$ & $11$ & $13$ & $13$\\
        $\code{120,8,13}$ & $15$ & $11$ & $13$ & $11$ & $13$ & $13$\\
        $\code{136,8,15}$ & $17$ & $10$ & $12$ & $10$ & $13$ & $13$\\
        \midrule
        $\code{54,18,3}$ & $3$ & $10$ & $10$ & $10$ & $10$ & $10$\\
        $\code{90,18,5}$ & $5$ & $10$ & $10$ & $10$ & $10$ & $10$\\
        $\code{126,18,7}$ & $7$ & $10$ & $11$ & $10$ & $10$ & $11$\\
        $\code{162,18,9}$ & $9$ & $11$ & $11$ & $11$ & $10$ & $11$\\
        $\code{198,18,10}$ & $11$ & $10$ & $12$ & $10$ & $10$ & $12$\\
        \bottomrule
      \end{tabular}
      \par
    
  \end{table*}

  For the $\code{136,8,15}$ CC code seed pair in Table~\ref{Table: cc_seed_matrices_8}, the nine uniform pair rounds have the following maximum check weights:
  \begin{center}
    \begin{tabular}{@{}c ccc@{}}
      \toprule
      $P\backslash Q$ & $X$ & $Y$ & $Z$\\
      \midrule
      $X$ & $10$ & $12$ & $12$\\
      $Y$ & $12$ & $13$ & $12$\\
      $Z$ & $12$ & $12$ & $10$\\
      \bottomrule
    \end{tabular}
  \end{center}
  For this $\code{136,8,15}$ code, all pair rounds and single-$\bX/\bZ$ measurements use $\Lambda=0$, while the eight single-$\bY$ measurements use a non-zero correction. The only other corrected constructions are the eight single-$\bY$ measurements of $\code{40,8,5}$ and four of $\code{88,8,10}$. Thus $20$ of the $579$ constructions use $\Lambda\neq0$; all selected $k=18$ constructions have $\Lambda=0$.

  \subsection{Verification of weak expansion}

  In the binary representation, $\ker H_Z^*$ is spanned by the all-ones vectors on the length-$l$ coefficient blocks. Each block can therefore be complemented independently without changing $H_Z^*s$. Writing $s_j$ for block $j$ of $s$, we obtain
  \[
    \rho(s)=\sum_{j=1}^{n_an_b}\min\{\wt(s_j),l-\wt(s_j)\}.
  \]
  We enumerate all non-zero supports $v$ of weight below $d$ that have even parity in every physical cluster and satisfy $H_Xv=0$. These are precisely the auxiliary $Z$ stabiliser supports. For each $v$, we solve $H_Z^*s=v$ and evaluate $\rho(s)$ using the formula above. Exchanging $X$ and $Z$ gives the same results, listed in Table~\ref{Table: finite weak expansion}.

  \begin{table*}[t]
    
      \centering
      \setlength{\tabcolsep}{6pt}
      \caption{\textbf{Weak expansion for the representative CC codes.} Only non-zero auxiliary stabilisers of weight below $d$ are listed. Each tuple $(\wt(v),\rho;N)$ gives the stabiliser weight, minimum number of merge generators and number of distinct supports. The restricted factor is the minimum ratio $\wt(v)/\rho$, and the last column gives the maximum $w_{\mathrm{att}}$ over the selected surgeries. The weak expansion condition of Eq.~\eqref{eq: weak expansion} holds if and only if $\wt(v)/\rho \geq w_{\mathrm{att}}$.}
      \label{Table: finite weak expansion}
      \begin{tabular}{cccc}
        \toprule
        Code & Stabiliser spectrum & Restricted factor & Max. $w_{\mathrm{att}}$\\
        \midrule
        $\code{24,8,3}$ & none & automatic & $2$\\
        $\code{40,8,5}$ & none & automatic & $2$\\
        $\code{56,8,7}$ & none & automatic & $4$\\
        $\code{72,8,9}$ & $(8,1;36)$ & $8$ & $4$\\
        $\code{88,8,10}$ & $(8,1;44)$ & $8$ & $3$\\
        $\code{104,8,11}$ & $(8,1;52)$ & $8$ & $5$\\
        $\code{120,8,13}$ & $(8,1;60),\ (12,2;60)$ & $6$ & $5$\\
        $\code{136,8,15}$ & $(8,1;68),\ (14,2;816)$ & $7$ & $4$\\
        \midrule
        $\code{54,18,3}$ & none & automatic & $2$\\
        $\code{90,18,5}$ & none & automatic & $2$\\
        $\code{126,18,7}$ & none & automatic & $2$\\
        $\code{162,18,9}$ & $(8,1;81)$ & $8$ & $3$\\
        $\code{198,18,10}$ & $(8,1;99)$ & $8$ & $4$\\
        \bottomrule
      \end{tabular}
      \par
    
  \end{table*}

  For every code, the restricted factor is at least the largest attachment weight, or there is no non-zero auxiliary stabiliser below $d$. Thus Eq.~\eqref{eq: weak expansion} holds for all $579$ selected surgeries. Together with the memory distances and the subsystem verification above, this establishes their protected-distance bounds, including the constructions with $\Lambda\neq0$.

\section{Circuit-level simulations}
\label{app:simulation-details}

We construct syndrome-extraction circuits without any flag qubit using left--right scheduling and residual-error criteria~\cite{StrikisBrowneBeverland2026}, together with randomized construction and filtering~\cite{BhardwajEtAl2026}. Candidate gate orderings are checked for correct stabilizer measurement and qubit conflicts, screened using residual-error proxies and bounded low-weight fault searches, and compared through trial simulations in both bases.

Simulations use Stim~\cite{Gidney_2021}. Memory experiments prepare $k$ ideal encoded Bell pairs with noiseless references and perform $R=d$ syndrome-extraction rounds. The two-parallel-$\bcnot$ experiments use four data--reference Bell pairs and four logical helpers in $\lvert0\rangle$, followed by $XX$ surgery, $ZZ$ surgery, and auxiliary $X$ readout. Each stage uses $d$ extraction rounds and the corresponding gadget split, giving $R=3d$. Measurement-dependent Pauli frames are tracked, and the full circuit is decoded jointly. A shot fails if any scored logical correlation is decoded incorrectly.

We use a circuit-level depolarizing noise model motivated by the long coherence times and programmable entangling gates of neutral-atom arrays~\cite{BluvsteinEtAl2024,LukinFaultTolerant}. Each syndrome-extraction $\mathit{CNOT}$ is followed by two-qubit depolarization of strength $p$. Single-qubit depolarization of strength $p$ is applied after auxiliary reset and before measurement, including split and final readouts. The same channel is applied to all code-data qubits immediately after ideal encoding. Idling noise is omitted, giving an idealised long-coherence benchmark.

Decoding follows the sequence of BP variants in the telescoping decoder~\cite{BhardwajEtAl2026} (layered min-sum BP on the initialisation-basis restriction of the detector error model, Relay-BP~\cite{mullerImprovedBeliefPropagation2025b} with quorum voting on its GARI representation~\cite{MaanGarciaHerreroPaler2026}, then serial BP and Relay-BP with perturbed priors on the GARI representation of the full correlated model), with the final integer-programming stage replaced by VibeLSD~\cite{KoutsioumpasNoszkoSayginel2025} using a 32-member BP ensemble and order-zero localised statistics decoding. Routing depends only on decoder convergence and consistency checks. Incorrect predictions accepted at any stage count as logical failures.

For basis $b\in\{X,Z\}$, let $\widehat P_b=E_b/N_b$, where $E_b$ and $N_b$ are failed and sampled shots. We combine the two basis experiments through
\begin{equation}
\begin{aligned}
\widehat u &= 1-(1-\widehat P_X)(1-\widehat P_Z),\\
N_* &= \min(N_X,N_Z),\qquad
E_* = \operatorname{round}(N_*\widehat u),\\
\widetilde u &= E_*/N_* .
\end{aligned}
\end{equation}
The combined block logical error rate per syndrome extraction round is
\begin{equation}
r_{XZ}=1-\left[
\frac{1+\left(2(1-\widetilde u)^{1/(2q)}-1\right)^{1/R}}{2}
\right]^{2q},
\end{equation}
where $q=k$ for memory and $q=4$ for the two-CNOT experiment.
Nominal $1\sigma$ error bars are obtained by transforming the endpoints of the central $68.27\%$ Clopper--Pearson interval for the effective counts $(E_*,N_*)$.

\clearpage
\onecolumngrid
\section{Examples of parallel product surgery of CC codes}
\label{Appendix: examples}

In this appendix, we give some concrete examples of parallel surgery with a CC code as both the data code patch and the auxiliary code patch. The first three examples are for same-type product surgery. We then give two mixed-type product surgeries and a single logical $Y$ measurement using general attachments.
\begin{example}[$\code{24,8,3}$, maximally $4$ pairs of merges with overhead of $24$ qubits and $24$ checks]
  \label{Example: [24,8,3] Id merge}
  Consider the $\code{24,8,3}$ CC code given in Table~\ref{Table: cc_seed_matrices_8}. The logical operator bases are represented by
  \begin{equation}
    \label{eq: 8 logical basis}
    L=\left(
    \begin{array}{cccc|cccc}
      \chi &  &  &  &  &  &  & \\
      & \chi &  &  &  &  &  & \\
      &  & \chi &  &  &  &  & \\
      &  &  & \chi &  &  &  & \\
      &  &  &  & \chi &  &  & \\
      &  &  &  &  & \chi &  & \\
      &  &  &  &  &  & \chi & \\
      &  &  &  &  &  &  & \chi
    \end{array}\right),
  \end{equation}
  where $\chi = 1+x+x^2$. We perform $4$ pairs of joint logical measurements on all the logical qubits, i.e., merging $\bZ_i^\lef$ with $\bZ_i^\rig$ for $i=1,2,3,4$ by choosing $\QQQ^\mathrm{auxi} =\QQQ = \ccab$ and $\PPP = \lpAB$ where
  \begin{equation}
    H_a'=H_b' =I_2.
  \end{equation}
  The corresponding merged code with PCMs
  \begin{equation}
    \widetilde{H_X}=
    \begin{pmatrix}
      H_X & H_X'\\
      & H_X
    \end{pmatrix}\tand \widetilde{H_Z}=
    \begin{pmatrix}
      H_Z & \\
      H_Z' & H_Z
    \end{pmatrix},
  \end{equation}
  defines a $\code{\widetilde{n}=48, \widetilde{k}=4,\widetilde{r}=4,\widetilde{d}=3}$ merged code, where $\widetilde{n}, \widetilde{k},\widetilde{r}\tand \widetilde{d}$ are the number of physical qubits, number of protected data logical qubits, number of auxiliary gauge qubits and protected code distance, respectively. The bottom half
  of $\widetilde{H_Z}$ is given by
  \eqa{
    \begin{pmatrix}
      H_Z' & H_Z
    \end{pmatrix} =&\ \left(
    \begin{array}{cccc|cccc|cccc|cccc}
      1 &  &  &  & 1 &  &  &  & 1+x & 1+x^2 &  &  & 1+x^2 &  & 1+x^2 & \\
      & 1 &  &  &  & 1 &  &  & 1+x & 1+x &  &  & & 1+x^2 &  & 1+x^2\\
      &  & 1 &  &  &  & 1 &  & & & 1+x & 1+x^2 & 1+x & & 1+x^2 & \\
      &  &  & 1 &  &  &  & 1 & & & 1+x & 1+x & & 1+x & & 1+x^2
    \end{array}\right).\\
    &\quad \underbrace{\hspace{2.8cm}}_{\text{Data patch}}
    \quad \underbrace{\hspace{8.4cm}}_{\text{Auxiliary patch}}}
  Therefore, $\supp\lr{\bZ_{i}^\lef \otimes \bZ_i^\rig} = \chi
  \begin{pmatrix}
    H_Z' & H_Z
  \end{pmatrix}_{i,\bullet}$. For example, we have
  \eqa{
    \supp\lr{\bZ_2^\lef \otimes \bZ_2^\rig} &= \chi  \left(
    \begin{array}{cccc|cccc|cccc|cccc}
      0& 1 & 0 & 0 & 0 & 1 & 0 & 0 & 1+x & 1+x & 0 &0  & 0& 1+x^2 & 0 & 1+x^2
    \end{array}\right)\\
    &= \left(
    \begin{array}{cccc|cccc|cccc|cccc}
      0& \chi & 0 & 0 & 0 & \chi & 0 & 0 & 0 & 0 & 0 &0  & 0& 0& 0 & 0
    \end{array}\right)\\
    &= \left(
    \begin{array}{cccc|cccc|cccc|cccc}
      0& \chi & 0 & 0 & 0 & 0 & 0 & 0 & 0 & 0 & 0 &0  & 0& 0& 0 & 0
    \end{array}\right) + \left(
    \begin{array}{cccc|cccc|cccc|cccc}
      0& 0 & 0 & 0 & 0 & \chi & 0 & 0 & 0 & 0 & 0 &0  & 0& 0& 0 & 0
    \end{array}\right).
  }
  The binary action of this `multiplying by $\chi$' is simply taking the sum of the rows in the binary representation of the $i$\textsuperscript{th} row of $
  \begin{pmatrix}
    H_Z' & H_Z
  \end{pmatrix}$. Namely, for each merge $\bZ_{i}^\lef \otimes \bZ_i^\rig$ we measure $l=3$ stabilisers of the merged code, whose product is this pair of $Z$-type logical operators.
  Initialising the data patch in ${\bket{+}}^{\otimes 8}$ and measuring the four joint products $\bZ_i^\lef \otimes \bZ_i^\rig$ produces four logical Bell states~\cite{Game_of_surface_code}.

\end{example}

\begin{example}[$\code{24,8,3}$, $4$ single-qubit logical measurements with overhead of $24$ qubits and $24$ checks]
  \label{Example: [24,8,3] Id logical measurements}
  Consider the $\code{24,8,3}$ CC code given in Table~\ref{Table: cc_seed_matrices_8}. The logical operator bases are represented by $L =\diag_8(\chi)$ as in Eq.~\eqref{eq: 8 logical basis}.
  We perform parallel logical measurements on $4$ logical operators, i.e., measuring $\bZ_i^\lef$ for $i=1,2,3,4$ by choosing $\QQQ^\mathrm{auxi} =\QQQ = \ccab$ and $\PPP = \lpAB$ where
  \begin{equation}
    H_a' = 0\tand H_b' = I_2 .
  \end{equation}
  The corresponding merged code with PCMs
  \begin{equation}
    \widetilde{H_X}=
    \begin{pmatrix}
      H_X & H_X'\\
      & H_X
    \end{pmatrix}\tand \widetilde{H_Z}=
    \begin{pmatrix}
      H_Z & \\
      H_Z' & H_Z
    \end{pmatrix},
  \end{equation}
  defines a $\code{\widetilde{n}=48, \widetilde{k}=4,\widetilde{r}=4,\widetilde{d}=3}$ merged code. The bottom half of $\widetilde{H_Z}$ is given by
  \eqa{
    \begin{pmatrix}
      H_Z' & H_Z
    \end{pmatrix} =&\ \left(
    \begin{array}{cccc|cccc|cccc|cccc}
      1 &  &  &  & 0 &  &  &  & 1+x & 1+x^2 &  &  & 1+x^2 &  & 1+x^2 & \\
      & 1 &  &  &  & 0 &  &  & 1+x & 1+x &  &  & & 1+x^2 &  & 1+x^2\\
      &  & 1 &  &  &  & 0 &  & & & 1+x & 1+x^2 & 1+x & & 1+x^2 & \\
      &  &  & 1 &  &  &  & 0 & & & 1+x & 1+x & & 1+x & & 1+x^2
    \end{array}\right).\\
    &\quad \underbrace{\hspace{2.8cm}}_{\text{Data patch}}
    \quad \underbrace{\hspace{8.4cm}}_{\text{Auxiliary patch}}}
  Therefore, $\supp\lr{\bZ_{i}^\lef} = \chi
  \begin{pmatrix}
    H_Z' & H_Z
  \end{pmatrix}_{i,\bullet}$. For example,
  \eqa{
    \supp\lr{\bZ_2^\lef} &= \chi  \left(
    \begin{array}{cccc|cccc|cccc|cccc}
      0& 1 & 0 & 0 & 0 & 0 & 0 & 0 & 1+x & 1+x & 0 &0  & 0& 1+x^2 & 0 & 1+x^2
    \end{array}\right)\\
    &= \left(
    \begin{array}{cccc|cccc|cccc|cccc}
      0& \chi & 0 & 0 & 0 & 0 & 0 & 0 & 0 & 0 & 0 &0  & 0& 0& 0 & 0
    \end{array}\right).
  }
  As before, the binary action of this `multiplying by $\chi$' is simply taking the sum of the rows in the binary representation of the $i$\textsuperscript{th} row of $
  \begin{pmatrix}
    H_Z' & H_Z
  \end{pmatrix}$. Namely, for each logical operator $\bZ_{i}^\lef$ we measure $l=3$ stabilisers of the merged code, whose product is this $Z$-type logical operator.
\end{example}

\begin{example}[$\code{24,8,3}$, $2$ pairs of merges with overhead of $24$ qubits and $24$ checks]
  Consider the $\code{24,8,3}$ CC code given in Table~\ref{Table: cc_seed_matrices_8}. The logical operator bases are represented by $L =\diag_8(\chi)$. We perform two pairs of joint logical measurements on the logical operators on the right sector, i.e., measuring $\bZ_1^\rig \otimes \bZ_3^\rig$ and $\bZ_2^\rig \otimes \bZ_4^\rig$ in parallel by choosing $\QQQ^\mathrm{auxi} =\QQQ = \ccab$ and $\PPP = \lpAB$ where
  \begin{equation}
    H_a' =
    \begin{pmatrix}
      1 & 0\\
      1 & 0
    \end{pmatrix}\tand H_b' = 0 .
  \end{equation}
  The corresponding merged
  code with PCMs
  \begin{equation}
    \widetilde{H_X}=
    \begin{pmatrix}
      H_X & H_X'\\
      & H_X
    \end{pmatrix}\tand \widetilde{H_Z}=
    \begin{pmatrix}
      H_Z & \\
      H_Z' & H_Z
    \end{pmatrix},
  \end{equation}
  defines a $\code{\widetilde{n}=48, \widetilde{k}=6,\widetilde{r}=6,\widetilde{d}=3}$ merged code. The
  bottom half of $\widetilde{H_Z}$ is given by
  \eqa{
    \begin{pmatrix}
      H_Z' & H_Z
    \end{pmatrix} =&\ \left(
    \begin{array}{cccc|cccc|cccc|cccc}
      0 &  &  &  & 1 &  & 1 &  & 1+x & 1+x^2 &  &  & 1+x^2 &  & 1+x^2 & \\
      & 0 &  &  &  & 1 &  & 1 & 1+x & 1+x &  &  & & 1+x^2 &  & 1+x^2\\
      &  & 0 &  &  &  & 0 &  & & & 1+x & 1+x^2 & 1+x & & 1+x^2 & \\
      &  &  & 0 &  &  &  & 0 & & & 1+x & 1+x & & 1+x & & 1+x^2
    \end{array}\right),\\
    &\quad \underbrace{\hspace{2.8cm}}_{\text{Data patch}}
    \quad \underbrace{\hspace{8.4cm}}_{\text{Auxiliary patch}}}
  Here $H_Z'$ has rank $2$, so the round performs two independent logical PPMs and leaves $8-2=6$ protected data logical qubits. To measure the joint Pauli product,
  \begin{equation}
    \supp\lr{\bZ_{1}^\rig\otimes \bZ_{3}^\rig} = \chi
    \begin{pmatrix}
      H_Z' & H_Z
    \end{pmatrix}_{1,\bullet} \tand \supp\lr{\bZ_{2}^\rig\otimes \bZ_{4}^\rig} = \chi
    \begin{pmatrix}
      H_Z' & H_Z
    \end{pmatrix}_{2,\bullet}.
  \end{equation}
  In this example, we perform measurements of two pairs of joint logical operators with the same overhead as in previous cases ($24$ physical qubits and $24$ checks). This illustrates a key feature of the parallel surgery protocol: for a fixed data-code patch, the total resource cost of a single round of PPMs remains constant. The actual efficiency of the implementation then depends on how the algorithm is compiled so that more PPMs can be executed within that fixed overhead. Ideally, one performs as many independent PPMs per round as possible, up to $k/2$, as laid out in Theorem~\ref{Theorem: max merge}.
\end{example}

  \begin{example}[$\code{24,8,3}$, four parallel mixed-type products]
    \label{Example: [24,8,3] mixed XZ merges}
    Consider the same $\code{24,8,3}$ code and logical basis. We now measure the four products
    \[
      \bX_i^\lef\bZ_i^\rig,\qquad i=1,2,3,4,
    \]
    in one surgery round. The seed pair in Table~\ref{Table: cc_seed_matrices_8} admits the monomial matrices
    \[
      T_a=\begin{pmatrix}0&x^2\\1&0\end{pmatrix},\qquad
      T_b=\begin{pmatrix}0&1\\x&0\end{pmatrix},
    \]
    satisfying $T_aH_a=H_a^*$ and $T_bH_b=H_b^*$ in $R=\FF_2[x]/(x^3+1)$. Choose the binary product connection seeds and $X$ attachment as
    \[
      \begin{aligned}
        H_a'&=I_2,& H_b'&=0,\\
        H_X'&=(I_4\mid0),& H_Z'&=(0\mid I_4),\\
        g_1^*&=(I_4\mid0),& g_0&=(I_2\ox T_b\mid0).
      \end{aligned}
    \]
    Thus $f_0=H_X'$ and $f_1^* = H_Z'$, retaining the product connection constraints. The dressing equation follows from
    \[
      H_Zg_1=I_2\ox H_b^*
      =(I_2\ox T_b)(I_2\ox H_b)=g_0H_Z^*,
    \]
    while $g_1^*{H_Z'}^*=0$. The merged PCM is therefore
    \[
      \widetilde H=
      \left(\begin{array}{cc|cc}
        H_X&(I_4\mid0)&0&0\\
        0&(I_2\ox T_b\mid0)&H_Z&0\\
        (I_4\mid0)&0&(0\mid I_4)&H_Z\\
        0&H_X&0&0
      \end{array}\right),
    \]
    in column order $(X_{\mathrm{data}},X_{\mathrm{auxi}}\mid Z_{\mathrm{data}},Z_{\mathrm{auxi}})$. Let $G_{i,t}$ denote the three physical merge checks obtained by lifting attachment row $i$, with $t=0,1,2$. Multiplying their outcomes measures
    \[
      \prod_{t=0}^{2}G_{i,t}
      =\bX_i^\lef\bZ_i^\rig\otimes I_{\mathrm{auxi}},
    \]
    because $\chi H_Z=0$. The four attachment rows are independent, so $M=4$. Only these four products are measured; their separate $X$ and $Z$ factors are not measured.

    Each dressed original check has weight at most $9$, each merge check has weight $10$, and the retained auxiliary checks have weight $8$. The maximum check weight is therefore $10$. The merged code has parameters $\code{\widetilde n=48,\widetilde k=4,\widetilde r=4,\widetilde d=3}$, with the same overhead of $24$ auxiliary data qubits and $24$ additional checks as the ordinary examples.

    The same $T_a,T_b$ also give explicit sparse dressings for all nine uniform choices $\overline P_i^\lef\overline Q_i^\rig$. With the binary Pauli indicators of Section~\ref{Subsec: mixed Pauli surgery}, choose
    \[
      \begin{aligned}
        H_a'&=q_ZI_2,& H_b'&=p_ZI_2,\\
        g_1^*&=(p_XI_4\mid q_XI_4),&
        g_0&=\left(p_X(I_2\ox T_b)\ \middle|\ q_X(T_a\ox I_2)\right).
      \end{aligned}
    \]
    In each case the maximum check weight is $10$.
  \end{example}

  \begin{example}[$\code{24,8,3}$, different Pauli products in one parallel round]
    \label{Example: [24,8,3] mixed Pauli round}
    The Pauli types need not be the same on every pair. For the same code, consider the four disjoint products
    \[
      \bX_1^\lef\overline Y_1^\rig,\qquad
      \overline Y_2^\lef\bZ_2^\rig,\qquad
      \bX_3^\lef\bX_3^\rig,\qquad
      \bZ_4^\lef\bX_4^\rig.
    \]
    These include $X$, $Y$ and $Z$ components in one surgery round. Their $Z$ components have the required product form: choose
    \[
      H_a'=\begin{pmatrix}1&0\\0&0\end{pmatrix},\qquad
      H_b'=\begin{pmatrix}0&0\\0&1\end{pmatrix},
    \]
    giving
    \[
      \begin{aligned}
        H_X'&=\bigl(\mathrm{diag}(1,1,0,0)\mid\mathrm{diag}(0,1,0,1)\bigr),\\
        H_Z'&=\bigl(\mathrm{diag}(0,1,0,1)\mid\mathrm{diag}(1,1,0,0)\bigr).
      \end{aligned}
    \]
    The $X$ components are specified by
    \[
      g_1^*=\bigl(\mathrm{diag}(1,1,1,0)\mid\mathrm{diag}(1,0,1,1)\bigr).
    \]
    Thus $f_0=H_X'$ and $f_1^*=H_Z'$ still come from binary product connection seeds. A compatible dressing is
    \[
      g_0=\left(\begin{array}{cccc|cccc}
        1&0&0&0&1&0&0&0\\
        x&0&x^2&x&0&0&0&0\\
        0&0&0&0&x&0&x^2&0\\
        0&0&0&1&0&0&0&0
      \end{array}\right).
    \]
    All non-zero entries of $g_0$ are monomials. Direct multiplication gives
    \[
      H_Zg_1=g_0H_Z^*,\qquad
      g_1^*{H_Z'}^*+H_Z'g_1=0,
    \]
    so these matrices define the merged PCM in Eq.~\eqref{eq: mixed product merged checks}.

    The four rows of $(g_1^*\mid H_Z')$ are independent and specify exactly the four requested products. In particular, $\rank H_Z'=3$, but $\rank(g_1^*\mid H_Z')=4$: the third row measures a pure $XX$ product through its $X$ attachment. Multiplying the three lifted merge checks in each row cancels the auxiliary support. With physical $Y$ factors chosen as $+Y$ and $\overline Y=i\bX\bZ$, the four products are obtained with the known signs $(-1,-1,+1,+1)$, respectively, since $\prod_{t=0}^2Y_{i,t}=-\overline Y_i$ for $l=3$. The surgery therefore measures these four commuting products jointly, without measuring their individual Pauli factors.

    The four row groups of the merged PCM have maximum check weights $10$, $11$, $10$ and $8$, respectively. The maximum check weight is therefore $11$, with $48$ data and auxiliary qubits in total. The merged code has parameters $\code{\widetilde n=48,\widetilde k=4,\widetilde r=1,\widetilde d=3}$. This example shows that different Pauli types can share a product surgery round when their $Z$ components retain the product form and a sparse compatible dressing is available.
  \end{example}

  \begin{example}[$\code{24,8,3}$, one addressed logical $Y$ measurement]
    \label{Example: [24,8,3] single Y}
    We now measure only $\overline Y_1^\lef$, using the general attachments described in Section~\ref{Subsec: mixed Pauli surgery}. This example retains the auxiliary product complex, but allows $f_0,f_1$ beyond the product connection seeds. This distinction is necessary: for $n_a=n_b=2$, every non-zero product connection $H_Z'$ has rank at least $2$, whereas a single logical $Y$ measurement requires a rank-one $Z$ attachment.

    Let $D=\mathrm{diag}(1,0,0,0)$ and choose
    \[
      f_1^*=g_1^*=(D\mid0).
    \]
    Compatible sparse dressings over $R=\FF_2[x]/(x^3+1)$ are
    \[
      f_0=\left(\begin{array}{cccc|cccc}
        0&0&0&0&x^2&0&x&0\\
        0&0&0&0&0&0&0&0\\
        0&0&0&0&x^2&0&x&0\\
        0&0&0&0&0&0&0&0
      \end{array}\right),\qquad
      g_0=\left(\begin{array}{cccc|cccc}
        x^2&x^2&0&0&0&0&0&0\\
        x^2&x^2&0&0&0&0&0&0\\
        0&0&0&0&0&0&0&0\\
        0&0&0&0&0&0&0&0
      \end{array}\right).
    \]
    Direct multiplication verifies $H_Xf_1=f_0H_Z^*$, $H_Zg_1=g_0H_Z^*$ and $g_1^*f_1+f_1^*g_1=0$. Substituting into the general attachment form gives
    \[
      \widetilde H=
      \left(\begin{array}{cc|cc}
        H_X&f_0&0&0\\
        0&g_0&H_Z&0\\
        (D\mid0)&0&(D\mid0)&H_Z\\
        0&H_X&0&0
      \end{array}\right).
    \]
    Only the first attachment row is non-zero. Its three lifted merge checks contain $Y$ on the three physical qubits of the first left cluster. The other attachment rows vanish, so their lifted merge checks act only on the auxiliary patch. With the physical merge checks chosen to contain $+Y$, their product is
    \[
      \prod_{t=0}^{2}G_{1,t}
      =\left(\prod_{t=0}^{2}Y_{1,t}^\lef\right)\otimes I_{\mathrm{auxi}}
      =-\overline Y_1^\lef\otimes I_{\mathrm{auxi}},
    \]
    where $\overline Y_1^\lef=i\bX_1^\lef\bZ_1^\lef$ in the clustered logical basis. Thus, if the three check outcomes are $m_0,m_1,m_2\in\{+1,-1\}$, the logical outcome is $-m_0m_1m_2$. This known sign is included in the outcome processing.

    The attachment space has dimension $M=1$ and is generated by $\overline Y_1^\lef$. Applying the split and Pauli-frame correction of Section~\ref{Subsec: surgery procedure} therefore leaves the other seven data logical qubits unmeasured. The dressed original checks have maximum weight $10$, the merge checks have maximum weight $9$, and the retained auxiliary checks have weight $8$. With the auxiliary logical modes treated as gauge, the merged code has parameters $\code{\widetilde n=48,\widetilde k=7,\widetilde r=6,\widetilde d=3}$ and maximum check weight $10$.
  \end{example}

  For the six examples above, direct binary verification confirms the commutation relations and the stated numbers of protected data and auxiliary gauge qubits. Exhaustive enumeration finds no non-trivial logical Pauli of weight $1$ or $2$ in any of the listed merged codes, while a surviving data logical has weight $3$, establishing the quoted distances. The same check verifies distance $3$ for all nine uniform mixed-type pair choices.

  \section{Two parallel logical CNOT gates by product surgery}
  \label{Appendix: two parallel CNOTs}

  We give an explicit product-surgery construction of $\bcnot_{1\to2}\bcnot_{3\to4}$ for a CC code with $n_a=n_b=2$ and $k=8$. Logical qubits $1,2,3,4$ form the left data sector, while $5,6,7,8$ form the right auxiliary sector. Thus $\overline P_i=\overline P_i^{\lef}$ and $\overline P_{i+4}=\overline P_i^{\rig}$ for $i=1,\ldots,4$. These auxiliary logical qubits belong to $\QQQ$ and are distinct from the physical auxiliary patch used for surgery. We use the $X/Z$-exchanged form of the standard parity-measurement circuit~\cite[Fig.~3]{LitinskiVonOppen2018}: prepare the auxiliary logicals in $\bket{0}$, couple them to the targets through $XX$ measurements and to the controls through $ZZ$ measurements, and complete the gates by $X$ readout.

  Set
  \[
    T=\begin{pmatrix}0&1\\1&0\end{pmatrix},\qquad
    D=\begin{pmatrix}1&0\\0&0\end{pmatrix}.
  \]
  The following binary product connection seeds give the complete schedule. All PPMs in a row are measured in one surgery round. Preparation is counted separately; its outcomes are corrected in the Pauli frame so that logicals $5,6,7,8$ are in $\bket{0}^{\otimes4}$.
  \begin{center}
    \begingroup
      \renewcommand{\arraystretch}{1.25}
      \setlength{\tabcolsep}{8pt}
      \begin{tabular}{@{}lccc@{}}
        \toprule
        Stage & $H_a'$ & $H_b'$ & Logical PPMs\\
        \midrule
        Preparation & $I_2$ & $0$ & $\bZ_5,\ \bZ_6,\ \bZ_7,\ \bZ_8$\\
        Round 1: $XX$ merge & $I_2$ & $T$ & $\bX_1\bX_6,\ \bX_2\bX_5,\ \bX_3\bX_8,\ \bX_4\bX_7$\\
        Round 2: $ZZ$ merge & $I_2$ & $D$ & $\bZ_1\bZ_5,\ \bZ_6,\ \bZ_3\bZ_7,\ \bZ_8$\\
        Round 3: $X$ readout & $0$ & $D$ & $\bX_5,\ \bX_7$\\
        \bottomrule
      \end{tabular}
    \endgroup
  \end{center}
  Indeed, the product connection PCMs selecting the three surgery rounds are
  \[
    \begin{aligned}
      \text{Round 1:}\quad&H_X'=(I_4\mid I_2\ox T),\\
      \text{Round 2:}\quad&H_Z'=(I_2\ox D\mid I_4),\\
      \text{Round 3:}\quad&H_X'=(0\mid I_2\ox D).
    \end{aligned}
  \]
  Their non-zero rows give exactly the listed observables, with ranks $4$, $4$ and $2$, respectively. In particular, $I_2\ox D=\diag(1,0,1,0)$ selects only $\bX_5$ and $\bX_7$ in the last round. The physical merged PCMs follow from Eq.~\eqref{eq: product merged checks}, with $X$ and $Z$ exchanged for the $X$-type rounds. No mixed-type attachment or logical permutation is required.

  Let $a_1,\ldots,a_4$ and $b_1,\ldots,b_4$ be the outcome bits of the first and second surgery rounds, in the order listed above, and let $t_5,t_7$ be the final $X$-readout bits. An outcome bit $u$ denotes eigenvalue $(-1)^u$, and all sums below are modulo two. After the second round, logicals $6,8$ are individually projected onto $\bket{b_2}_6\bket{b_4}_8$ and are disentangled from the remaining logical qubits. Their final $X$ readouts are therefore unnecessary, although their $Z$ outcomes enter the Pauli-frame correction. Logicals $5,7$ remain in the two CNOT gadgets and are read out in the third round.

  After including the split corrections of Section~\ref{Subsec: surgery procedure}, update the data Pauli frame by
  \[
    C=\bX_1^{b_2}\bX_2^{b_1+b_2}
    \bX_3^{b_4}\bX_4^{b_3+b_4}
    \bZ_1^{a_2+t_5}\bZ_3^{a_4+t_7}.
  \]
  To verify the logical action, let $K$ be the unnormalised conditional map on the four data logicals, with the known auxiliary states factored out. Direct multiplication of the logical measurement projectors gives
  \[
    CK=2^{-5}e^{i\phi}\bcnot_{1\to2}\bcnot_{3\to4},
  \]
  where $\phi$ depends only on the outcomes. Every corrected branch is therefore proportional to the same pair of CNOTs, also for entangled inputs, and the complete outcome record reveals no unwanted data observable.

  All connection PCMs above have maximum row weight at most two. The merged checks consequently have maximum weight $\widetilde w\leq w+2$, giving $\widetilde w\leq10$ for the check-weight-eight CC codes. The protocol takes three surgery rounds after preparation and reuses the same physical auxiliary patch, with $n$ auxiliary data qubits and $n$ additional check auxiliaries per round.

\clearpage
\begin{landscape}
  \ifXeTeX\fi
  \section{Representative CC codes}
  \label{Appendix: code examples}

  Tables~\ref{Table: cc_seed_matrices_8} and~\ref{Table: cc_seed_matrices_18} fix the thirteen seed pairs introduced in Table~\ref{Table: example code para}. We write $\code{n,k,d}\,(l)$ for the code parameters and lift. For all eight $k=8$ and five $k=18$ codes, $d_X=d_Z=d$ is established by excluding all non-trivial logical operators below weight $d$ and verifying explicit logical operators of weight $d$.

  \vspace{10pt}
  \noindent
  \begin{minipage}{\linewidth}
    \centering
    \captionsetup{hypcap=false}
    \captionof{table}{\textbf{Representative seed matrices for the $k=8$ CC codes over $R=\FF_2[x]/(x^l+1)$.}}
    \label{Table: cc_seed_matrices_8}
    \renewcommand{\arraystretch}{1.15}
    \setlength{\tabcolsep}{2.5pt}
    \setlength{\arraycolsep}{3pt}
    \begin{tabular*}{\linewidth}{@{\extracolsep{\fill}}c*{8}{c}@{}}
      \toprule
      CC code & $\code{24,8,3}\ (3)$ & $\code{40,8,5}\ (5)$ & $\code{56,8,7}\ (7)$ & $\code{72,8,9}\ (9)$ & $\code{88,8,10}\ (11)$ & $\code{104,8,11}\ (13)$ & $\code{120,8,13}\ (15)$ & $\code{136,8,15}\ (17)$\\
      \midrule
      $H_a$ & $\begin{pmatrix}1+x & 1+x^{2}\\1+x & 1+x\end{pmatrix}$ & $\begin{pmatrix}x+x^{3} & x+x^{4}\\x+x^{4} & x+x^{4}\end{pmatrix}$ & $\begin{pmatrix}x+x^{3} & x^{2}+x^{5}\\1+x^{3} & 1+x^{2}\end{pmatrix}$ & $\begin{pmatrix}x^{4}+x^{6} & x^{4}+x^{5}\\x^{7}+x^{8} & x^{4}+x^{6}\end{pmatrix}$ & $\begin{pmatrix}x+x^{10} & x^{2}+x^{9}\\x^{2}+x^{9} & x+x^{10}\end{pmatrix}$ & $\begin{pmatrix}x+x^{5} & x^{2}+x^{5}\\x+x^{4} & x^{9}+x^{10}\end{pmatrix}$ & $\begin{pmatrix}x^{2}+x^{13} & x^{6}+x^{9}\\x^{6}+x^{9} & x^{6}+x^{7}\end{pmatrix}$ & $\begin{pmatrix}x^{2}+x^{15} & x^{4}+x^{13}\\x^{4}+x^{13} & x^{4}+x^{8}\end{pmatrix}$\\\addlinespace[5pt]
      $H_b$ & $\begin{pmatrix}1+x^{2} & 1+x^{2}\\1+x & 1+x^{2}\end{pmatrix}$ & $\begin{pmatrix}1+x & x^{3}+x^{4}\\1+x^{4} & 1+x^{4}\end{pmatrix}$ & $\begin{pmatrix}x^{2}+x^{6} & x^{4}+x^{5}\\1+x^{6} & 1+x^{4}\end{pmatrix}$ & $\begin{pmatrix}x+x^{5} & x^{2}+x^{8}\\x^{2}+x^{5} & 1+x^{5}\end{pmatrix}$ & $\begin{pmatrix}x^{2}+x^{9} & x^{5}+x^{6}\\x^{5}+x^{6} & x^{4}+x^{7}\end{pmatrix}$ & $\begin{pmatrix}x^{2}+x^{8} & 1+x^{8}\\x+x^{9} & x+x^{3}\end{pmatrix}$ & $\begin{pmatrix}x+x^{14} & x^{4}+x^{11}\\x^{4}+x^{11} & x^{12}+x^{14}\end{pmatrix}$ & $\begin{pmatrix}x^{8}+x^{9} & x+x^{16}\\x+x^{16} & x^{14}+x^{15}\end{pmatrix}$\\
      \bottomrule
    \end{tabular*}
  \end{minipage}

  \par\vspace{24pt}

  \noindent
  \begin{minipage}{\linewidth}
    \centering
    \captionsetup{hypcap=false}
    \captionof{table}{\textbf{Representative seed matrices for the $k=18$ CC codes over $R=\FF_2[x]/(x^l+1)$.}}
    \label{Table: cc_seed_matrices_18}
    \renewcommand{\arraystretch}{1.15}
    \setlength{\tabcolsep}{2.5pt}
    \setlength{\arraycolsep}{3pt}
    \begin{tabular*}{\linewidth}{@{\extracolsep{\fill}}c*{5}{c}@{}}
      \toprule
      CC code & $\code{54,18,3}\ (3)$ & $\code{90,18,5}\ (5)$ & $\code{126,18,7}\ (7)$ & $\code{162,18,9}\ (9)$ & $\code{198,18,10}\ (11)$\\
      \midrule
      $H_a$ & $\begin{pmatrix}0 & x+x^{2} & 1+x\\x+x^{2} & 0 & x+x^{2}\\1+x^{2} & x+x^{2} & 0\end{pmatrix}$ & $\begin{pmatrix}0 & x+x^{4} & 1+x^{3}\\x+x^{4} & 0 & x+x^{4}\\1+x^{2} & x+x^{4} & 0\end{pmatrix}$ & $\begin{pmatrix}0 & x+x^{6} & 1+x\\x+x^{6} & 0 & x^{3}+x^{4}\\1+x^{6} & x^{3}+x^{4} & 0\end{pmatrix}$ & $\begin{pmatrix}x+x^{8} & x^{4}+x^{5} & 0\\x^{4}+x^{5} & 0 & x+x^{8}\\0 & x+x^{8} & x^{4}+x^{5}\end{pmatrix}$ & $\begin{pmatrix}x+x^{10} & x^{5}+x^{6} & 0\\x^{5}+x^{6} & 0 & x+x^{10}\\0 & x+x^{10} & x^{5}+x^{6}\end{pmatrix}$\\\addlinespace[5pt]
      $H_b$ & $\begin{pmatrix}x+x^{2} & x+x^{2} & 0\\x+x^{2} & 0 & x+x^{2}\\0 & x+x^{2} & x+x^{2}\end{pmatrix}$ & $\begin{pmatrix}x+x^{4} & x^{2}+x^{3} & 0\\x^{2}+x^{3} & 0 & x^{2}+x^{3}\\0 & x^{2}+x^{3} & x^{2}+x^{3}\end{pmatrix}$ & $\begin{pmatrix}x+x^{6} & x^{2}+x^{5} & 0\\x^{2}+x^{5} & 0 & x+x^{6}\\0 & x+x^{6} & x^{2}+x^{5}\end{pmatrix}$ & $\begin{pmatrix}x^{2}+x^{7} & x^{3}+x^{6} & 0\\x^{3}+x^{6} & 0 & x^{2}+x^{7}\\0 & x^{2}+x^{7} & x^{3}+x^{6}\end{pmatrix}$ & $\begin{pmatrix}x^{2}+x^{9} & x^{2}+x^{9} & 0\\x^{2}+x^{9} & 0 & x^{4}+x^{7}\\0 & x^{4}+x^{7} & x^{3}+x^{8}\end{pmatrix}$\\
      \bottomrule
    \end{tabular*}
  \end{minipage}
\end{landscape}

\end{document}